\documentclass[a4paper,11pt,reqno,noindent]{amsart}
\usepackage[centertags]{amsmath}
\usepackage{amsfonts,amssymb,amsthm,dsfont,cases,amscd,esint,enumerate, stmaryrd}
\usepackage[T1]{fontenc}
\usepackage[english]{babel}
\usepackage[applemac]{inputenc}
\usepackage{newlfont}
\usepackage{color}
\usepackage[body={17cm,21.5cm}, top = 1in, bottom = 1.5in, centering]{geometry} 
\usepackage{fancyhdr}
\usepackage{enumerate}

\theoremstyle{plain}
\newtheorem{theor}{Theorem}[section]

\newtheorem{prop}[theor]{Proposition}
\newtheorem{lem}[theor]{Lemma}
\newtheorem{defin}[theor]{Definition}
\newtheorem{rem}[theor]{Remark}

\numberwithin{equation}{section}

\newcommand{\ind}{\mathds1}

\renewcommand{\>}{\rangle}
\newcommand{\<}{\langle}
\newcommand{\Span}{\operatorname{span}}
\newcommand{\conv}{\operatorname{conv}}
\newcommand{\dist}{\operatorname{dist}}
\newcommand{\aff}{\operatorname{aff}}

\newcommand{\V}{{\hat{ \mathcal{V}}}}
\newcommand{\Vc}{{\mathcal{V}}}

\newcommand{\Dd}{\mathbb D}

\newcommand{\N}{\mathbb N}

\newcommand{\R}{\mathbb R}
\newcommand{\Sp}{\mathbb S}
\newcommand{\Z}{\mathbb Z}

\newcommand{\T}{\mathbb T}
\newcommand{\e}{\varepsilon}

\newcommand{\Pc}{\mathcal{P}}
\newcommand{\Lc}{\mathcal{L}}

\newcommand{\Hc}{\mathcal H}

\newcommand{\Gc}{\mathcal G}

\newcommand{\Ic}{\mathcal I}

\newcommand{\calH}{\mathcal H}

\newcommand{\loc}{{\operatorname{loc}}}
\newcommand{\Id}{\operatorname{Id}}

\newcommand{\Ld}{L}

\newcommand{\Div}{{\operatorname{div}}}

\newcommand{\3}{\operatorname{|\hspace{-0.4mm}|\hspace{-0.4mm}|}}

\usepackage[colorlinks,citecolor=black,urlcolor=black]{hyperref}

\usepackage{tikz}
\usetikzlibrary{fit}
\usetikzlibrary{shapes.geometric}
\usetikzlibrary{calc}
\usetikzlibrary{decorations.pathmorphing}

\newcommand{\BOX}{{\tiny\begin{tikzpicture}[baseline={([yshift=-1ex]current bounding box.center)},scale=0.8]
\begin{scope}[every node/.style={circle,draw,fill=white,inner sep=0pt,minimum size=3pt}]
\node (1) at (0,0) {};
\node (2) at (0.5,0) {};
\end{scope}
\draw (1) -- (2);
\draw (1) -- (0,0.3) -- (0.5,0.3) -- (2);
\end{tikzpicture}}}

\newcommand{\BOXd}{{\tiny\begin{tikzpicture}[baseline={([yshift=-1ex]current bounding box.center)},scale=0.8]
\begin{scope}[every node/.style={circle,draw,fill=white,inner sep=0pt,minimum size=3pt}]
\node (1) at (0,0) {};
\node (2) at (0.5,0) {};
\end{scope}
\draw (1) -- (2);
\draw (1) -- (0,0.25) -- (0.5,0.25) -- (2);
\draw (0,0.5) -- (0.5,0.5);
\end{tikzpicture}}}

\newcommand{\BOXu}{{\tiny\begin{tikzpicture}[baseline={([yshift=-1ex]current bounding box.center)},scale=0.8]
\begin{scope}[every node/.style={circle,draw,fill=white,inner sep=0pt,minimum size=3pt}]
\node (1) at (0,0.5) {};
\node (2) at (0.5,0.5) {};
\end{scope}
\draw (1) -- (2);
\draw (1) -- (0,0.25) -- (0.5,0.25) -- (2);
\draw (0,0) -- (0.5,0);
\end{tikzpicture}}}

\newcommand{\WIG}{{\tiny\begin{tikzpicture}[baseline={([yshift=-.8ex]current bounding box.center)},scale=0.8]
\draw[decorate,decoration={snake,amplitude=1pt,segment length=3pt}] (0,0) -- (0.8,0);
\draw[decorate,decoration={snake,amplitude=1pt,segment length=3pt}] (0,0.3) -- (0.8,0.3);
\end{tikzpicture}}}

\newcommand{\WIGdot}{{\tiny\begin{tikzpicture}[baseline={([yshift=-.8ex]current bounding box.center)},scale=0.8]
\draw[decorate,decoration={snake,amplitude=1pt,segment length=3pt}] (0,0) -- (0.8,0);
\draw[decorate,decoration={snake,amplitude=1pt,segment length=3pt}] (0,0.45) -- (0.8,0.45);
\draw[dotted] (0.4,0.1) -- (0.4,0.4);
\end{tikzpicture}}}

\newcommand{\WIGdoT}[1]{{\tiny\begin{tikzpicture}[baseline={([yshift=-2.2ex]current bounding box.center)},scale=0.8]
\draw[decorate,decoration={snake,amplitude=1pt,segment length=3pt}] (0,0) -- (0.8,0);
\draw[decorate,decoration={snake,amplitude=1pt,segment length=3pt}] (0,0.45) -- (0.8,0.45);
\draw[dotted] (0.4,0.1) -- (0.4,0.4);
\node at (0.4,0.7) {$#1$};
\end{tikzpicture}}}

\newcommand{\WIGbox}{{\tiny\begin{tikzpicture}[baseline={([yshift=-1ex]current bounding box.center)},scale=0.8]
\begin{scope}[every node/.style={circle,draw,fill=white,inner sep=0pt,minimum size=3pt}]
\node (1) at (0,0) {};
\node (2) at (0.5,0) {};
\end{scope}
\draw[decorate,decoration={snake,amplitude=1pt,segment length=3pt}] (1) -- (2);
\draw[decorate,decoration={snake,amplitude=1pt,segment length=3pt}] (0,0.3) -- (0.5,0.3);
\draw (1) -- (0,0.3);
\draw (0.5,0.3) -- (2);
\end{tikzpicture}}}

\newcommand{\WIGBOX}{{\tiny\begin{tikzpicture}[baseline={([yshift=-1ex]current bounding box.center)},scale=0.8]
\begin{scope}[every node/.style={circle,draw,fill=white,inner sep=0pt,minimum size=3pt}]
\node (1) at (0.5,0.5) {};
\node (2) at (1,0.5) {};
\end{scope}
\draw [decorate,decoration={snake,amplitude=1pt,segment length=3pt}] (1) -- (2);
\draw [decorate,decoration={snake,amplitude=1pt,segment length=3pt}](0.5,0.75) -- (1,0.75);
\draw (1) -- (0.5,0.75);
\draw (2) -- (1,0.75);
\draw [decorate,decoration={snake,amplitude=1pt,segment length=3pt}](0.5,0.05) -- (1,0.05);
\draw[dotted] (0.75,0.15) -- (0.75,0.45);
\end{tikzpicture}}}

\newcommand{\BOXdot}{{\tiny\begin{tikzpicture}[baseline={([yshift=-1ex]current bounding box.center)},scale=0.8]
\begin{scope}[every node/.style={circle,draw,fill=white,inner sep=0pt,minimum size=3pt}]
\node (1) at (0.5,0.5) {};
\node (2) at (1,0.5) {};
\end{scope}
\draw (1) -- (2);
\draw (0.5,0.75) -- (1,0.75);
\draw (1) -- (0.5,0.75);
\draw (2) -- (1,0.75);
\draw (0.5,0.05) -- (1,0.05);
\draw[dotted] (0.75,0.15) -- (0.75,0.45);
\end{tikzpicture}}}

\title[Lenard-Balescu thermalization: rigorous derivation from a toy model]{Lenard-Balescu thermalization:\\rigorous derivation from a toy model}
\author[M. Duerinckx]{Mitia Duerinckx}
\address[Mitia Duerinckx]{Universit\'e Libre de Bruxelles, D\'epartement de Math\'ematique, Brussels, Belgium}
\email{mitia.duerinckx@ulb.be}
\author[C. Le Bihan]{Corentin Le Bihan}
\address[Corentin Le Bihan]{Universit\'e Libre de Bruxelles, D\'epartement de Math\'ematique, Brussels, Belgium}
\email{corentin.le.bihan@ulb.be}

\begin{document}

\begin{abstract}
We study the long-time dynamics of a tagged particle coupled to a background of $N$ other particles, all interacting through long-range pairwise forces in the mean-field scaling, with the background initially at thermal equilibrium. Starting from the $N$-particle BBGKY hierarchy, we introduce a simplified (truncated) hierarchical model and show, in sufficiently large spatial dimension, that the tagged-particle density converges, on timescales $t\sim N$, to the solution of a linear Fokker-Planck equation, viewed as the linearization of the Landau equation. This provides, in a simplified setting, a rigorous derivation of the slow thermalization predicted by Lenard-Balescu theory.

Our approach relies on a rigorous Dyson expansion in terms of Feynman diagrams and on a novel renormalization scheme that removes leading recollisions. The main technical challenge is to control the effect of phase-space filamentation within the diagrams, which we achieve by combining phase mixing and hypoelliptic regularity. Although restricted to a simplified model, our analysis offers new insight into Lenard-Balescu thermalization: notably, the renormalization appears to transform free propagators into hypoelliptic ones, providing a key mechanism that compensates for filamentation.
\end{abstract}

\maketitle

\setcounter{tocdepth}{1}
\tableofcontents
\allowdisplaybreaks

\section{Introduction}

\subsection{General overview}
Consider the dynamics of a tagged particle (labeled `$0$') in a system of $N+1$ interacting particles in the $d$-dimensional box $\T^d=[-\frac12,\frac12]^d$ with periodic boundary conditions, as described by Newton's equations
\begin{equation*}
\left\{\begin{array}{l}
\tfrac{d}{dt}X_{j}=V_j,\\[2mm]
\tfrac{d}{dt}V_{j}=-\frac1N\sum_{0\le l\le N}\nabla\Vc(X_j-X_l),\qquad0\le j\le N,\\[2mm]
(X_{j},V_j)|_{t=0}=(X_j^\circ,V_j^\circ),
\end{array}\right.
\end{equation*}
where $\{Z_j:=(X_j,V_j)\}_{0\le j\le N}$ is the set of positions and velocities of the particles in the phase space $\T^d\times\R^d$, where $\Vc:\T^d\to\R$ is a (long-range) interaction potential, and where the mean-field scaling is considered. For simplicity, we focus throughout this work on the case of a smooth potential~$\Vc\in C^\infty_c$.
In terms of a probability density $F_N$ on the $N$-particle phase space $(\T^d\times\R^d)^N$, this system of ODEs for trajectories leads formally to the Liouville equation
\begin{equation}\label{eq:Liouville}
\left\{\begin{array}{l}
\partial_tF_N+\sum_{0\le j\le N}v_j\cdot\nabla_{x_j}F_N\,=\,\frac1N\sum_{0\le j, l\le N}\nabla\Vc(x_j-x_l)\cdot\nabla_{v_j}F_N,\\[1mm]
F_N|_{t=0}\,=\,F_N^\circ.
\end{array}\right.
\end{equation}
For this system, we aim to study one of the predictions of the Lenard-Balescu theory, that is, the slow thermalization of the tagged particle `$0$' when the background particles are initially at thermal equilibrium. For simplicity, we assume that the tagged particle has initially spatially-homogeneous distribution: this ensures that the mean-field force on the tagged particle vanishes, so that thermalization becomes the leading effect. In other words, we choose initially
\[F_N^\circ(z_0,\ldots,z_N)\,=\,f^\circ(v_0)M_{N}(z_1,\ldots,z_N),\qquad z_j=(x_j,v_j),\]
where $f^\circ\in \Pc(\R^d)$ is the initial velocity density of the tagged particle
and where $M_{N}$ stands for the Gibbs equilibrium
\[M_{N}(z_1,\ldots,z_N)\,:=\,Z_{N}^{-1}\exp\bigg(-\frac\beta2\sum_{1\le j\le N}|v_j|^2-\frac\beta{2N}\sum_{1\le j,l\le N\atop j\ne l}\Vc(x_j-x_l)\bigg),\]
for some fixed inverse temperature $\beta\in(0,\infty)$ and normalization factor $Z_{N}$.
We emphasize that background particles are exchangeable
and that the system is invariant under spatial translations.
In this setting, due to interactions with the background, the Lenard-Balescu theory predicts that the tagged particle should slowly thermalize on timescales $t\gg N$ and progressively acquire a Maxwellian velocity distribution as the background itself,
\[M(v)\,:=\,(\tfrac\beta{2\pi})^\frac d2e^{-\frac\beta2|v|^2}.\]
More precisely, focussing on the velocity distribution of the tagged particle
\[f_{N,0}(t,v)\,:=\,\int_{\T^d\times(\T^d\times\R^d)^N}F_{N}(t,x,v,z_1,\ldots,z_N)\,dxdz_1\ldots dz_N,\]
it is predicted in~\cite{Piasecki-81,PS-87,DSR-21} (see also~\cite{RR-60,TGM-64} for the corresponding test particle problem without back-reaction) that the time-rescaled density $f_{N,0}(N\tau,v)$ should converge as $N\uparrow\infty$ to the solution~$f(\tau,v)$ of the linearized Lenard-Balescu equation at Maxwellian equilibrium (without loss term), which takes form of the following Fokker-Planck equation,
\begin{equation}\label{eq:FP-exp}
\left\{\begin{array}{l}
\partial_\tau f=\Div_v(A (\nabla+\beta v)f),\\
f|_{\tau=0}=f^\circ,
\end{array}\right.
\end{equation}
with diffusion tensor field $A$ given by the (periodic) Lenard-Balescu formula
\begin{eqnarray}
A(v)&:=&\int_{\R^d}B(v,v-v_*)M(v_*)\,dv_*,\label{eq:coeff-LB}\\
B(v,w)&:=&\sum_{k\in2\pi\Z^d}(k\otimes k)\,\pi\V(k)^2\frac{\delta(k\cdot w)}{|\e(k,k\cdot v)|^2},\nonumber\\
\e(k,k\cdot v)&:=&1+\V(k)\int_{\R^d}\frac{k\cdot \nabla M(v_*)}{k\cdot(v-v_*)-i0}\,dv_*.\nonumber
\end{eqnarray}
In this formula, the dispersion function $\e$ modulating the collision kernel $B$ accounts for collective screening effects at equilibrium.
This Fokker-Planck equation~\eqref{eq:FP-exp} quantifies the thermalization
\[f(\tau,v)\to M(v)\qquad\text{as $\tau=N^{-1}t\uparrow\infty$.}\]
A rigorous justification is still beyond reach at the moment, and we refer to~\cite{VW-18,Winter-21,DSR-21,MD-21} for some partial results on the topic. To date, the best result is the consistency obtained in~\cite{DSR-21}, which only holds at best for relatively short times $\tau=N^{-1}t\ll N^{-3/4}$, thus missing the thermalization timescale $\tau\sim1$. We emphasize that the difficulty includes understanding the emergence of irreversibility, which here would only occur on long times as a fluctuation around the (trivial) mean-field behavior. In the present work, we introduce a simplified, hierarchical toy model for which a full derivation can be achieved, which will shed some new light on the problem.

\subsection{BBGKY approach}
To motivate our simplified model, we start by briefly recalling the standard BBGKY framework for addressing the above thermalization problem.
Intuitively, the slow relaxation of the tagged particle arises from the nontrivial correlations it develops with the background. To capture these effects, we introduce a suitable notion of correlation functions.
For $0\le m\le N$, we first define the joint density of the tagged particle with $m$ background particles as the following marginal of the $N$-particle density $F_N$,
\[F_{N,m}(z_0,\ldots,z_m)\,:=\,\int_{(\T^d\times\R^d)^{N-m}}F_N(z_0,\ldots ,z_N)\,dz_{m+1}\ldots dz_N,\]
which is symmetric in its last $m$ variables $\{z_j\}_{1\le j\le m}$. Next, for $0\le m\le N$, we define the $m$-th order correlation function (or cumulant) by
\begin{equation}\label{eq:def-CNm}
G_{N,m}(z_0,\ldots,z_m)\,:=\,\sum_{j=0}^m(-1)^{m-j}\sum_{\sigma\in \mathcal S_j^m}\frac{F_{N,j}}{M^{\otimes j+1}}(z_0,z_\sigma),
\end{equation}
where $\mathcal S_j^m$ stands for the set of all subsets of $\{1,\ldots,m\}$ with $j$ elements and where we use the short-hand notation $z_\sigma=(z_{i_1},\ldots,z_{i_j})$ for $\sigma=\{i_1,\ldots,i_j\}$. This definition ensures the orthogonality property 
\begin{equation}\label{eq:orthogonality0}
\int_{\T^d\times\R^d}G_{N,m}(z_0,\ldots,z_m)\,M(z_j)\,dz_j=0,\qquad\text{for all $1\le j\le m$},
\end{equation}
and also ensures that marginals can be recovered by means of a (linear) cluster expansion,
\begin{equation}\label{eq:cluster}
F_{N,m}(z_0,\ldots,z_m)\,=\,M^{\otimes m+1}(z_0,\ldots,z_m)\sum_{j=0}^m\sum_{\sigma\in\mathcal S_j^m}G_{N,j}(z_0,z_\sigma),\qquad\text{for all $0\le m\le N$}.
\end{equation}
By orthogonality~\eqref{eq:orthogonality0}, we deduce that correlation functions satisfy
\begin{equation}\label{eq:unitarity0}
\sum_{m=0}^N\binom{N}{m}\int_{(\T^d\times\R^d)^{m+1}}|G_{N,m}|^2M^{\otimes m+1}\,=\,\int_{(\T^d\times\R^d)^{N+1}}\frac{|F_N|^2}{M^{\otimes N+1}}.
\end{equation}
The right-hand side in this identity would be a conserved quantity if $M^{\otimes N+1}$ were replaced by the Gibbs ensemble~$M_{N+1}$. Still, it can be checked to be uniformly controlled in time: as shown in~\cite[Lemma~2.2]{DSR-21}, in the spirit of~\cite{BGSR-17}, we can deduce for all $0\le m\le N$ and $t\ge0$,
\begin{equation}\label{eq:apriori-CNm}
\Big(\int_{(\T^d\times\R^d)^{m+1}}|G_{N,m}(t)|^2M^{\otimes m+1}\Big)^\frac12\,\lesssim\,\binom{N}{m}^{-\frac12}\,\lesssim\,m^{\frac m2}N^{-\frac m2}.
\end{equation}
This can be viewed as some form of chaos estimates, showing that higher-order correlations are tinier in the limit $N\uparrow\infty$. Yet, the scaling in $N$ is not expected to be optimal: in particular, $G_{N,1}$ is expected to be only $O(N^{-1})$ instead of $O(N^{-1/2})$, cf.~\cite{DSR-21}.

To get more precise estimates on correlations, we need to go back to the Liouville equation~\eqref{eq:Liouville} for~$F_N$. Inserting the cluster expansion~\eqref{eq:cluster}, we check that correlations satisfy a hierarchy of equations of the following form, for $0\le m\le N$,
\begin{equation}\label{eq:hierarchy-true}
\partial_tG_{N,m}+iL_{N,m}G_{N,m}=iS^+_{N,m}G_{N,m+1}+\tfrac1N\Big(iS^\circ_{N,m}G_{N,m}+iS^-_{N,m}G_{N,m-1}+iS^=_{N,m}G_{N,m-2}\Big),
\end{equation}
where for notational convenience we set $G_{N,m}\equiv0$ for $m>N$ or $m<0$, where the operators~$S^\square_{N,m}$
are viewed as some `creation/annihilation' operators on the space of correlations,
\begin{equation*}
\begin{array}{rllll}
S_{N,m}^+&:&\Ld^2(M)^{\otimes m+1}&\to&\Ld^2(M)^{\otimes m}\\
S_{N,m}^\circ&:&\Ld^2(M)^{\otimes m}&\to&\Ld^2(M)^{\otimes m}\\
S_{N,m}^-&:&\Ld^2(M)^{\otimes m-1}&\to&\Ld^2(M)^{\otimes m},\\
S_{N,m}^=&:&\Ld^2(M)^{\otimes m-2}&\to&\Ld^2(M)^{\otimes m},
\end{array}
\end{equation*}
and where~$iL_{N,m}$ stands for the $m$-particle linearized Vlasov operator
\begin{multline}\label{eq:lin-mf0}
iL_{N,m}G_{N,m}\,=\,\sum_{j=0}^mv_j\cdot\nabla_{x_j}G_{N,m}\\[-3.5mm]
+\tfrac{N+1-m}N\sum_{j=1}^m\beta v_j\cdot\int_{\T^d\times\R^d}\nabla \Vc(x_j-x_{*})G_{N,m}(z_{[0,m]\setminus \{j\}},z_*)\,M(v_{*})\,dz_{*}.
\end{multline}
We refer to~\cite[Lemma~2.4]{DSR-21} for a detailed formulation of this hierarchy.

If the a priori correlation estimates~\eqref{eq:apriori-CNm} were known to hold in a stronger, smooth topology, a direct analysis of the above hierarchy would readily yield the expected kinetic limit $f_{N,0}(N\tau,v)\to f(\tau,v)$; see~\cite{DSR-21}.
However, due to filamentation effects in phase space, correlations become increasingly oscillatory over time and are only controlled a priori in $L^2$. This leads to possible resonances that may, in principle, prevent convergence to the kinetic limit. Consequently, in~\cite{DSR-21}, we only managed to establish a partial consistency result, valid on some intermediate timescale~$t\sim N^r$ with $r<\frac1{18}$ (which could be improved at best to~$r<\frac14$).

This motivates a more refined analysis of the hierarchy~\eqref{eq:hierarchy-true}, with the goal of tracking the oscillatory structure of correlations in greater detail and demonstrating that resonances cannot occur. An important observation --- crucial to the present work but not exploited in earlier studies --- is that the operator~$S_{N,m}^+$ involves a velocity average: hence, while filamentation is a priori viewed as a problem, phase mixing may in fact be leveraged beneficially in some terms.

\subsection{A simplified hierarchy}\label{sec:simpl}
As a first step toward analyzing the exact hierarchy~\eqref{eq:hierarchy-true}, we introduce a simplified setting that preserves the essential features of the original system, but for which we are able to rigorously justify the expected thermalization.
We start with the following two convenient simplifications, which do not affect the physical relevance of the model:
\begin{enumerate}[(H1)]
\item\label{H1} To avoid specific resonance issues on the torus, we consider the problem on the whole space~$\R^d$ instead of the periodic box $\T^d$. Physically, this amounts to considering a system of $NL^d$ background particles in a rescaled box $[-\frac L2,\frac L2]^d$ with periodic boundary conditions, and taking the large-box limit $L\uparrow\infty$ as $N\uparrow\infty$. This does not change much in the system physically, but has a very important simplifying effect as Fourier variables become continuous.
\smallskip\item\label{H2} We include a small $O(\frac1N)$ diffusion in velocity in the particle system. As this only acts on the slow relaxation timescale $t\sim N$, we do not expect it to change much to the problem, but it is useful to simplify the analysis.
\end{enumerate}
We introduce two additional simplifications, which we believe are not essential to our arguments but are highly convenient for streamlining the computations:
\begin{enumerate}[(H3)]
\item[(H3)]\label{H1} The linearized mean-field operator $iL_{N,m}$, cf.~\eqref{eq:lin-mf0}, is replaced by pure transport: this allows to avoid many technicalities and in particular to perform direct computations in Fourier space instead of appealing to linear Landau damping. Physically, this amounts to neglecting collective screening effects.
\smallskip\item[(H4)]\label{H2} In the exact hierarchy~\eqref{eq:hierarchy-true}, the creation and annihilation operators $S_{N,m}^\pm$ are not exact adjoints: they only become approximately so in the limit $N\uparrow\infty$. For finite $N$, the additional operators~$S_{N,m}^\circ$ and~$S_{N,m}^=$ correct this lack of adjointness, so as to ensure the approximate unitarity~\eqref{eq:unitarity0}. This reflects the fact that the Gibbs equilibrium $M_{N+1}$ differs from the tensorized mean-field ensemble~$M^{\otimes N+1}$ around which correlations are defined, cf.~\eqref{eq:def-CNm}. To simplify the structure, we replace $S_{N,m}^{\pm}$ by operators $S_m^\pm$ that are truly adjoint and are independent of $N$. The additional operators $S_{N,m}^\circ$ and~$S_{N,m}^=$ are then set to $0$. This leads to a hierarchy with a neater unitarity structure, in particular making~\eqref{eq:apriori-CNm} trivial.
\end{enumerate}
Finally, we introduce a last simplification, which plays a crucial role in our analysis and constitutes the main restriction of the present work:
\begin{enumerate}[(H5)]
\smallskip\item[(H5)]\label{H5} We truncate the hierarchy at a fixed order $m_0\ge1$ independently of $N$, thus setting $G_{N,m}\equiv0$ in~\eqref{eq:hierarchy-true} for all~$m>m_0$. This means that correlations of the tagged particle are restricted to involve at most $m_0$ background particles at once.
\end{enumerate}
As we shall see, truncating the hierarchy enables us to control the complexity of possible Feynman diagrams in the perturbative expansion. Nonetheless, the model remains highly nontrivial as it still leads to an infinite Dyson series; see Section~\ref{sec:diag}.

Let us introduce more precisely the simplified hierarchy that we will study. In view of~(H1), with the large-box limit, the phase space for the particles is now
\[\Dd:=\R^d\times\R^d\,\ni\,z=(x,v).\]
Recalling the symmetry, invariance, and orthogonality properties of correlation functions~\eqref{eq:def-CNm}, the state space for simplified correlations is similarly chosen as the direct sum
\[\Hc\,:=\,\bigoplus_{m=0}^\infty\Hc_m,\]
where $\Hc_m$ is the set of functions $G_m:\Dd^{m+1}\to\R:(z_0,\ldots,z_m)\mapsto G_m(z_0,\ldots,z_m)$ that are symmetric in their last $m$ variables $z_1,\ldots,z_m$, that are invariant under spatial translations $(z_0,\ldots,z_m)\mapsto(x_0+x,v_0,\ldots,x_m+x,v_m)$, $x\in\R^d$, and such that
\begin{gather*}
\int_{\Dd^{m+1}} \delta(x_0)\,|G_m(z_{[m]})|^2\,M^{\otimes m+1}(v_{[m]})\,dz_{[m]}<\infty,\\
\int_\Dd G_m(z_{[m]})\,M(v_j)\,dz_j=0,\quad\text{for all $1\le j\le m$,}
\end{gather*}
where we set for abbreviation $[m]:=\{0,\ldots,m\}$.
Recalling the unitarity structure~\eqref{eq:unitarity0} with $\binom{N}{m}\sim \frac{N^m}{m!}$ as $N\uparrow\infty$ for fixed $m$, we endow the space~$\Hc_m$ with the Hilbert norm
\begin{equation}\label{eq:defHm}
\|G_m\|_{\Hc_m}^2\,:=\,\langle G_m,G_m\rangle_{\Hc_m}\,:=\,\frac1{m!}\int_{\Dd^{m+1}}\delta(x_0)\,|G_m(z_{[m]})|^2\,M^{\otimes m+1}(v_{[m]})\,dz_{[m]},
\end{equation}
thus leading to the following norm on the direct sum~$\Hc$,
\[\|(G_m)_m\|_\Hc^2\,:=\,\sum_{m=0}^\infty\|G_m\|_{\Hc_m}^2.\]
By spatial homogeneity, setting $G_0(v_0)\equiv G_0(z_0)$, we emphasize that $\Hc_0$ is identified with the weighted space~$\Ld^2(M\,dv)$.
In this setting, for all $m\ge0$, we consider the skew-adjoint transport operator $iL_m$ and the self-adjoint velocity-diffusion operator $D_m$ on $\Hc_m$,
\begin{equation}\label{eq:mainop}
iL_m:=\sum_{0\le j\le m}v_j\cdot\nabla_{x_j},\qquad
D_m:=-\sum_{0\le j\le m}(\nabla_{v_j}-\tfrac\beta2 v_j)^2.
\end{equation}
Next, only keeping the leading contributions in the actual BBGKY creation and annihilation operators described e.g.\@ in~\cite{DSR-21}, we define
\begin{eqnarray}
iS_m^-G_{m-1}(z_{[m]})&:=&\sum_{0\le j\le m}\sum_{1\le l\le m\atop l\ne j}\nabla\Vc(x_j-x_l)\cdot(\nabla_{v_j}-\tfrac\beta2 v_j)G_{m-1}(z_{[m]\setminus\{l\}}),\label{eq:def-Sm+-}\\
iS_{m-1}^+G_{m}(z_{[m-1]})&:=&\sum_{0\le j\le m-1}\int_\Dd\nabla\Vc(x_j-x_{m})\cdot(\nabla_{v_j}-\tfrac\beta2 v_j)G_{m}(z_{[m]})M(v_{m})\,dz_{m},\nonumber
\end{eqnarray}
in terms of the interaction potential $\Vc\in C^\infty_c(\R^d)$.
These simplified operators have the advantage of satisfying the exact adjointness relation 
\begin{equation}\label{eq:adjS}
\langle H_{m-1},S_{m-1}^+G_{m}\rangle_{\Hc_{m-1}}=\langle S_{m}^-H_{m-1}, G_{m}\rangle_{\Hc_m},\qquad m\ge1.
\end{equation}
Instead of~\eqref{eq:hierarchy-true},
letting $m_0$ denote the maximal order of correlations, cf.~(H5), and including an $O(\frac1N)$ diffusion in velocity, cf.~(H2), we consider the following simplified hierarchy of equations on $\calH$,
\begin{equation}\label{eq:hierarchy-simpl}
\fbox{\text{$\partial_tG_m^{N,m_0}+(iL_m+\tfrac\kappa ND_m)G_m^{N,m_0}=iS_m^+G_{m+1}^{N,m_0}+\tfrac1N iS_m^-G_{m-1}^{N,m_0},\qquad 0\le m\le m_0,$}}
\end{equation}
where we set $G_{m}^{N,m_0}\equiv0$ for $m>m_0$ or $m<0$, and where we let $\kappa>0$ be some diffusion constant. Regarding initial data, we set
\begin{equation}\label{eq:hierarchy-simpl-ci}
G_m^{N,m_0}|_{t=0}\,=\,\left\{\begin{array}{lll}
\mathfrak G&:&m=0,\\
0&:&m\ne0,
\end{array}\right.
\end{equation}
for some initial velocity density $\mathfrak G\in\Hc_0$.
While in the original hierarchy~\eqref{eq:hierarchy-true} we have $m_0=N\uparrow\infty$, we consider here a fixed truncation parameter $m_0$ independently of $N\uparrow\infty$.

By the exact adjointness relation~\eqref{eq:adjS} for $S_m^\pm$,
recalling that the transport operator $iL_m$ is skew-adjoint and that $D_m$ is nonnegative, we directly find
\[\frac{d}{dt}\sum_{m=0}^{m_0}N^{m}\|G_m^{N,m_0}\|_{\Hc_m}^2\,=\,-\frac{2\kappa}N\sum_{m=0}^{m_0}N^{m}\sum_{j=1}^m\|(\nabla_{v_j} -\tfrac\beta2v_j)\,G_m^{N,m_0}\|_{\Hc_m}^2\,\le\,0.\]
Hence, for all $0\le m\le m_0$ and $t\ge0$, we deduce
\begin{equation}\label{eq:apriori-GNm-mod}
\|G_m^{N,m_0}(t)\|_{\Hc_m}\,\le\,N^{-\frac m2}\|\mathfrak G\|_{\Hc_0},
\end{equation}
which is the analogue of~\eqref{eq:apriori-CNm} in the present simplified setting.

Similarly as for the original hierarchy~\eqref{eq:hierarchy-true}, for fixed $m_0$, the time-rescaled tagged particle density~$G_0^{N,m_0}(N\tau,v)$ is now expected to converge weakly as $N\uparrow\infty$ to the solution $G_0(\tau,v)$ of the following Fokker-Planck equation,
\begin{equation}\label{eq:FP-exp-re}
\left\{\begin{array}{l}
\partial_\tau G_0=(\nabla_v-\tfrac\beta 2v)\cdot (\kappa\Id+A_0)(\nabla_{v}-\tfrac\beta 2v)G_0,\\
G_0|_{\tau=0}=\mathfrak G,
\end{array}\right.
\end{equation}
where the diffusion tensor field $A_0\ge0$ is given by the Landau formula
\begin{eqnarray}
A_0(v)&:=&(B_0\ast M)(v),\label{eq:coeff-LB-re}\\
B_0(v)&:=&\frac{\Lambda_\Vc}{|v|}\Big(\Id-\frac{v\otimes v}{|v|^2}\Big)=\Lambda_\Vc\nabla^2|v|,\nonumber
\end{eqnarray}
with the explicit prefactor
$\Lambda_\Vc:=\frac{\omega_{d-1}}{d\omega_d}\int_{\R^d}|k|\pi\V(k)^2\frac{dk}{(2\pi)^d}$,
where $\omega_n$ stands for the volume of the $n$-dimensional unit ball.

\begin{rem}[From Lenard-Balescu to Landau kernel]
Formula~\eqref{eq:coeff-LB-re} for the diffusion tensor can be directly compared with the original Lenard-Balescu expression~\eqref{eq:coeff-LB}, once the different simplifying assumptions of the model are taken into account:
\begin{enumerate}[---]
\item As we now consider a large-box limit, cf.~(H1), Fourier variables become continuous and the collision kernel~$B$ in~\eqref{eq:coeff-LB} is replaced by its continuum version
\begin{equation}\label{eq:coeff-LB-re-LB!}
B(v,w)=\int_{\R^d}(k\otimes k)\,\pi\V(k)^2\frac{\delta(k\cdot w)}{|\e(k,k\cdot v)|^2}\,\frac{dk}{(2\pi)^{d}}.
\end{equation}
\item As we neglect collective screening, cf.~(H3), we replace the dispersion function $\e$ by the constant~$1$. A direct computation then shows that the collision kernel reduces precisely to the above Landau kernel~$B_0$,
\begin{equation}\label{eq:comput-Landau}
B(v,w)\leadsto \int_{\R^d}(k\otimes k)\,\pi\V(k)^2\delta(k\cdot w)\,\frac{dk}{(2\pi)^{d}}\,=\,\frac{\Lambda_\Vc}{|w|}\Big(\Id-\frac{w\otimes w}{|w|^2}\Big).
\end{equation}
\item As we have introduced $O(\frac1N)$ diffusion in velocity, cf.~(H2), we add identity to the obtained diffusion tensor in~\eqref{eq:FP-exp-re}.
\end{enumerate}
\end{rem}
\begin{rem}[Modified equilibrium structure]\label{rem:mod-equ}
Beyond the value of the diffusion tensor, we note that the equilibrium structure in the Fokker-Planck equation~\eqref{eq:FP-exp-re} also differs slightly from that in~\eqref{eq:FP-exp}. More precisely, as the velocity density is recovered via $F_0=MG_0$, we find that~\eqref{eq:FP-exp-re} describes relaxation to~$\sqrt M$ instead of~$M$. This discrepancy arises from our symmetric choice of the zeroth-order terms in the definition of the creation and annihilation operators~\eqref{eq:def-Sm+-} (and accordingly for the velocity-diffusion operator~\eqref{eq:mainop}). While this choice is made primarily for computational convenience, it has no essential impact on the results.
\end{rem}

\subsection{Main result}
In the framework of the simplified hierarchy~\eqref{eq:hierarchy-simpl}--\eqref{eq:hierarchy-simpl-ci}, for a fixed truncation parameter $m_0\ge1$, we rigorously prove thermalization and derive the expected Fokker-Planck equation~\eqref{eq:FP-exp-re} for the tagged particle density.

\begin{theor}\label{th:main}
Fix a truncation parameter $m_0\ge1$ and
assume that:
\begin{enumerate}[---]
\item the space dimension is sufficiently large depending on $m_0$, in the sense that
\begin{equation}\label{eq:def de d_0}
d\ge d_0=\left\{\begin{array}{cl}
2,&\text{if $m_0=1$},\\
8,&\text{if $m_0=2$},\\
28m_0+70,&\text{if $m_0\ge3$},
\end{array}\right.\end{equation}
\item the diffusion constant $\kappa$ is sufficiently large, in the sense that $\kappa\ge C_{m_0}\int_{\R^d}|k|\V(k)\,dk$ for some constant~$C_{m_0}>0$ depending only on $m_0$. In case $m_0=1$, this condition can be dropped.
\end{enumerate}
Then, for any initial condition $\mathfrak G\in C_c^\infty(\R^d)$, the solution $(G_m^{N,m_0})_{0\le m\le m_0}$ of the simplified hierarchy~\eqref{eq:hierarchy-simpl}--\eqref{eq:hierarchy-simpl-ci} satisfies that the time-rescaled tagged-particle density
\[(\tau,v)\mapsto G_0^{N,m_0}(N\tau,v)\]
converges strongly in $\Ld^2_\loc(\R^+;\Ld^2(M\, dv))$ to the solution $G_0$ of the Fokker-Planck equation~\eqref{eq:FP-exp-re}.
More precisely, we have
\begin{equation*}
\Big(\int_0^\infty e^{-2\tau}\|G_0^{N,m_0}(N\tau)-G_0(\tau)\|_{L^2(M\,dv)}^2\,d\tau\Big)^\frac12 \,\lesssim\,N^{-\frac1{12}}\|\<(\nabla_{v_0},v_0)\>^{4m_0+10}\mathfrak G\|_{L^2(M\, dv)},
\end{equation*}
and this bound can be improved to $N^{-1}\|\<(\nabla_{v_0},v_0)\>^{4m_0+21}\mathfrak G\|_{L^2(M\,dv)}$ if $d\ge 28m_0+147$.
\end{theor}

As detailed in Section~\ref{sec:strategy}, the proof relies on a careful analysis of the Dyson expansion of the hierarchy in terms of Feynman diagrams and builds on three key new ingredients:
\begin{enumerate}[---]
\item {\it Renormalization:} The Dyson series must be renormalized to eliminate the leading recollisions, which amounts to factoring out part of the expansion. Concretely, the free transport propagators around which the expansion is performed are replaced by renormalized propagators, obtained by adding to the free transport an approximate (non-Markovian) version of the limiting Fokker-Planck operator (so-called ``hat operator'' below).
\smallskip\item {\it Phase mixing:} Since annihilation operators involve velocity averages, cf.~\eqref{eq:def-Sm+-}, phase mixing can be effectively exploited in several terms of the expansion. In our approach, this is achieved by systematically applying contour deformations in all the integrals over free velocity variables.
\smallskip\item {\it Hypoelliptic estimates:} As explained, the renormalization effectively augments the free transport with a (non-Markovian) Fokker-Planck-type operator carrying a prefactor of order $O(\frac1N)$. By treating this, heuristically, as a genuine Fokker-Planck operator, filamentation effects in phase space are mitigated through the induced $O(\frac1N)$ velocity diffusion. This allows to replace naive resolvent estimates with hypoelliptic ones, which substantially improve the scaling in $N$. Since such estimates for the actual renormalized propagators are not directly accessible due to non-Markovian effects, we exploit the additional $O(\frac1N)$ velocity diffusion included for simplicity in the system to justify these improved bounds.
\end{enumerate}
As discussed heuristically in Section~\ref{sec:limitation} below, the restriction~\eqref{eq:def de d_0} on the space dimension relative to the truncation parameter arises from geometric constraints in the complex deformations used for phase-mixing arguments and does not appear to be avoidable at present. Several refinements of the method could improve the value of $d_0$ in~\eqref{eq:def de d_0}, but none seem enough to remove the linear dependence on the truncation parameter.

\begin{rem}[Scope and limitations of the model]
Several of the simplifying assumptions in our model could be relaxed without substantial modifications. In particular, the free transport operator could be replaced by the linearized Vlasov operator appearing in the original hierarchy, cf.~\eqref{eq:lin-mf0}, in which case the Landau kernel in the formula~\eqref{eq:coeff-LB-re} for the diffusion tensor would be replaced by the Lenard-Balescu kernel~\eqref{eq:coeff-LB-re-LB!}. Likewise, the modification of the equilibrium structure noted in Remark~\ref{rem:mod-equ} could be removed with only minor adjustments.
A more consequential simplification of the model concerns the choice of creation and annihilation operators as exact adjoints; we expect that the approximate unitarity of the original hierarchy~\eqref{eq:hierarchy-true} would still suffice for the analysis, but this remains to be verified. Another technical simplification is the inclusion of an $O(\frac1N)$ velocity diffusion in the system, which we cannot yet dispense with. Finally, the most restrictive assumption is, of course, the truncation of the hierarchy itself, which currently appears unavoidable.
\end{rem}

\begin{rem}[Related work]
While completing this work, we became aware of an ongoing parallel investigation by Bodineau and Le Bris~\cite{BLB}, which addresses another simplified model in the direction of a rigorous justification of Lenard-Balescu thermalization. Specifically, they derive a similar linearized Landau equation for a tagged particle interacting via mean-field forces within an ideal Rayleigh gas --- namely, a background of particles that do not interact with each other, only indirectly through the tagged particle.
Their analysis follows a trajectorial approach, but our hierarchical framework could in principle be applied to their model as well. In that case, the diagrammatic structure is significantly simpler: since background particles do not interact, the creation and annihilation operators always correspond to collisions involving the tagged particle, and the pathological diagrams~\eqref{eq:err-ex} and~\eqref{eq:FN-diag-rem} that constitute the main difficulty in our analysis would be absent. In this simplified setting, we expect that the truncation of the hierarchy could be avoided, leading to a hierarchical proof of their result. We note, however, that our approach would still require the inclusion of an $O(\frac1N)$ velocity diffusion in the system, which is not needed in~\cite{BLB}.
\end{rem}

\section{Strategy of the proof}\label{sec:strategy}
We start with a brief overview of previous attempts to establish thermalization, focussing on the above simplified hierarchy~\eqref{eq:hierarchy-simpl}, and we then explain and motivate our main new ideas. Throughout this section, we restrict ourselves to high-level, non-rigorous explanations, aiming to provide intuition and motivation for our approach.

\subsection{Rescaled Fourier reformulation}
Before going further, we rescale time by introducing the slow variable~$\tau:=t/t_N$, where the timescale $t_N\gg1$ is left unspecified for now but will later be chosen as the predicted thermalization timescale $t_N=N$.
For all $j\ge0$, we denote by $k_j$ the Fourier-dual variable associated with $x_j$, and we set $\hat z_j:=(k_j,v_j)$.
In these terms, further recalling the a priori estimates~\eqref{eq:apriori-GNm-mod} and the definition~\eqref{eq:defHm} of the norms, we make the following change of unknowns: for all~$0\le m\le m_0$,
\begin{equation}\label{eq:def-gm}
g^{N,m_0}_m(\tau,\hat z_{[m]}) \,:=\, \sqrt{\tfrac{N^m}{m!}M^{\otimes m+1}(v_{[m]})}
\int_{(\R^d)^{m+1}}\Big(\prod_{j=0}^m e^{-ik_j\cdot x_j}\Big)\,G^{N,m_0}_m(t_N\tau,z_{[m]})\,d x_{[m]}.
\end{equation}
and we further let $g^{N,m_0}_m=0$ for $m>m_0$ or $m<0$.
Since $G_m^{N,m_0}$ is invariant under spatial translations, we note that $g_m^{N,m_0}$ is concentrated on the set
\[\hat\Dd^{m+1} \,:=\, \bigg\{(\hat z_0,\ldots,\hat z_m)\in \Dd^{m+1}~:~\sum_{j= 0}^mk_j=0\bigg\}.\]
The a priori estimates~\eqref{eq:apriori-GNm-mod} are then replaced by the following, for all $0\le m\le m_0$ and $\tau\ge0$,
\begin{equation}\label{eq:apriori-GNm-re}
\|g_m^{N,m_0}(\tau)\|_{\Ld^2(\hat\Dd^{m+1})}\,\le\,\|\mathfrak g\|_{\Ld^2(\hat\Dd^1)},
\end{equation}
where we have set for abbreviation
\[\mathfrak g:=\sqrt M\mathfrak G\qquad\text{on $\hat\Dd^1=\{0\}\times\R^d$.}\]
In this setting, the state space $\Hc_m$ is thus replaced by $\Ld^2(\hat\Dd^{m+1})$
and the hierarchy~\eqref{eq:hierarchy-simpl} is transformed into the following,
\begin{equation}\label{eq:hierarchy}
\left\{\begin{array}{l}
t_N^{-1}\partial_\tau g_m^{N,m_0}+(i\hat L_m+\tfrac\kappa N\hat D_m)g_m^{N,m_0}
=\frac1{\sqrt N}(i\hat S^+_mg^{N,m_0}_{m+1}+i\hat S_m^-g^{N,m_0}_{m-1}),\qquad0\le m\le m_0,\\[1mm]
g_m^{N,m_0}|_{\tau=0}=\mathfrak g\mathds1_{m=0},
\end{array}\right.
\end{equation}
where the transport and diffusion operators are now replaced by
\begin{equation}\label{eq:def-hatop}
\hat L_m\,:=\,\sum_{0\le j\le m}k_j\cdot v_j,\qquad
\hat D_m\,:=\,-\sum_{0\le j\le m}\triangle_{v_j},
\end{equation}
and where creation and annihilation operators take the form
\begin{eqnarray}
\hat S_m^-g_{m-1}&:=&\sum_{0\le j\le m}\sum_{1\le l\le m\atop l\ne j}\hat S^{m,-}_{j,l}g_{m-1},\label{eq:def-hatS}\\
\hat S_{m-1}^+g_{m}&:=&\sum_{0\le j\le m-1}\hat S^{m-1,+}_{j,m}g_m,\nonumber
\end{eqnarray}
in terms of
\begin{eqnarray}
\hat S^{m,-}_{j,l}g_{m-1}(\hat z_{[m]})
&:=&-\frac1{\sqrt m}\sqrt M(v_l)\,k_l\V(k_l)\cdot\nabla_{v_j}g_{m-1}(\hat z_{[m]\setminus\{j,l\}},(k_j+k_l,v_j)),\label{eq:def-hat-coll}\\
\hat S_{j,m}^{m-1,+}g_m(\hat z_{[m-1]})
&:=&\sqrt m\int_{\Dd} \sqrt M(v_{m})\,k_{m}\V(k_{m})\cdot\nabla_{v_j}g_{m}(\hat z_{[m]\setminus\{j\}},(k_j-k_{m},v_j))\,d^*\hat z_{m},\nonumber
\end{eqnarray}
with the short-hand notation $d^*\hat z_j=d^*k_jdv_j$ and $d^*k_j=(2\pi)^{-d}dk_j$.
In this setting, the adjointness relation~\eqref{eq:adjS} becomes
\begin{equation}\label{eq:adjS-re}
\langle h_{m-1},\hat S_{m-1}^+g_{m}\rangle_{\Ld^2(\hat\Dd^{m})}=\langle \hat S_{m}^-h_{m-1}, g_{m}\rangle_{\Ld^2(\hat\Dd^{m+1})},\qquad 1\le m\le m_0.
\end{equation}

\subsection{Formal derivation}\label{sec:formal}
A formal examination of the hierarchy~\eqref{eq:hierarchy} leads to expect
\begin{equation}\label{eq:expect-hier}
g_m^{N,m_0}=O(C_mN^{-\frac m2}),
\end{equation}
so that the a priori estimates~\eqref{eq:apriori-GNm-re} (or~\eqref{eq:apriori-GNm-mod}) would be far from optimal. If this was true, we could truncate the hierarchy: neglecting the effect of $g_2^{N,m_0}=O(N^{-1})$, and focussing on the critical timescale~$t_N=N$, we would be formally led to
\begin{eqnarray}
\partial_\tau g_0^{N,m_0}+\kappa\hat D_0g_0^{N,m_0}&=&i\hat S_0^+(\sqrt Ng_1^{N,m_0}),\label{eq:trunc-hier}\\
(\tfrac1N\partial_\tau +\hat L_1+\tfrac\kappa N\hat D_1)(\sqrt N g_1^{N,m_0})&=&i\hat S_1^-g_0^{N,m_0}+O(N^{-1}).\nonumber
\end{eqnarray}
This amounts to a slow equation for $g_0^{N,m_0}$ coupled to a fast transport equation for the rescaled correlation~$\sqrt Ng_1^{N,m_0}=O(1)$. Solving the latter, we are led to a closed non-Markovian description for the tagged particle density,
\begin{equation}\label{eq:solve-trunc-hier}
\partial_\tau g_0^{N,m_0}+\kappa\hat D_0g_0^{N,m_0}=-\hat S_0^+\Big(\tfrac1N\partial_\tau +\hat L_1+\tfrac\kappa N\hat D_1\Big)^{-1}\hat S_1^-g_0^{N,m_0}+O(N^{-1}),
\end{equation}
where we expect memory effects to average out as $N\uparrow\infty$.
To investigate this more precisely, we introduce the following version of the Laplace transform, for $\varphi:\mathbb{R}^+\to \mathbb{C}$,
\begin{equation}\label{eq:Lap}
\Lc\varphi(\alpha) \,:=\, \int_0^\infty\varphi(\tau)\,e^{-(1+i\alpha )\tau}\,d\tau,\qquad\varphi(\tau)=e^\tau\int_\R e^{i\alpha\tau}\Lc\varphi(\alpha)\,d^*\alpha,\qquad d^*\alpha:=\tfrac{d\alpha}{2\pi}.
\end{equation}
Taking Laplace transform and inserting the definition of $\hat S_0^+,\hat S_1^-,\hat L_1$, equation~\eqref{eq:solve-trunc-hier} becomes
\begin{equation*}
(1+i\alpha)\Lc g_0^{N,m_0}-\Div_{v_0}((\kappa\Id+ A_0^N)\nabla_{v_0}\Lc g_0^{N,m_0})=\mathfrak g+O(N^{-1}),
\end{equation*}
in terms of
\[ A_0^N(\alpha,v_0)\,:=\,\int_\Dd (k_1\otimes k_1)\V(k_1)^2\sqrt M(v_1)\Big(\tfrac{1+i\alpha}N +ik_1\cdot(v_1-v_0)+\tfrac\kappa N\hat D_1\Big)^{-1}\sqrt M(v_1)\,d^*\hat z_1.\]
In this formulation, the non-Markovian character of the equation arises from the $\alpha$-dependence of $A_0^N$. Passing to the limit $N\uparrow\infty$, one can check that $A_0^N(\alpha,v_0)$ converges pointwise to the diffusion coefficient~$A_0(v_0)$ defined in~\eqref{eq:coeff-LB-re} (see the proof of Lemma~\ref{lem:kin-lim} for details).
Inverting Laplace transform, we are thus led to $g_0^{N,m_0}\to g_0$, where $g_0$ satisfies
\begin{equation}\label{eq:FP-deriv}
\left\{\begin{array}{l}
\partial_\tau g_0=\Div_{v}((\kappa\Id+A_0)\nabla_{v}g_0),\\
g_0|_{t=0}=\mathfrak g.
\end{array}\right.
\end{equation}
Recalling~$g_0^{N,m_0}(\tau)=\sqrt MG_0^{N,m_0}(N\tau)$, cf.~\eqref{eq:def-gm}, this would imply
\[G_0^{N,m_0}(N\tau)\to G_0(\tau):=\sqrt M^{-1}g_0(\tau),\]
which satisfies the expected Fokker-Planck equation~\eqref{eq:FP-exp-re}.

This formal derivation has several important gaps.
First, we currently lack a method to establish the improved correlation estimates~\eqref{eq:expect-hier} beyond short timescales and, moreover, such estimates might only hold in weak topologies.
Second, even if these bounds were valid in $\Ld^2$ up to the relevant timescale $t_N=N$, the $O(N^{-1})$ error in~\eqref{eq:trunc-hier} would still only be controlled in $H^{-1}$, owing to the loss of a velocity derivative in $\hat S_1^+$. As a consequence, when solving the fast transport equation for correlations, the contribution of this error in~\eqref{eq:solve-trunc-hier} could be significantly amplified due to filamentation effects.
Both difficulties stem from the fast oscillations of correlation functions over long timescales, which may induce problematic resonances. A potential way forward, as we already argued in~\cite{DSR-21}, is to start from a perturbative expansion of all correlation functions as Dyson series, in the hope that explicit oscillatory contributions never resonate.

\subsection{Dyson series expansion}\label{sec:diag}
Applying Laplace transform in the form~\eqref{eq:Lap}, the simplified hierarchy~\eqref{eq:hierarchy} reads as follows,
\begin{equation}\label{eq:Duhamel}
\left\{\begin{array}{l}
\big(1+i\alpha+\kappa\tfrac{t_N}{N}\hat D_0\big)\Lc g_0^{N,m_0}=\mathfrak g+\tfrac{t_N}{\sqrt N}i\hat S_0^+\Lc g_{1}^{N,m_0},\\[2mm]
\big(\tfrac{1+i\alpha}{t_N}+i\hat L_m+\tfrac\kappa N\hat D_m\big)\Lc g_m^{N,m_0}=\tfrac1{\sqrt N}\Big(i\hat S_m^+\Lc g_{m+1}^{N,m_0}+i\hat S_m^-\Lc g_{m-1}^{N,m_0}\Big),\qquad 1\le m\le m_0.
\end{array}\right.
\end{equation}
Iteratively solving this hierarchy, we obtain an infinite Dyson series for each correlation function $g_m^{N,m_0}$. Each term in this expansion consists of a sequence of resolvents
\[(\tfrac{1+i\alpha}{t_N}+i\hat L_m+\tfrac\kappa N\hat D_m)^{-1},\]
interlaced with creation or annihilation operators $\hat S_m^\pm$, acting on the initial data $\mathfrak g$.
To organize these contributions, we introduce a representation using Feynman-type diagrams:
\begin{enumerate}[---]
\item Each connected horizontal line corresponds to a different particle, with the lower line corresponding to the tagged particle `$0$'.
\smallskip\item Parallel horizontal segments stand for free propagators,
\[\big(\tfrac{1+i\alpha}{t_N}+i\hat L_1+\tfrac\kappa N\hat D_1\big)^{-1}=
{\tiny\begin{tikzpicture}[baseline={([yshift=-1ex]current bounding box.center)},scale=0.8]
\draw (0,0) -- (1,0);
\draw (0,0.3) -- (1,0.3);
\end{tikzpicture}}\qquad
\big(\tfrac{1+i\alpha}{t_N}+i\hat L_2+\tfrac\kappa N\hat D_2\big)^{-1}=
{\tiny\begin{tikzpicture}[baseline={([yshift=-1ex]current bounding box.center)},scale=0.8]
\draw (0,0) -- (1,0);
\draw (0,0.3) -- (1,0.3);
\draw (0,0.6) -- (1,0.6);
\end{tikzpicture}}\qquad\text{etc.}\]
\item Vertical segments merging two horizontal levels are viewed as ``collisions'' and correspond to applying creation or annihilation operators,
\[\hat S^{m,+}_{j,l}
={\tiny\begin{tikzpicture}[baseline={([yshift=-1ex]current bounding box.center)},scale=0.8]
\begin{scope}[every node/.style={circle,draw,fill=white,inner sep=0pt,minimum size=3pt}]
\node (1) at (0,0) {};
\end{scope}
\draw (1) -- (0,0.5);
\draw[dotted] (-0.4,0) -- (1) -- (0.4,0);
\draw[dotted] (0,0.5) -- (0.4,0.5);
\node at (1.5,0.5) {\tiny (particle $l$)};
\node at (1.5,0) {\tiny (particle $j$)};
\end{tikzpicture}}\qquad
\hat S^{m,-}_{j,l}
={\tiny\begin{tikzpicture}[baseline={([yshift=-1ex]current bounding box.center)},scale=0.8]
\begin{scope}[every node/.style={circle,draw,fill=white,inner sep=0pt,minimum size=3pt}]
\node (1) at (0,0) {};
\end{scope}
\draw (1) -- (0,0.5);
\draw[dotted] (-0.4,0) -- (1) -- (0.4,0);
\draw[dotted] (0,0.5) -- (-0.4,0.5);
\node at (1.5,0.5) {\tiny (particle $l$)};
\node at (1.5,0) {\tiny (particle $j$)};
\end{tikzpicture}}\]
We emphasize that for~$\hat S_{j,l}^{m,-}$ the annihilated index $l$ only runs over $1\le l\le m$, meaning that the tagged particle `$0$' is never annihilated.
\smallskip\item Diagrams are arbitrary compositions of the above two ingredients, alternating between free propagators and collisions.
The \emph{complexity of a diagram} is the maximal number of horizontal levels that appear at once on top of the lower level: by assumption, it is always bounded by the truncation parameter~$m_0$. Viewing time as flowing from right to left, we call \emph{input} (resp.\@ \emph{output}) variables the list of particle indices that are present at the right (resp.\@ left) of the diagram.
\end{enumerate}
To give a concrete example,
\[{\tiny\begin{tikzpicture}[baseline={([yshift=-1ex]current bounding box.center)},scale=0.8]
\begin{scope}[every node/.style={circle,draw,fill=white,inner sep=0pt,minimum size=3pt}]
\node (1) at (0,0) {};
\node (2) at (0.5,0) {};
\node (3) at (1,0.5) {};
\node (4) at (1.5,0) {};
\end{scope}
\draw (1) -- (2) -- (4) -- (1.5,0.5) -- (3) -- (0,0.5) -- (1);
\draw (2) -- (0.5,0.25) -- (1,0.25) -- (3);
\end{tikzpicture}}
\,=\, S_{0,1}^{0,+}\big(\tfrac{1+i\alpha}{t_N}+i\hat L_1+\tfrac\kappa N\hat D_1\big)^{-1}S_{0,2}^{1,+}\big(\tfrac{1+i\alpha}{t_N}+i\hat L_2+\tfrac\kappa N\hat D_2\big)^{-1}S_{1,2}^{2,-}\big(\tfrac{1+i\alpha}{t_N}+i\hat L_1+\tfrac\kappa N\hat D_1\big)^{-1}S_{0,1}^{1,-},\]
and thus, inserting the definition of the operators and carefully tracking the Fourier variables,
\begin{multline*}
{\tiny\begin{tikzpicture}[baseline={([yshift=-1ex]current bounding box.center)},scale=0.8]
\begin{scope}[every node/.style={circle,draw,fill=white,inner sep=0pt,minimum size=3pt}]
\node (1) at (0,0) {};
\node (2) at (0.5,0) {};
\node (3) at (1,0.5) {};
\node (4) at (1.5,0) {};
\end{scope}
\draw (1) -- (2) -- (4) -- (1.5,0.5) -- (3) -- (0,0.5) -- (1);
\draw (2) -- (0.5,0.25) -- (1,0.25) -- (3);
\end{tikzpicture}}\,\mathfrak g(\alpha,v_0)
\,=\, \int_{\Dd^2}
\sqrt M(v_1)\,k_1\V(k_1)\cdot\nabla_{v_0}\Big(\tfrac{1+i\alpha}{t_N}+ik_1\cdot(v_1-v_0)-\tfrac\kappa N\triangle_{v_{[1]}}\Big)^{-1}\\
\times  \sqrt M(v_2)\,k_2\V(k_2)\cdot\nabla_{v_0}\Big(\tfrac{1+i\alpha}{t_N}+ik_1\cdot(v_1-v_0)+ik_2\cdot(v_2-v_0)-\tfrac\kappa N\triangle_{v_{[2]}}\Big)^{-1}\\
\times \sqrt M(v_2)\,k_2\V(k_2)\cdot\nabla_{v_0}\Big(\tfrac{1+i\alpha}{t_N}+i(k_1+k_2)\cdot(v_1-v_0)-\tfrac\kappa N\triangle_{v_{[1]}}\Big)^{-1}\\
\times\sqrt M(v_1)\,(k_1+k_2)\V(k_1+k_2)\cdot\nabla_{v_0}\mathfrak g(v_0)\,d^*\hat z_1d^*\hat z_2,
\end{multline*}
where for an index set $A$ we let $\triangle_{v_A}=\sum_{j\in A}\triangle_{v_{j}}$.
Occasionally, we shall indicate the associated Fourier momentum variables above each horizontal segment, thus representing the above integral as
\[=\,\int_{(\R^d)^2}
{\tiny\begin{tikzpicture}[baseline={([yshift=-2ex]current bounding box.center)},scale=0.8]
\begin{scope}[every node/.style={circle,draw,fill=white,inner sep=0pt,minimum size=3pt}]
\node (1) at (0,0) {};
\node (2) at (1.5,0) {};
\node (3) at (3,0.8) {};
\node (4) at (4.5,0) {};
\end{scope}
\draw (1) -- (2) -- (4) -- (4.5,0.8) -- (3) -- (0,0.8) -- (1);
\draw (2) -- (1.5,0.4) -- (3,0.4) -- (3);
\node at (0.75,0.17) {$-k_1$};
\node at (1.5,0.97) {$k_1$};
\node at (3,0.17) {$-k_1-k_2$};
\node at (2.25,0.57) {$k_2$};
\node at (3.75,0.97) {$k_1+k_2$};
\end{tikzpicture}}
\,\mathfrak g(v_0)\,d^*k_1d^*k_2.\]
Note that the definition of the creation and annihilation operators ensures the conservation of the sum of Fourier variables at each collision.

With this diagrammatic notation, we return to~\eqref{eq:Duhamel}: taking resolvents, the Duhamel formula can be written as follows, for all $0\le m\le m_0$,
\begin{equation}\label{eq:Duhamel-re}
\Lc g_m^{N,m_0}=
\tfrac{i}{\sqrt N}\Big(\sum ~{\tiny\begin{tikzpicture}[baseline={([yshift=-1ex]current bounding box.center)},scale=0.8]
\begin{scope}[every node/.style={circle,draw,fill=white,inner sep=0pt,minimum size=3pt}]
\node (1) at (0.8,0.5) {};
\end{scope}
\draw (0,0) -- (0.8,0);
\draw (0,0.5) -- (1);
\draw (0,0.8) -- (0.8,0.8) -- (1);
\draw[dotted] (0.4,0.1) -- (0.4,0.4);
\end{tikzpicture}}\,\Big)\Lc g_{m-1}^{N,m_0}
+\tfrac{i}{\sqrt N}\Big(\sum {\tiny\begin{tikzpicture}[baseline={([yshift=-1ex]current bounding box.center)},scale=0.8]
\begin{scope}[every node/.style={circle,draw,fill=white,inner sep=0pt,minimum size=3pt}]
\node (1) at (0.8,0.5) {};
\end{scope}
\draw (0,0) -- (0.9,0);
\draw (0,0.5) -- (1) -- (0.8,0.8) -- (0.9,0.8);
\draw (1) -- (0.9,0.5);
\draw[dotted] (0.4,0.1) -- (0.4,0.4);
\end{tikzpicture}}~\Big)\Lc g_{m+1}^{N,m_0}+\mathds1_{m=0}t_N^{-1}\mathfrak g,
\end{equation}
where the sums run over all possibilities for creating or annihilating a particle.
Iterating this hierarchy yields the standard Dyson series, that is, a formal expansion of each correlation function in terms of diagrams acting on the initial data $\mathfrak g$.
For our purposes, however, it is convenient to reorganize this expansion slightly: we stop the iteration as soon as we reach~$\Lc g_0^{N,m_0}$, which produces a series expansion in which the diagrams act directly on the tagged particle density $\Lc g_0^{N,m_0}$.

More precisely, for $1\le m\le m_0$, the $m$th-order correlation is formally expressed as the sum of all diagrams with complexity at most~$m_0$, with~$1$ input and~$m+1$ output particles, starting with a collision operator and ending with a free propagator, and never returning to a single particle at intermediate times. Each diagram is applied to $\Lc g_0^{N,m_0}$, and a diagram involving~$n$ collisions carries a prefactor~$t_N(\frac{i}{\sqrt N})^n$. Inserting such expansions into the first equation in~\eqref{eq:Duhamel}, we obtain the following for the tagged particle density,
\begin{multline}\label{eq:Dyson}
\big(1+i\alpha-\tfrac{t_N}{N}\triangle_{v_0}\big)\Lc g_0^{N,m_0}
\,=\,\mathfrak g
\,+\,t_N\big(\tfrac{i}{\sqrt N}\big)^2\,{\tiny\begin{tikzpicture}[baseline={([yshift=-1ex]current bounding box.center)},scale=0.8]
\begin{scope}[every node/.style={circle,draw,fill=white,inner sep=0pt,minimum size=3pt}]
\node (1) at (0,0) {};
\node (2) at (0.5,0) {};
\end{scope}
\draw (1) -- (2) -- (0.5,0.3) -- (0,0.3) -- (1);
\end{tikzpicture}}\,\Lc g_0^{N,m_0}\\
\,+\,t_N\big(\tfrac{i}{\sqrt N}\big)^4\Big(
{\tiny\begin{tikzpicture}[baseline={([yshift=-1ex]current bounding box.center)},scale=0.8]
\begin{scope}[every node/.style={circle,draw,fill=white,inner sep=0pt,minimum size=3pt}]
\node (1) at (0,0) {};
\node (2) at (0.5,0) {};
\node (3) at (1,0) {};
\node (4) at (1.5,0) {};
\end{scope}
\draw (1) -- (2) -- (3) -- (4);
\draw (1) -- (0,0.25) -- (1,0.25) -- (3);
\draw (2) -- (0.5,0.5) -- (1.5,0.5) -- (4);
\end{tikzpicture}}
+{\tiny\begin{tikzpicture}[baseline={([yshift=-1ex]current bounding box.center)},scale=0.8]
\begin{scope}[every node/.style={circle,draw,fill=white,inner sep=0pt,minimum size=3pt}]
\node (1) at (0,0) {};
\node (2) at (0.5,0.25) {};
\node (3) at (1,0.25) {};
\node (4) at (1.5,0) {};
\end{scope}
\draw (1) -- (0,0.25) -- (2) -- (0.5,0.5) -- (1,0.5) -- (3) -- (1.5,0.25) -- (4) -- (1);
\draw (2) -- (3);
\end{tikzpicture}}
+{\tiny\begin{tikzpicture}[baseline={([yshift=-1ex]current bounding box.center)},scale=0.8]
\begin{scope}[every node/.style={circle,draw,fill=white,inner sep=0pt,minimum size=3pt}]
\node (1) at (0,0) {};
\node (2) at (0.5,0) {};
\node (3) at (1,0) {};
\node (4) at (1.5,0) {};
\end{scope}
\draw (1) -- (2) -- (3) -- (4);
\draw (1) -- (0,0.5) -- (1.5,0.5) -- (4);
\draw (2) -- (0.5,0.25) -- (1,0.25) -- (3);
\end{tikzpicture}}
+{\tiny\begin{tikzpicture}[baseline={([yshift=-1ex]current bounding box.center)},scale=0.8]
\begin{scope}[every node/.style={circle,draw,fill=white,inner sep=0pt,minimum size=3pt}]
\node (1) at (0,0) {};
\node (2) at (0.5,0.25) {};
\node (3) at (1,0.5) {};
\node (4) at (1.5,0) {};
\end{scope}
\draw (1) -- (4);
\draw (1) -- (0,0.25) -- (2) -- (1,0.25) -- (3) -- (1.5,0.5) -- (4);
\draw (2) -- (0.5,0.5) -- (3);
\end{tikzpicture}}
\\
+{\tiny\begin{tikzpicture}[baseline={([yshift=-1ex]current bounding box.center)},scale=0.8]
\begin{scope}[every node/.style={circle,draw,fill=white,inner sep=0pt,minimum size=3pt}]
\node (1) at (0,0) {};
\node (2) at (0.5,0.5) {};
\node (3) at (1,0) {};
\node (4) at (1.5,0) {};
\end{scope}
\draw (1) -- (3) -- (4);
\draw (1) -- (0,0.5) -- (2) -- (1.5,0.5) -- (4);
\draw (2) -- (0.5,0.25) -- (1,0.25) -- (3);
\end{tikzpicture}}
+{\tiny\begin{tikzpicture}[baseline={([yshift=-1ex]current bounding box.center)},scale=0.8]
\begin{scope}[every node/.style={circle,draw,fill=white,inner sep=0pt,minimum size=3pt}]
\node (1) at (0,0) {};
\node (2) at (0.5,0) {};
\node (3) at (1,0.5) {};
\node (4) at (1.5,0) {};
\end{scope}
\draw (1) -- (2) -- (4);
\draw (1) -- (0,0.5) -- (3) -- (1.5,0.5) -- (4);
\draw (2) -- (0.5,0.25) -- (1,0.25) -- (3);
\end{tikzpicture}}
+{\tiny\begin{tikzpicture}[baseline={([yshift=-1ex]current bounding box.center)},scale=0.8]
\begin{scope}[every node/.style={circle,draw,fill=white,inner sep=0pt,minimum size=3pt}]
\node (1) at (0,0) {};
\node (2) at (0.5,0.25) {};
\node (3) at (1,0) {};
\node (4) at (1.5,0) {};
\end{scope}
\draw (1) -- (3) -- (4);
\draw (1) -- (0,0.25) -- (2) -- (1,0.25) -- (3);
\draw (2) -- (0.5,0.5) -- (1.5,0.5) -- (4);
\end{tikzpicture}}
+{\tiny\begin{tikzpicture}[baseline={([yshift=-1ex]current bounding box.center)},scale=0.8]
\begin{scope}[every node/.style={circle,draw,fill=white,inner sep=0pt,minimum size=3pt}]
\node (1) at (0,0) {};
\node (2) at (0.5,0) {};
\node (3) at (1,0.25) {};
\node (4) at (1.5,0) {};
\end{scope}
\draw (1) -- (2) -- (4);
\draw (1) -- (0,0.5) -- (1,0.5) -- (3);
\draw (2) -- (0.5,0.25) -- (3) -- (1.5,0.25) -- (4);
\end{tikzpicture}}
\Big)\,\Lc g_0^{N,m_0}+\ldots
\end{multline}
Note that the expansion is always an infinite series despite the finite maximal complexity $m_0$ of the diagrams.
Each collision comes with a velocity derivative, cf.~\eqref{eq:def-hat-coll}, so a term with~$n$ collisions involves up to~$n$ derivatives of $\Lc g_0^{N,m_0}$ --- or of the test function, by duality.

For~$t_N=N$, the first line in~\eqref{eq:Dyson} coincides exactly with the truncated non-Markovian description~\eqref{eq:solve-trunc-hier} derived formally,
\begin{equation}\label{eq:solve-trunc-hier-re}
\big(1+i\alpha-\kappa\triangle_{v_0}+\BOX\big)\Lc g_0^{N,m_0}
\,=\,\mathfrak g+\ldots,
\end{equation}
where the `hat' operator is by definition
\begin{equation}\label{eq:BOX}
\BOX\,:=\,\hat S_{0,1}^{0,+}\Big(\tfrac{1+i\alpha}{t_N}+i\hat L_1+\tfrac\kappa N\hat D_1\Big)^{-1}\hat S_{0,1}^{1,-}.
\end{equation}
The problem is thus reduced to showing that all higher-order terms in the expansion vanish as~\mbox{$N\uparrow\infty$}. We emphasize the advantage of expanding around the tagged particle density $\Lc g_0^{N,m_0}$ rather than the initial data $\mathfrak g$: it effectively factors out part of the usual Dyson expansion, allowing the leading dynamics to emerge directly. This is crucial since the Dyson expansion of the limit dynamics converges only for specific initial data of exponential type, owing to the accumulation of velocity derivatives --- in contrast with other problems in kinetic theory, such as Lanford's theorem, where the limiting Boltzmann equation admits a Dyson series that converges for a wide range of data over short times.

As we shall see, recognizing all diagrams in~\eqref{eq:Dyson} as velocity averages and using phase mixing in the form of contour deformations in the $v$-integrals, it is easily checked that each diagram in this series is uniformly bounded in the weak sense in large enough dimension. We thus find that, term by term, the error is $O(t_N(\frac1{\sqrt N})^4)=O(N^{-1})$, but convergence of the series is problematic.

Aiming at a rigorous justification, we truncate the formal expansion~\eqref{eq:Dyson} and need to estimate the remainder.
Restricting the expansion to diagrams with $\le 2n$ collisions, the Duhamel formula for the remainder involves various diagrams with $2n+1$ or $2n+2$ collisions applied to correlation functions~$\Lc g_m^{N,m_0}$ with~$1\le m\le m_0$, such as
\begin{equation}\label{eq:err-ex}
E_N\,:=\,t_N\big(\tfrac1{\sqrt N}\big)^{2n+1}\,{\tiny\begin{tikzpicture}[baseline={([yshift=1ex]current bounding box.center)},scale=0.8]
\begin{scope}[every node/.style={circle,draw,fill=white,inner sep=0pt,minimum size=3pt}]
\node (0) at (-0.5,0) {};
\node (1) at (0,0) {};
\node (2) at (0.5,0) {};
\node (3) at (1.5,0) {};
\node (4) at (2,0) {};
\end{scope}
\draw (0) -- (-0.5,0.6) -- (2,0.6);
\draw (0) -- (1) -- (2) -- (3) -- (4);
\draw (1) -- (0,0.3) -- (0.5,0.3) -- (2);
\draw (3) -- (1.5,0.3) -- (2,0.3) -- (4);
\node at (1,0.15) {\ldots};
\node at (1,-0.3) {($n$ times)};
\end{tikzpicture}}
~\Lc g_1^{N,m_0}.
\end{equation}
Estimating this error in the weak sense, and using the $\Ld^2$ a priori estimate~\eqref{eq:apriori-GNm-re} on correlations,
we obtain for a test function $h\in C^\infty_c(\R^d)$,
\begin{equation}\label{eq:error-weak}
\<h,E_N\>
\,\lesssim\,t_N N^{-n-\frac12}\Big\|\,{\tiny\begin{tikzpicture}[baseline={([yshift=1ex]current bounding box.center)},scale=0.8]
\begin{scope}[every node/.style={circle,draw,fill=white,inner sep=0pt,minimum size=3pt}]
\node (1) at (0,0) {};
\node (2) at (0.5,0) {};
\node (3) at (1.5,0) {};
\node (4) at (2,0) {};
\node (5) at (2.5,0) {};
\end{scope}
\draw (1) -- (2) -- (3) -- (4) -- (5);
\draw (1) -- (0,0.3) -- (0.5,0.3) -- (2);
\draw (3) -- (1.5,0.3) -- (2,0.3) -- (4);
\draw (0,0.6) -- (2.5,0.6) -- (5);
\node at (1,0.15) {\ldots};
\node at (1,-0.3) {($n$ times)};
\end{tikzpicture}}\,h\Big\|.
\end{equation}
As we shall see in Section~\ref{sec:ren}, each occurrence of $\BOX$ can be controlled by a diffusion in velocity. We may therefore heuristically write
\[{\tiny\begin{tikzpicture}[baseline={([yshift=1ex]current bounding box.center)},scale=0.8]
\begin{scope}[every node/.style={circle,draw,fill=white,inner sep=0pt,minimum size=3pt}]
\node (1) at (0,0) {};
\node (2) at (0.5,0) {};
\node (3) at (1.5,0) {};
\node (4) at (2,0) {};
\node (5) at (2.5,0) {};
\end{scope}
\draw (1) -- (2) -- (3) -- (4) -- (5);
\draw (1) -- (0,0.3) -- (0.5,0.3) -- (2);
\draw (3) -- (1.5,0.3) -- (2,0.3) -- (4);
\draw (0,0.6) -- (2.5,0.6) -- (5);
\node at (1,0.15) {\ldots};
\node at (1,-0.3) {($n$ times)};
\end{tikzpicture}}\,h\,(\alpha,k,v_0,v_1)
\,\approx\,
\bigg[\nabla_{v_0}^2\Big(\tfrac{1+i\alpha}{t_N}+ik\cdot(v_1-v_0)-\tfrac\kappa N\triangle_{v_{[1]}}\Big)^{-1}\bigg]^nM(v_1)k\V(k)\cdot\nabla_{v_0}h(v_0).\]
Evaluating the $\Ld^2$ norm of the resolvents then yields
\[\Big\|{\tiny\begin{tikzpicture}[baseline={([yshift=1ex]current bounding box.center)},scale=0.8]
\begin{scope}[every node/.style={circle,draw,fill=white,inner sep=0pt,minimum size=3pt}]
\node (1) at (0,0) {};
\node (2) at (0.5,0) {};
\node (3) at (1.5,0) {};
\node (4) at (2,0) {};
\node (5) at (2.5,0) {};
\end{scope}
\draw (1) -- (2) -- (3) -- (4) -- (5);
\draw (1) -- (0,0.3) -- (0.5,0.3) -- (2);
\draw (3) -- (1.5,0.3) -- (2,0.3) -- (4);
\draw (0,0.6) -- (2.5,0.6) -- (5);
\node at (1,0.15) {\ldots};
\node at (1,-0.3) {($n$ times)};
\end{tikzpicture}}\,h\,\Big\|
\,\lesssim_n\, t_N^{3n}\|\langle\nabla_{v_0}\rangle^{2n+1}h\|.\]
Substituting this estimate into~\eqref{eq:error-weak}, we conclude that the truncation error in the Dyson expansion can be controlled, at best, on the timescale $t_N\ll N^{1/3}$. In order to reach the thermalization timescale $t_N=N$, we thus need to find another way.

For $\kappa>0$, the preceding estimates can be partially improved, as the velocity diffusion allows one to use hypoelliptic resolvent estimates. However, as we show, this improvement alone remains insufficient to close the argument. More precisely, instead of the following naive bound, which also holds for $\kappa=0$,
\[\|(\tfrac1{t_N}+ik\cdot v-\tfrac\kappa N\triangle_v)^{-1}\|_{L^2(dv)\to L^2(dv)}\,\lesssim\,t_N,\]
one can use Airy-type resolvent estimates,
\begin{equation}\label{eq:airy-est0}
\|(\tfrac1{t_N}+ik\cdot v-\tfrac\kappa N\triangle_v)^{-1}\|_{L^2(dv)\to L^2(dv)}\,\lesssim\,\kappa^{-\frac13}|k|^{-\frac23}N^\frac13,
\end{equation}
see Lemma~\ref{lem:Airy-res} below.
Getting back to~\eqref{eq:error-weak} and using this improved scaling, we obtain for~$t_N=N$ and~$\kappa>0$,
\[\eqref{eq:error-weak}\lesssim_{\kappa,n}t_NN^{-n-\frac12}N^n=\sqrt N.\]
This error bound does not improve for larger~$n$, showing that diagrams of the type~\eqref{eq:err-ex} in the error analysis remain problematic.

\subsection{Renormalization}
Our first key observation is that the worst contributions in the error estimates arise from diagrams such as~\eqref{eq:err-ex} that involve repeated applications of the hat operator $\BOX$.
This should not come as a surprise: this hat operator can be anticipated to play a central role as it precisely encodes the thermalization emerging on the slow timescale $t_N=N$, cf.~\eqref{eq:solve-trunc-hier-re}.
This motivates a renormalization procedure, in which we would resum iterated hat operators.

For illustration, we start by describing this renormalization in case $m_0=2$, that is, when correlations are restricted to two background particles. We then consider the formal geometric series
\begin{multline*}
\WIG\,=\,
{\tiny\begin{tikzpicture}[baseline={([yshift=-1ex]current bounding box.center)},scale=0.8]
\draw (0,0) -- (0.8,0);
\draw (0,0.3) -- (0.8,0.3);
\end{tikzpicture}}
+\big(\tfrac{i}{\sqrt N}\big)^2\Big(
{\tiny\begin{tikzpicture}[baseline={([yshift=-0.7ex]current bounding box.center)},scale=0.8]
\begin{scope}[every node/.style={circle,draw,fill=white,inner sep=0pt,minimum size=3pt}]
\node (1) at (0,0) {};
\node (2) at (0.5,0) {};
\end{scope}
\draw (-0.5,0) -- (1) -- (2) -- (1,0);
\draw (1) -- (0,0.25) -- (0.5,0.25) -- (2);
\draw (-0.5,0.5) -- (1,0.5);
\end{tikzpicture}}
+
{\tiny\begin{tikzpicture}[baseline={([yshift=-1ex]current bounding box.center)},scale=0.8]
\begin{scope}[every node/.style={circle,draw,fill=white,inner sep=0pt,minimum size=3pt}]
\node (1) at (0,0.25) {};
\node (2) at (0.5,0.25) {};
\end{scope}
\draw (-0.5,0.25) -- (1) -- (2) -- (1,0.25);
\draw (1) -- (0,0.5) -- (0.5,0.5) -- (2);
\draw (-0.5,0) -- (1,0);
\end{tikzpicture}}
\Big)
\\
+\big(\tfrac{i}{\sqrt N}\big)^4\Big(
{\tiny\begin{tikzpicture}[baseline={([yshift=-0.9ex]current bounding box.center)},scale=0.8]
\begin{scope}[every node/.style={circle,draw,fill=white,inner sep=0pt,minimum size=3pt}]
\node (1) at (0,0) {};
\node (2) at (0.5,0) {};
\node (3) at (1,0) {};
\node (4) at (1.5,0) {};
\end{scope}
\draw (-0.5,0) -- (1) -- (2) -- (3) -- (4) -- (2,0);
\draw (1) -- (0,0.25) -- (0.5,0.25) -- (2);
\draw (3) -- (1,0.25) -- (1.5,0.25) -- (4);
\draw (-0.5,0.5) -- (2,0.5);
\end{tikzpicture}}
+
{\tiny\begin{tikzpicture}[baseline={([yshift=-1.7ex]current bounding box.center)},scale=0.8]
\begin{scope}[every node/.style={circle,draw,fill=white,inner sep=0pt,minimum size=3pt}]
\node (1) at (0,0) {};
\node (2) at (0.5,0) {};
\node (3) at (1,0.5) {};
\node (4) at (1.5,0.5) {};
\end{scope}
\draw (-0.5,0) -- (1) -- (2) -- (2,0);
\draw (-0.5,0.5) -- (3) -- (4) -- (2,0.5);
\draw (1) -- (0,0.25) -- (0.5,0.25) -- (2);
\draw (3) -- (1,0.75) -- (1.5,0.75) -- (4);
\end{tikzpicture}}
+
{\tiny\begin{tikzpicture}[baseline={([yshift=-1.7ex]current bounding box.center)},scale=0.8]
\begin{scope}[every node/.style={circle,draw,fill=white,inner sep=0pt,minimum size=3pt}]
\node (1) at (0,0.5) {};
\node (2) at (0.5,0.5) {};
\node (3) at (1,0) {};
\node (4) at (1.5,0) {};
\end{scope}
\draw (-0.5,0) -- (3) -- (4) -- (2,0);
\draw (-0.5,0.5) -- (1) -- (2) -- (2,0.5);
\draw (1) -- (0,0.75) -- (0.5,0.75) -- (2);
\draw (3) -- (1,0.25) -- (1.5,0.25) -- (4);
\end{tikzpicture}}
+
{\tiny\begin{tikzpicture}[baseline={([yshift=-1ex]current bounding box.center)},scale=0.8]
\begin{scope}[every node/.style={circle,draw,fill=white,inner sep=0pt,minimum size=3pt}]
\node (1) at (0,0.25) {};
\node (2) at (0.5,0.25) {};
\node (3) at (1,0.25) {};
\node (4) at (1.5,0.25) {};
\end{scope}
\draw (-0.5,0.25) -- (1) -- (2) -- (3) -- (4) -- (2,0.25);
\draw (1) -- (0,0.5) -- (0.5,0.5) -- (2);
\draw (3) -- (1,0.5) -- (1.5,0.5) -- (4);
\draw (-0.5,0) -- (2,0);
\end{tikzpicture}}
\Big)
+\ldots
\end{multline*}
Since the hat operator is comparable to diffusion in velocity, cf.~Section~\ref{sec:ren}, this formal series converges only when applied to functions of exponential type. A more robust approach is to define~$\WIG$ as the modified resolvent
\begin{equation}\label{eq:def-WIG}
\WIG\,:=\,\Big(\tfrac{1+i\alpha}{t_N}+\hat L_1+\tfrac\kappa N\hat D_1+\tfrac1N\BOXd+\tfrac1N\BOXu\Big)^{-1}.
\end{equation}
We show in Section~\ref{sec:ren} that this resolvent is well-defined and bounded on $\Ld^2$, and we refer to it as the ($2$-particle) {\it renormalized propagator}.

In case $m_0>2$, some additional care is required as it is convenient to further renormalize all {\it nested} hat operators, that is, contributions like
\[{\tiny\begin{tikzpicture}[baseline={([yshift=-1ex]current bounding box.center)},scale=0.8]
\begin{scope}[every node/.style={circle,draw,fill=white,inner sep=0pt,minimum size=3pt}]
\node (1) at (0,0.25) {};
\node (2) at (1.5,0.25) {};
\node (3) at (0.5,0.5) {};
\node (4) at (1,0.5) {};
\end{scope}
\draw (-0.5,0.25) -- (1) -- (2) -- (2,0.25);
\draw (1) -- (0,0.5) -- (3) -- (4) -- (1.5,0.5) -- (2);
\draw (3) -- (0.5,0.75) -- (1,0.75) -- (4);
\draw (-0.5,0) -- (2,0);
\end{tikzpicture}}
\qquad
{\tiny\begin{tikzpicture}[baseline={([yshift=-1ex]current bounding box.center)},scale=0.8]
\begin{scope}[every node/.style={circle,draw,fill=white,inner sep=0pt,minimum size=3pt}]
\node (1) at (0,0.25) {};
\node (2) at (1.5,0.25) {};
\node (3) at (0.5,0.25) {};
\node (4) at (1,0.25) {};
\end{scope}
\draw (-0.5,0.25) -- (1) --(3) -- (4) -- (2) -- (2,0.25);
\draw (1) -- (0,0.75) -- (1.5,0.75) -- (2);
\draw (3) -- (0.5,0.5) -- (1,0.5) -- (4);
\draw (-0.5,0) -- (2,0);
\end{tikzpicture}}
\qquad\text{and so on}.\]
For that purpose, we define an iterative renormalization scheme; see Lemma~\ref{lem:box-ren} below for its well-posedness. We emphasize that the renormalization of nested hat operators is not strictly needed in our approach, but it strongly simplifies the combinatorial structure in the error analysis.

\begin{defin}[Renormalization]\label{def:renorm}
Given $m_0\ge1$, the hat propagator $\WIGdot$ on $L^2(\hat\Dd^{m+1})$ is defined iteratively as follows, for all $1\le m\le m_0$:
\begin{enumerate}[---]
\item For $m=m_0$, we define on $\Ld^2(\hat\Dd^{m+1})$,
\[\WIGdot \,:=\,\Big(\tfrac{1+i\alpha}{t_N}+\hat L_m+\tfrac\kappa N\hat D_m\Big)^{-1}\,=\,
{\tiny\begin{tikzpicture}[baseline={([yshift=-.8ex]current bounding box.center)},scale=0.8]
\draw (0,0) -- (0.8,0);
\draw (0,0.45) -- (0.8,0.45);
\draw[dotted] (0.4,0.1) -- (0.4,0.4);
\end{tikzpicture}}\]
\item For $1\le m<m_0$, we define on $\Ld^2(\hat\Dd^{m+1})$,
\[\WIGdot\,:=\,\Big(\tfrac{1+i\alpha}{t_N}+\hat L_{m}+\tfrac\kappa N\hat D_{m}+\tfrac{1}{N}\sum \WIGBOX\Big)^{-1},\]
where we have set
\[\sum\WIGBOX\,:=\,\sum_{j=0}^m\hat S_{j,m+1}^{m,+}\,\WIGdot\,\hat S_{j,m+1}^{m+1,-}\]
in terms of the operator~$\WIGdot$ assumed to be already defined on $\Ld^2(\hat\Dd^{m+2})$.
\end{enumerate}
\end{defin}

With this renormalization procedure,
the Dyson series~\eqref{eq:Dyson} gets drastically simplified: up to replacing propagators by their renormalized versions, we can remove all the diagrams involving hat operators. For $t_N=N$, this means
\begin{multline}\label{eq:Dyson-re}
\big(1+i\alpha-\triangle_{v_0}+\BOX\big)\Lc g_0^{N,m_0}
\,=\,\mathfrak g
+
\tfrac1N\Big(
{\tiny\begin{tikzpicture}[baseline={([yshift=-1ex]current bounding box.center)},scale=0.8]
\begin{scope}[every node/.style={circle,draw,fill=white,inner sep=0pt,minimum size=3pt}]
\node (1) at (0,0) {};
\node (2) at (0.5,0) {};
\node (3) at (1,0) {};
\node (4) at (1.5,0) {};
\end{scope}
\draw[decorate,decoration={snake,amplitude=1pt,segment length=3pt}] (1) -- (2);
\draw[decorate,decoration={snake,amplitude=1pt,segment length=3pt}] (2) -- (3);
\draw[decorate,decoration={snake,amplitude=1pt,segment length=3pt}] (3)-- (4);
\draw (1) -- (0,0.3);
\draw (2) -- (0.5,0.6);
\draw (3) -- (1,0.3);
\draw (4) -- (1.5,0.6);
\draw[decorate,decoration={snake,amplitude=1pt,segment length=3pt}] (0,0.3) -- (1,0.3);
\draw[decorate,decoration={snake,amplitude=1pt,segment length=3pt}] (0.5,0.6) -- (1.5,0.6);
\end{tikzpicture}}
+
{\tiny\begin{tikzpicture}[baseline={([yshift=-1ex]current bounding box.center)},scale=0.8]
\begin{scope}[every node/.style={circle,draw,fill=white,inner sep=0pt,minimum size=3pt}]
\node (1) at (0,0) {};
\node (2) at (0.5,0.3) {};
\node (3) at (1,0.6) {};
\node (4) at (1.5,0) {};
\end{scope}
\draw[decorate,decoration={snake,amplitude=1pt,segment length=3pt}] (1) -- (4);
\draw (1) -- (0,0.3);
\draw[decorate,decoration={snake,amplitude=1pt,segment length=3pt}] (0,0.3)-- (2) -- (1,0.3);
\draw[decorate,decoration={snake,amplitude=1pt,segment length=3pt}] (0.5,0.6)-- (3) -- (1.5,0.6);
\draw (2) -- (0.5,0.6);
\draw (1.5,0.6) -- (4);
\draw (2) -- (0.5,0.6);
\draw (1,0.3) -- (3);
\end{tikzpicture}}
+
{\tiny\begin{tikzpicture}[baseline={([yshift=-1ex]current bounding box.center)},scale=0.8]
\begin{scope}[every node/.style={circle,draw,fill=white,inner sep=0pt,minimum size=3pt}]
\node (1) at (0,0) {};
\node (2) at (0.5,0.6) {};
\node (3) at (1,0) {};
\node (4) at (1.5,0) {};
\end{scope}
\draw [decorate,decoration={snake,amplitude=1pt,segment length=3pt}] (1)-- (3);
\draw [decorate,decoration={snake,amplitude=1pt,segment length=3pt}] (3) -- (4);
\draw (1) -- (0,0.6);
\draw [decorate,decoration={snake,amplitude=1pt,segment length=3pt}] (0,0.6) -- (2) -- (1.5,0.6);
\draw (1.5,0.6) -- (4);
\draw (2) -- (0.5,0.3);
\draw [decorate,decoration={snake,amplitude=1pt,segment length=3pt}] (0.5,0.3) -- (1,0.3);
\draw (1,0.3) -- (3);
\end{tikzpicture}}
\\
+
{\tiny\begin{tikzpicture}[baseline={([yshift=-1ex]current bounding box.center)},scale=0.8]
\begin{scope}[every node/.style={circle,draw,fill=white,inner sep=0pt,minimum size=3pt}]
\node (1) at (0,0) {};
\node (2) at (-0.5,0.6) {};
\node (3) at (-1,0) {};
\node (4) at (-1.5,0) {};
\end{scope}
\draw [decorate,decoration={snake,amplitude=1pt,segment length=3pt}] (1)-- (3);
\draw [decorate,decoration={snake,amplitude=1pt,segment length=3pt}]
(3) -- (4);
\draw (1) -- (0,0.6);
\draw [decorate,decoration={snake,amplitude=1pt,segment length=3pt}] (0,0.6) -- (2) -- (-1.5,0.6);
\draw (-1.5,0.6) -- (4);
\draw (2) -- (-0.5,0.3);
\draw [decorate,decoration={snake,amplitude=1pt,segment length=3pt}] (-0.5,0.3) -- (-1,0.3);
\draw (-1,0.3) -- (3);
\end{tikzpicture}}
+{\tiny\begin{tikzpicture}[baseline={([yshift=-1ex]current bounding box.center)},scale=0.8]
\begin{scope}[every node/.style={circle,draw,fill=white,inner sep=0pt,minimum size=3pt}]
\node (1) at (0,0) {};
\node (2) at (0.5,0.3) {};
\node (3) at (1,0) {};
\node (4) at (1.5,0) {};
\end{scope}
\draw [decorate,decoration={snake,amplitude=1pt,segment length=3pt}] (1) -- (3);
\draw [decorate,decoration={snake,amplitude=1pt,segment length=3pt}] (3) -- (4);
\draw (1) -- (0,0.3);
\draw [decorate,decoration={snake,amplitude=1pt,segment length=3pt}] (0,0.3) -- (2) -- (1,0.3);
\draw (1,0.3) -- (3);
\draw (2) -- (0.5,0.6);
\draw [decorate,decoration={snake,amplitude=1pt,segment length=3pt}] (0.5,0.6) -- (1.5,0.6);
\draw (1.5,0.6) -- (4);
\end{tikzpicture}}
+{\tiny\begin{tikzpicture}[baseline={([yshift=-1ex]current bounding box.center)},scale=0.8]
\begin{scope}[every node/.style={circle,draw,fill=white,inner sep=0pt,minimum size=3pt}]
\node (1) at (0,0) {};
\node (2) at (-0.5,0.3) {};
\node (3) at (-1,0) {};
\node (4) at (-1.5,0) {};
\end{scope}
\draw [decorate,decoration={snake,amplitude=1pt,segment length=3pt}] (1) -- (3);
\draw [decorate,decoration={snake,amplitude=1pt,segment length=3pt}] (3) -- (4);
\draw (1) -- (0,0.25);
\draw [decorate,decoration={snake,amplitude=1pt,segment length=3pt}] (0,0.3) -- (2) -- (-1,0.3);
\draw (-1,0.3) -- (3);
\draw (2) -- (-0.5,0.6);
\draw [decorate,decoration={snake,amplitude=1pt,segment length=3pt}] (-0.5,0.6) -- (-1.5,0.6);
\draw (-1.5,0.6) -- (4);
\end{tikzpicture}}
\Big)\Lc g_0^{N,m_0}+\ldots
\end{multline}
and it remains to show that all the diagrams in the right-hand side are $O(\frac1N)$ error terms.

\subsection{Error estimates: `ansatz' approach}
Instead of working directly with the renormalized expansion~\eqref{eq:Dyson-re} and estimating the truncation error via Duhamel formula, it is often advantageous to view the remainder as satisfying itself a PDE. Equivalently, this means interpreting the truncated Dyson expansion as an ansatz providing an {approximate solution} of the original problem, with the error typically appearing as lower-order source terms represented by explicit diagrams. This perspective allows to leverage PDE techniques for the original problem to obtain neat error estimates. See e.g.~\cite[Section~3.5]{Deng-Hani} and references therein for a discussion in the context of weak turbulence.

In this work, we adopt a similar viewpoint, interpreting the truncated renormalized Dyson expansion~\eqref{eq:Dyson-re} as an approximate solution of the hierarchy~\eqref{eq:hierarchy}. The specific unitary structure of the latter can then be exploited and is crucial for deriving neat error estimates. This approach is developed in detail in Section~\ref{sec:ansatz}. We emphasize that the ansatz must then be constructed for the entire sequence of correlation functions, requiring consideration of the Dyson expansion for each of them, not just for~$g_0^{N,m_0}$ as we focus on in this introduction.

\subsection{Toolbox for diagrammatic estimates}
To estimate the iterated resolvents appearing in the diagrammatic expressions, we systematically rely on two key analytical tools:
\begin{enumerate}[---]
\item\emph{Phase mixing:}
Resolvent estimates in operator norm yield $N$-dependent bounds due to filamentation effects. However, since annihilation operators involve velocity averages, phase mixing can be exploited to obtain uniform-in-$N$ estimates. As usual, this mechanism is conveniently implemented via contour deformations. For instance, deforming the integration path as $v\mapsto v-i\hat k$ with $\hat k:=\tfrac{k}{|k|}$, we obtain
\begin{equation}\label{eq:deform-0}
\hspace{0.8cm}\Big|\int_{\R^d}(\tfrac1{t_N}+ik\cdot v-\tfrac\kappa N\triangle_v)^{-1}M(v)\,dv\Big|
\,=\,\Big|\int_{\R^d}(\tfrac1{t_N}+|k|+ik\cdot v-\tfrac\kappa N\triangle_v)^{-1}M(v-i\hat k)\,dv\Big|\,\lesssim\,|k|^{-1}.
\end{equation}
Such contour deformations remain valid after renormalization; see Section~\ref{sec:renorm et complex}.
\smallskip\item\emph{Hypoelliptic estimates:}
As already noted in~\eqref{eq:airy-est0}, for $\kappa>0$, the $O(\frac1N)$ velocity diffusion included in the model allows one to exploit hypoelliptic estimates and the resulting enhanced dissipation. At the level of resolvent estimates, this yields the bound
\begin{equation}\label{eq:airy-est1}
\|(\tfrac1{t_N}+ik\cdot v-\tfrac\kappa N\triangle_v)^{-1}\|_{L^2(dv)\to L^2(dv)}\,\lesssim_\kappa\,|k|^{-\frac23}N^\frac13,
\end{equation}
see Lemma~\ref{lem:Airy-res}, which constitutes a substantial improvement over the naive bound $O(t_N)$ that would be optimal for $\kappa=0$. This also remains valid after renormalization; see Section~\ref{sec:PDE-est}.
\end{enumerate}
If no $O(\frac1N)$ velocity diffusion was included in the model (that is, if $\kappa=0$), the hat operator --- being comparable to a diffusion in velocity --- would still make the renormalization procedure effectively add such a diffusion to the free transport operator. One would then expect hypoelliptic estimates to emerge for the renormalized propagators even in this case. Although it is usual for renormalized propagators to satisfy improved bounds, this appears to be the first instance where the improvement arises via regularity theory. Nevertheless, the hat operator satisfies only $0\le\BOX\lesssim-\triangle_v$ and is not uniformly elliptic in velocity (notably because of its time dependence). As a consequence, its precise hypoelliptic properties remain difficult to characterize and exploit. For the present analysis, we therefore restrict to the simplified setting in which an explicit $O(\frac1N)$ velocity diffusion is included in the model~(that is,~$\kappa>0$); see assumption~(H2).

\subsection{Limitations of the strategy}\label{sec:limitation}
When truncating the formal renormalized Dyson expansion~\eqref{eq:Dyson-re} with $t_N=N$, the remainder involves a variety of diagrams, such as
\begin{equation}\label{eq:FN-diag-rem}
F_N\,:=\,
N\big(\tfrac{i}{\sqrt N}\big)^{2\ell+m+2}~{\tiny\begin{tikzpicture}[baseline={([yshift=-3ex]current bounding box.center)},scale=0.8]
\begin{scope}[every node/.style={circle,draw,fill=white,inner sep=0pt,minimum size=3pt}]
\node (0) at (-0.5,0) {};
\node (1) at (0.5,0) {};
\node (2) at (0,0.6) {};
\node (3) at (2,0) {};
\node (4) at (1.5,0.6) {};
\node (5) at (2.5,0) {};
\node (6) at (-1,-0.5) {};
\node (7) at (-3,-1.3) {};
\end{scope}
\draw[decorate,decoration={snake,amplitude=1pt,segment length=3pt}] (0) -- (1);
\draw[decorate,decoration={snake,amplitude=1pt,segment length=3pt}] (1) -- (3);
\draw[decorate,decoration={snake,amplitude=1pt,segment length=3pt}] (3) -- (5);
\draw (0) -- (-0.5,0.6);
\draw[decorate,decoration={snake,amplitude=1pt,segment length=3pt}] (-0.5,0.6) -- (2);
\draw (2) -- (0,0.3);
\draw[decorate,decoration={snake,amplitude=1pt,segment length=3pt}] (0,0.3) -- (0.5,0.3);
\draw (0.5,0.3) -- (1);
\draw (4) -- (1.5,0.3);
\draw[decorate,decoration={snake,amplitude=1pt,segment length=3pt}] (1.5,0.3) -- (2,0.3);
\draw (2,0.3) -- (3);
\draw[decorate,decoration={snake,amplitude=1pt,segment length=3pt}] (2) -- (4);
\draw[decorate,decoration={snake,amplitude=1pt,segment length=3pt}] (4) -- (2.5,0.6);
\draw (2.5,0.6) -- (5);
\draw[decorate,decoration={snake,amplitude=1pt,segment length=3pt}] (0) -- (-1,0);
\draw (-1,0) -- (6);
\draw[decorate,decoration={snake,amplitude=1pt,segment length=3pt}] (2.5,-0.5) -- (6);
\draw[decorate,decoration={snake,amplitude=1pt,segment length=3pt}] (6) -- (-1.5,-0.5);
\draw (-1.5,-0.5) -- (-1.5,-0.65);
\draw[dotted] (-1.5,-0.5) -- (-1.5,-1);
\draw (7) -- (-3,-1.1);
\draw[dotted] (-3,-1.1) -- (-3,-0.75);
\draw[decorate,decoration={snake,amplitude=1pt,segment length=3pt}] (7) -- (2.5,-1.3);
\node at (1,0.3) {\ldots};
\node at (1,0.8) {($\ell$ times)};
\node at (-2.3,-0.5) {($m$ times)};
\end{tikzpicture}}~\Lc g^{N,m_0}_m.
\end{equation}
Estimating this error term in the weak sense and using the a priori $\Ld^2$ bounds~\eqref{eq:apriori-GNm-re} on correlations,
we find for a test function $h\in C^\infty_c(\R^d)$,
\begin{equation}\label{eq:bound-weak-FN}
\<h,F_N\>\,\lesssim\,N^{-\ell-\frac m2}\,\bigg\|~{\tiny\begin{tikzpicture}[baseline={([yshift=-3ex]current bounding box.center)},scale=0.8]
\begin{scope}[every node/.style={circle,draw,fill=white,inner sep=0pt,minimum size=3pt}]
\node (0) at (-0.5,0) {};
\node (1) at (0,0) {};
\node (2) at (0.5,0.6) {};
\node (3) at (1.5,0) {};
\node (4) at (2,0.6) {};
\node (5) at (2.5,0) {};
\node (6) at (3,-0.5) {};
\node (7) at (5,-1.3) {};
\end{scope}
\draw[decorate,decoration={snake,amplitude=1pt,segment length=3pt}] (0) -- (1);
\draw[decorate,decoration={snake,amplitude=1pt,segment length=3pt}] (1) -- (3);
\draw[decorate,decoration={snake,amplitude=1pt,segment length=3pt}] (3) -- (5);
\draw (0) -- (-0.5,0.6);
\draw[decorate,decoration={snake,amplitude=1pt,segment length=3pt}] (-0.5,0.6) -- (2);
\draw (1) -- (0,0.3);
\draw[decorate,decoration={snake,amplitude=1pt,segment length=3pt}] (0,0.3) -- (0.5,0.3);
\draw (0.5,0.3) -- (2);
\draw (3) -- (1.5,0.3);
\draw[decorate,decoration={snake,amplitude=1pt,segment length=3pt}] (1.5,0.3) -- (2,0.3);
\draw (2,0.3) -- (4);
\draw[decorate,decoration={snake,amplitude=1pt,segment length=3pt}] (2) -- (4);
\draw[decorate,decoration={snake,amplitude=1pt,segment length=3pt}] (4) -- (2.5,0.6);
\draw (2.5,0.6) -- (5);
\draw[decorate,decoration={snake,amplitude=1pt,segment length=3pt}] (5) -- (3,0);
\draw (3,0) -- (6);
\draw[decorate,decoration={snake,amplitude=1pt,segment length=3pt}] (-0.5,-0.5) -- (6);
\draw[decorate,decoration={snake,amplitude=1pt,segment length=3pt}] (6) -- (3.5,-0.5);
\draw (3.5,-0.5) -- (3.5,-0.65);
\draw[dotted] (3.5,-0.5) -- (3.5,-1);
\draw (7) -- (5,-1.1);
\draw[dotted] (5,-1.1) -- (5,-0.75);
\draw[decorate,decoration={snake,amplitude=1pt,segment length=3pt}] (7) -- (-0.5,-1.3);
\node at (1,0.3) {\ldots};
\node at (1,0.8) {($\ell$ times)};
\node at (4.3,-0.5) {($m$ times)};
\end{tikzpicture}}\,h\bigg\|.
\end{equation}
Compared to hat operators, a key feature of the iterated subdiagram
$\begin{tikzpicture}[baseline={([yshift=-0.8ex]current bounding box.center)},scale=0.8]
\begin{scope}[every node/.style={circle,draw,fill=white,inner sep=0pt,minimum size=3pt}]
\node (0) at (0,0) {};
\node (1) at (0.5,0.6) {};
\end{scope}
\draw[decorate,decoration={snake,amplitude=1pt,segment length=3pt}] (0) -- (0.5,0);
\draw[decorate,decoration={snake,amplitude=1pt,segment length=3pt}] (1) -- (0,0.6);
\draw[decorate,decoration={snake,amplitude=1pt,segment length=3pt}] (0,0.3) -- (0.5,0.3);
\draw (0) -- (0,0.3);
\draw (1) -- (0.5,0.3);
\end{tikzpicture}$
is that it mixes momentum variables: omitting the renormalization for simplicity, we indeed find
\begin{multline*}
\begin{tikzpicture}[baseline={([yshift=-0.8ex]current bounding box.center)},scale=0.8]
\begin{scope}[every node/.style={circle,draw,fill=white,inner sep=0pt,minimum size=3pt}]
\node (0) at (0,0) {};
\node (1) at (0.5,0.6) {};
\end{scope}
\draw (0) -- (0.5,0);
\draw (1) -- (0,0.6);
\draw (0,0.3) -- (0.5,0.3);
\draw (0) -- (0,0.3);
\draw (1) -- (0.5,0.3);
\end{tikzpicture}g(\hat z_0,\hat z_1)
\,=\,
-\int_{\Dd}\sqrt M(v_2)k_2\V(k_2)\cdot\nabla_{v_0}\Big(\tfrac1{t_N}+i(k_0-k_2)\cdot v_0+ik_1\cdot v_1+ik_2\cdot v_2-\tfrac\kappa N\triangle_{v_{[2]}}\Big)^{-1}\\
\times\sqrt M(v_2)k_2\V(k_2)\cdot\nabla_{v_1}g\Big((k_0-k_2,v_0),(k_1+k_2,v_1)\Big)\,d\hat z_2,
\end{multline*}
and thus, applying contour deformation $v_2\mapsto v_2-i\hat k_2$ with $\hat k_2=\frac{k_2}{|k_2|}$,
\begin{equation*}
\big|\begin{tikzpicture}[baseline={([yshift=-0.8ex]current bounding box.center)},scale=0.8]
\begin{scope}[every node/.style={circle,draw,fill=white,inner sep=0pt,minimum size=3pt}]
\node (0) at (0,0) {};
\node (1) at (0.5,0.6) {};
\end{scope}
\draw (0) -- (0.5,0);
\draw (1) -- (0,0.6);
\draw (0,0.3) -- (0.5,0.3);
\draw (0) -- (0,0.3);
\draw (1) -- (0.5,0.3);
\end{tikzpicture}g(\hat z_0,\hat z_1)\big|
\,\lesssim\,
\int_{\R^d}\langle k_2\rangle\V(k_2)^2\,\Big|\nabla^2_{v_0v_1}g\Big((k_0-k_2,v_0),(k_1+k_2,v_1)\Big)\Big|\,dk_2.
\end{equation*}
Returning to~\eqref{eq:bound-weak-FN} and using such estimates, we may heuristically write
\begin{multline*}
{\tiny\begin{tikzpicture}[baseline={([yshift=-3ex]current bounding box.center)},scale=0.8]
\begin{scope}[every node/.style={circle,draw,fill=white,inner sep=0pt,minimum size=3pt}]
\node (0) at (-0.5,0) {};
\node (1) at (0,0) {};
\node (2) at (0.5,0.6) {};
\node (3) at (1.5,0) {};
\node (4) at (2,0.6) {};
\node (5) at (2.5,0) {};
\end{scope}
\draw (0) -- (1);
\draw (1) -- (3);
\draw (3) -- (5);
\draw (0) -- (-0.5,0.6);
\draw (-0.5,0.6) -- (2);
\draw (1) -- (0,0.3);
\draw (0,0.3) -- (0.5,0.3);
\draw (0.5,0.3) -- (2);
\draw (3) -- (1.5,0.3);
\draw (1.5,0.3) -- (2,0.3);
\draw (2,0.3) -- (4);
\draw (2) -- (4);
\draw (4) -- (2.5,0.6);
\draw (2.5,0.6) -- (5);
\node at (1,0.3) {\ldots};
\node at (1,0.8) {($\ell$ times)};
\end{tikzpicture}}\,h(k_0,v_0)
\,\approx\,
\int_{(\R^d)^{\ell+1}\times\R^d}\sqrt M(v_1)k_1\V(k_1)\cdot\nabla_{v_0}\Big(\tfrac1{t_N}+i(k_0-k_1)\cdot v_0+ik_1\cdot v_1-\tfrac\kappa N\triangle_{v_{[1]}}\Big)^{-1}\\
\hspace{5cm}\times \prod_{j=2}^{\ell+1}\langle k_j\rangle\V(k_j)^2\,\nabla_{v_0v_1}^2\Big(\tfrac1{t_N}+i(k_0-\bar k_j)\cdot v_0+i\bar k_j\cdot v_1-\tfrac\kappa N\triangle_{v_{[1]}}\Big)^{-1}\\
\times \sqrt M(v_1)k_1\V(k_1)\cdot\nabla_{v_0} h(k_0,v_0)\,dk_1\ldots dk_{\ell+1}dv_1,
\end{multline*}
with the short-hand notation $\bar k_j:=\sum_{l=1}^jk_l$.
Since this expression involves an average of resolvents with respect to $v_1$ and since hypoelliptic resolvent estimates would not provide sufficiently good bounds here, we have to attempt another contour deformation. To do so, we need to find a suitable direction to deform: given $k_1,\ldots,k_{\ell+1}$, if there exists $\nu\in\Sp^{d-1}$ such that $\nu\cdot \bar k_j>0$ for all $j$, then the deformation $v_1\mapsto v_1-i\nu$ would yield an $O(1)$ bound. Yet, by Wendel's theorem, such a direction $\nu$ exists for almost all $k_1,\ldots,k_{\ell+1}$ only provided that the space dimension is $d>\ell$. We then get formally
\[\Big|{\tiny\begin{tikzpicture}[baseline={([yshift=-3ex]current bounding box.center)},scale=0.8]
\begin{scope}[every node/.style={circle,draw,fill=white,inner sep=0pt,minimum size=3pt}]
\node (0) at (-0.5,0) {};
\node (1) at (0,0) {};
\node (2) at (0.5,0.6) {};
\node (3) at (1.5,0) {};
\node (4) at (2,0.6) {};
\node (5) at (2.5,0) {};
\end{scope}
\draw (0) -- (1);
\draw (1) -- (3);
\draw (3) -- (5);
\draw (0) -- (-0.5,0.6);
\draw (-0.5,0.6) -- (2);
\draw (1) -- (0,0.3);
\draw (0,0.3) -- (0.5,0.3);
\draw (0.5,0.3) -- (2);
\draw (3) -- (1.5,0.3);
\draw (1.5,0.3) -- (2,0.3);
\draw (2,0.3) -- (4);
\draw (2) -- (4);
\draw (4) -- (2.5,0.6);
\draw (2.5,0.6) -- (5);
\node at (1,0.3) {\ldots};
\node at (1,0.8) {($\ell$ times)};
\end{tikzpicture}}\,h(k_0,v_0)\Big|\,\lesssim\,\|\nabla_{v_0}^{\ell+2}h\|_{L^\infty}.\]
With this estimate at hand, let us now return to the estimation of the error term~\eqref{eq:bound-weak-FN}. As for the last~$m$ creation operators in the diagram we can only rely on hypoelliptic resolvent estimates~\eqref{eq:airy-est1}, we obtain, provided $d>\ell$,
\[\<h,F_N\>\,\lesssim\,N^{-\ell-\frac m2}N^\frac{2m+\ell+1}3\,=\,N^{-\frac16(4\ell-m-2)},\]
which becomes $o(1)$ whenever $\ell>\frac m4+\frac12$.
Since diagrams of the form~\eqref{eq:FN-diag-rem} appear among remainder terms for all $m\le m_0$, while $\ell$ depends on how far the Dyson series is expanded, we conclude that remainder terms cannot be made small unless the space dimension satisfies $d>\frac{m_0}4$. In other words, the space dimension must be sufficiently large relative to the truncation parameter $m_0$ in the hierarchy. This explains the main limitation appearing in Theorem~\ref{th:main}.

\subsection*{Notation}
\begin{enumerate}[---]
\item We denote by $C\ge1$ any constant that only depends on $d,\beta,m_0$, and on controlled norms of~$\Vc$. We use the notation $\lesssim$ for $\le C\times$ up to such a multiplicative constant $C$, and we write $\ll$ for $\le C\times$ up to a sufficiently large constant $C$. We add subscripts to $C,\lesssim,\ll$ to indicate dependence on other parameters.
\item For any $m$, we denote by $\langle\cdot,\cdot\rangle$ and $\|\cdot\|$ the scalar product and norm on $\Ld^2(\hat\Dd^m)$. We add a subscript `$k$' and write $\langle\cdot,\cdot\rangle_k$ and $\|\cdot\|_k$ for the corresponding scalar product and norm on $\Ld^2((\R^d)^m)$ with fixed momentum variables. We further write $\3\cdot\3$ for the norm on $\Ld^2(\R\times\hat\Dd^m)$, including integration over the Laplace variable.
\item We use the short-hand notation $[a]=\{0,\ldots,a\}$ and $[a,b]=\{a,\ldots,b\}$ for integers $0\le a\le b$, and we set~$z_A=(z_{i_1},\ldots,z_{i_k})$ for an index subset $A=\{i_1,\ldots,i_k\}$.
\end{enumerate}

\section{Construction of an ansatz}\label{sec:ansatz}
As explained, instead of working directly with the formal renormalized Dyson expansion~\eqref{eq:Dyson-re} and estimating the truncation error, we interpret the truncated renormalized expansion as an ansatz defining an approximate solution of the hierarchy~\eqref{eq:hierarchy}, which is particularly convenient to derive error estimates. In this approach, an ansatz must be constructed for all correlation functions, thus requiring to consider the corresponding Dyson expansion for each. For instance, for $m=1$, the expansion reads
\begin{multline}\label{eq:Dyson-ter}
\Lc g_1^{N,m_0}
\,=\,\big(\tfrac{i}{\sqrt N}\big)\,{\tiny\begin{tikzpicture}[baseline={([yshift=-1ex]current bounding box.center)},scale=0.8]
\begin{scope}[every node/.style={circle,draw,fill=white,inner sep=0pt,minimum size=3pt}]
    \node (1) at (0.8,0) {};
\end{scope}
\draw[decorate,decoration={snake,amplitude=1pt,segment length=3pt}] (0,0.3) -- (0.8,0.3);
\draw[decorate,decoration={snake,amplitude=1pt,segment length=3pt}] (0,0) -- (1);
\draw (1) -- (0.8,0.3);
\end{tikzpicture}}\,\Lc g_0^{N,m_0}\\
+
\big(\tfrac{i}{\sqrt N}\big)^3\Big(
{\tiny\begin{tikzpicture}[baseline={([yshift=-1ex]current bounding box.center)},scale=0.8]
\begin{scope}[every node/.style={circle,draw,fill=white,inner sep=0pt,minimum size=3pt}]
\node (2) at (0.5,0) {};
\node (3) at (1,0) {};
\node (4) at (1.5,0) {};
\end{scope}
\draw[decorate,decoration={snake,amplitude=1pt,segment length=3pt}] (0,0) -- (2);
\draw[decorate,decoration={snake,amplitude=1pt,segment length=3pt}] (2) -- (3);
\draw[decorate,decoration={snake,amplitude=1pt,segment length=3pt}] (3)-- (4);
\draw (2) -- (0.5,0.6);
\draw (3) -- (1,0.3);
\draw (4) -- (1.5,0.6);
\draw[decorate,decoration={snake,amplitude=1pt,segment length=3pt}] (0,0.3) -- (1,0.3);
\draw[decorate,decoration={snake,amplitude=1pt,segment length=3pt}] (0.5,0.6) -- (1.5,0.6);
\end{tikzpicture}}
+
{\tiny\begin{tikzpicture}[baseline={([yshift=-1ex]current bounding box.center)},scale=0.8]
\begin{scope}[every node/.style={circle,draw,fill=white,inner sep=0pt,minimum size=3pt}]
\node (2) at (0.5,0.3) {};
\node (3) at (1,0.6) {};
\node (4) at (1.5,0) {};
\end{scope}
\draw[decorate,decoration={snake,amplitude=1pt,segment length=3pt}] (0,0) -- (4);
\draw[decorate,decoration={snake,amplitude=1pt,segment length=3pt}] (0,0.3)-- (2) -- (1,0.3);
\draw[decorate,decoration={snake,amplitude=1pt,segment length=3pt}] (0.5,0.6)-- (3) -- (1.5,0.6);
\draw (2) -- (0.5,0.6);
\draw (1.5,0.6) -- (4);
\draw (2) -- (0.5,0.6);
\draw (1,0.3) -- (3);
\end{tikzpicture}}
+
{\tiny\begin{tikzpicture}[baseline={([yshift=-1ex]current bounding box.center)},scale=0.8]
\begin{scope}[every node/.style={circle,draw,fill=white,inner sep=0pt,minimum size=3pt}]
\node (2) at (0.5,0.6) {};
\node (3) at (1,0) {};
\node (4) at (1.5,0) {};
\end{scope}
\draw [decorate,decoration={snake,amplitude=1pt,segment length=3pt}] (0,0)-- (3);
\draw [decorate,decoration={snake,amplitude=1pt,segment length=3pt}] (3) -- (4);
\draw [decorate,decoration={snake,amplitude=1pt,segment length=3pt}] (0,0.6) -- (2) -- (1.5,0.6);
\draw (1.5,0.6) -- (4);
\draw (2) -- (0.5,0.3);
\draw [decorate,decoration={snake,amplitude=1pt,segment length=3pt}] (0.5,0.3) -- (1,0.3);
\draw (1,0.3) -- (3);
\end{tikzpicture}}
+
{\tiny\begin{tikzpicture}[baseline={([yshift=-1ex]current bounding box.center)},scale=0.8]
\begin{scope}[every node/.style={circle,draw,fill=white,inner sep=0pt,minimum size=3pt}]
\node (1) at (0,0) {};
\node (2) at (-0.5,0.6) {};
\node (3) at (-1,0) {}; 
\end{scope}
\draw [decorate,decoration={snake,amplitude=1pt,segment length=3pt}] (1)-- (3);
\draw [decorate,decoration={snake,amplitude=1pt,segment length=3pt}] (3) -- (-1.5,0);
\draw (1) -- (0,0.6);
\draw [decorate,decoration={snake,amplitude=1pt,segment length=3pt}] (0,0.6) -- (2) -- (-1.5,0.6);
\draw (2) -- (-0.5,0.3);
\draw [decorate,decoration={snake,amplitude=1pt,segment length=3pt}] (-0.5,0.3) -- (-1,0.3);
\draw (-1,0.3) -- (3);
\end{tikzpicture}}
+
{\tiny\begin{tikzpicture}[baseline={([yshift=-1ex]current bounding box.center)},scale=0.8]
\begin{scope}[every node/.style={circle,draw,fill=white,inner sep=0pt,minimum size=3pt}]
\node (2) at (0.5,0.3) {};
\node (3) at (1,0) {};
\node (4) at (1.5,0) {};
\end{scope}
\draw [decorate,decoration={snake,amplitude=1pt,segment length=3pt}] (0,0) -- (0,0);
\draw [decorate,decoration={snake,amplitude=1pt,segment length=3pt}] (0,0) -- (3);
\draw [decorate,decoration={snake,amplitude=1pt,segment length=3pt}] (3) -- (4);
\draw [decorate,decoration={snake,amplitude=1pt,segment length=3pt}] (0,0.3) -- (2) -- (1,0.3);
\draw (1,0.3) -- (3);
\draw (2) -- (0.5,0.6);
\draw [decorate,decoration={snake,amplitude=1pt,segment length=3pt}] (0.5,0.6) -- (1.5,0.6);
\draw (1.5,0.6) -- (4);
\end{tikzpicture}}
+
{\tiny\begin{tikzpicture}[baseline={([yshift=-1ex]current bounding box.center)},scale=0.8]
\begin{scope}[every node/.style={circle,draw,fill=white,inner sep=0pt,minimum size=3pt}]
\node (1) at (0,0) {};
\node (2) at (-0.5,0.3) {};
\node (3) at (-1,0) {};
\end{scope}
\draw [decorate,decoration={snake,amplitude=1pt,segment length=3pt}] (1) -- (3);
\draw [decorate,decoration={snake,amplitude=1pt,segment length=3pt}] (3) -- (-1.5,0);
\draw (1) -- (0,0.3);
\draw [decorate,decoration={snake,amplitude=1pt,segment length=3pt}] (0,0.3) -- (2) -- (-1,0.3);
\draw (-1,0.3) -- (3);
\draw (2) -- (-0.5,0.6);
\draw [decorate,decoration={snake,amplitude=1pt,segment length=3pt}] (-0.5,0.6) -- (-1.5,0.6);
\end{tikzpicture}}
\Big)\Lc g_0^{N,m_0} +\ldots
\end{multline} 
The construction of a suitable ansatz follows the following strategy:
\begin{enumerate}[---]
\item For the tagged particle density, we simply choose $\tilde{g}_0^{N,m_0}:=\tilde{g}_0^N$ as the solution of the truncated non-Markovian equation~\eqref{eq:solve-trunc-hier} obtained in the formal derivation, independently of the truncation parameter $m_0$.
\smallskip\item For correlation functions, $1\le m\le m_0$, the ansatz $\tilde g_m^{N,m_0}$ is defined by truncating the formal Dyson expansion (e.g.~\eqref{eq:Dyson-ter} for $m=1$) and replacing all occurrences of the tagged particle density~${g}_0^{N,m_0}$ with its ansatz $\tilde{g}_0^{N}$. By including a sufficient number of diagrams, depending on $m_0$, this procedure is expected to yield a good approximation of $g_m^{N,m_0}$.
\end{enumerate}

\subsection{Ansatz for tagged particle density}
For the tagged particle density $g_0^{N,m_0}$, taking inspiration from~\eqref{eq:solve-trunc-hier}, we define $\tilde g_0^{N,m_0}:=\tilde g_0^{N}$ as the solution of the following equation,
\begin{equation*}
\left\{\begin{array}{l}
(\partial_\tau+\tfrac{\kappa t_N}N\hat D_0)\tilde g_0^{N}\,=\,
-\big(\frac{t_N}{\sqrt N}\big)^2\int_0^{\tau}\hat S_{0,1}^{0,+} e^{-t_N(\tau-\tau')(i\hat L_1+\frac\kappa N\hat D_1)} \hat S^{1,-}_{0,1}\tilde g_0^{N}(\tau')\,d\tau',\\[1mm]
\tilde g_0^N|_{\tau=0}=\mathfrak g.
\end{array}\right.
\end{equation*}
We emphasize that this choice is independent of the truncation parameter $m_0$.
Applying Laplace transform, this equation reads
\begin{equation}\label{eq:def-tilg0-re}
\Big(1+i\alpha+\tfrac{\kappa t_N}N\hat D_0+\tfrac{t_N}N\BOX\Big)\Lc\tilde g_0^{N}
= \mathfrak g.
\end{equation}
Well-posedness and regularity estimates for $\tilde g_0^{N}$ easily follow by noting that the hat operator is nonnegative and is controlled by diffusion in velocity; see Section~\ref{sec:g0} below.

\subsection{Ansatz for correlations}
The suitable choice of an ansatz for correlations $\{\tilde g_m^{N,m_0}\}_{1\le m\le m_0}$ is more delicate and will depend on the value of the truncation parameter $m_0$.
We start by introducing some combinatorial structure associated with the sequences of creation and annihilation operators that appear in the renormalized Dyson series.
\begin{enumerate}[---]
\item Informally, we define \emph{abstracts} as finite sequences of $\pm1$, where $+1$ and $-1$ correspond to creation and annihilation operators, respectively. As diagrams have complexity bounded by~$m_0$, we require that at intermediate steps we never reach more than $m_0$ background particles. We say that an abstract has degree $m$ if it creates in total $m$ background particles.
\smallskip\item We define \emph{histories} by complementing a given abstract with the set of annihilated or created indices. More precisely, at $i$-th collision, if $s_i=1$ (resp. $s_i=-1$), we let $b_i$ be the index of the created particle (resp. annihilated particle) and we let $a_i$ be the index of the particle with which it collides. As renormalization has removed all hats in the diagrams, we require that for $(s_i,s_{i+1})=(1,-1)$ we have $(a_i,b_i)\ne(a_{i+1},b_{i+1})$.
\end{enumerate}
There are various ways to truncate the Dyson series, but for convenience we shall proceed by restricting the summation to some special sets of abstracts, which we call \emph{admissible} sets.
More precisely, abstracts, histories, and admissible sets of abstracts are defined as follows.

\begin{defin}[Abstracts and histories]\label{def:hist}$ $
\begin{enumerate}[$\bullet$]
\item An \emph{abstract} is a finite sequence $(s_1,\ldots,s_n)\in\{\pm 1\}^n$ such that we have for some $0\le m\le m_0$,
\[m+\sum_{i=1}^n s_i =0,
\qquad\text{and}\qquad
0<m+\sum_{i=1}^j s_i\le m_0\quad\text{for all $1\le j< n$},\]
and $m$ is then called the \emph{degree} of the abstract.
\smallskip\item A set of abstracts $\Omega$ is \emph{admissible} if for any $(s_1,\ldots, s_n)\in \Omega$ we have $(s_2,\cdots s_n)\in\Omega$ and also $(-1,s_1,\cdots s_n)\in\Omega$ provided that its degree is still $\le m_0$. We naturally decompose $\Omega=\sqcup_{0\le m\le m_0}\Omega_m$, where $\Omega_m$ stands for the subset of abstracts with degree~$m$,
\[\Omega_m\,:=\,\Big\{(s_1,\ldots,s_n)\in\Omega\,:\,m+\sum_{i=1}^ns_i=0\Big\}.\]
In addition, we define the \emph{boundary} $\partial\Omega=\sqcup_{0\le m<m_0}\partial\Omega_m$ with
\[\partial\Omega_m\,:=\,\Big\{(1,s_1,\ldots,s_n)\notin\Omega_m\,:\,(s_1,\ldots,s_n)\in\Omega_{m+1}\Big\}.\]
\item Given an abstract $(s_1,\ldots,s_n)$ with degree $m$, we define associated \emph{histories} as sequences $(s_i,a_i,b_i)_{1\le i\le n}$ with $a_i\in\N$, $b_i\in\N\setminus\{0\}$, $a_i\ne b_i$, such that:
\begin{enumerate}[---]
\item if $(s_i,s_{i+1})=(1,-1)$, then $(a_i,b_i)\ne (a_{i+1},b_{i+1})$;
\item if $s_i=1$, then $a_i\in\omega_{i-1}$ and $b_i=1+\max\cup_{j:j<i}\omega_j=1+\max(m,a_1,b_1,\ldots,a_{i-1},b_{i-1})$;
\item if $s_i=-1$, then $a_i\in\omega_{i-1}$ and $b_i\in\omega_{i-1}\setminus\{a_i,0\}$;
\end{enumerate}
where the index sets $(\omega_i)_{0\le i\le n}$ are defined iteratively as follows,
\begin{equation}\label{eq:def-omj}
\omega_0:=\{0,\ldots,m\},\qquad \omega_{i}\,:=\,\left\{\begin{array}{ll}
\omega_{i-1}\cup\{1+\max\cup_{j:j<i}\omega_j\},&\text{if $s_i=1$},\\
\omega_{i-1}\setminus\{b_i\},&\text{if $s_i=-1$}.
\end{array}\right.
\end{equation}
We denote by $\mathfrak{H}(s_1,\ldots,s_n)$ the set of histories associated with the abstract $(s_1,\ldots,s_n)$.
\smallskip\item The contribution of an abstract $(s_1,\ldots,s_n)$ is defined as the following operator on $\Ld^2(\hat\Dd)$,
\begin{equation}\label{eq:contr-abstr}
\mathcal{I}_{(s_1,\ldots,s_n)} \,:=\, \sum_{(s_j,a_j,b_j)_{j}\in\mathfrak{H}(s_1,\ldots,s_n)} \,\hat{S}^{s_1}_{a_1,b_1}\,\WIGdot\,\ldots\,\WIGdot\,\hat{S}^{s_n}_{a_n,b_n},
\end{equation}
where we omit the superscript indicating the number of particles in $\hat S^{\pm}_{a_j,b_j}$ as there is no ambiguity. This defines an operator $\Ld^2(\hat\Dd)\to\Ld^2(\hat\Dd^{m+1})$ if the abstract has degree $m$.
\end{enumerate}
\end{defin}

In terms of this combinatorial structure, we are now in position to define the suitable ansatz for correlation functions: given an admissible set of abstracts $\Omega$, which will be chosen later on, we define for all $1\le m\le m_0$,
\begin{equation}\label{eq:def-tilgm}
\Lc\tilde g^{N, m_0}_m \,:=\, \sum_{n\ge1}\Big(\frac{i}{\sqrt{N}}\Big)^n\sum_{(s_1,\ldots,s_n)\in\Omega_m} 	\WIGdot\,\mathcal{I}_{(s_1,\ldots,s_n)}\,\Lc\tilde g^{N}_0,
\end{equation}
where we recall that $\tilde g^{N}_0$ is the ansatz~\eqref{eq:def-tilg0-re} for the tagged particle density. (We emphasize that for~$m=0$ this formula does not hold as it does not coincide with the actual choice $\tilde g_0^{N,m_0}=\tilde g_0^N$.)

\subsection{Approximate hierarchy}
In order to compare the above constructed ansatz with the true solution of the hierarchy~\eqref{eq:hierarchy}, we compare the equations they satisfy: the ansatz only satisfies the desired hierarchy up to some remainder terms. The notion of admissible abstracts in the definition of the ansatz is precisely tailored to yield a tractable formula for the remainder terms.

\begin{lem}\label{lem:approx-hier}
The ansatz~\eqref{eq:def-tilg0-re}--\eqref{eq:def-tilgm} satisfies for all $0\le m\le m_0$,
\begin{equation}\label{eq:rem}
\Big(\partial_\tau + it_N\hat L_m  +\tfrac{\kappa t_N}{N}\hat D_m\Big) \tilde{g}^{N,m_0}_m \,=\, \frac{it_N}{\sqrt{N}} \Big(\hat S^+_{m}\tilde{g}^{N,m_0}_{m+1}+\hat S^-_{m}\tilde{g}^{N,m_0}_{m-1}\Big)+R^{N,m_0}_m,
\end{equation}
with initial data $\tilde g_m^{N,m_0}|_{\tau=0}=\mathfrak g\mathds1_{m=0}$, where we have set for convenience $\tilde g_m^{N,m_0}=0$ for $m<0$ or $m>m_0$, and where the remainder terms are given as follows for all $1\le m\le m_0$,
\begin{eqnarray}
\Lc R_0^{N,m_0}&:=&-\frac{t_N}{N^2}\sum_{n\ge3}\Big(\frac{i}{\sqrt N}\Big)^{n-3}\!\!\sum_{(s_1,\ldots,s_n)\in\Omega_1}\hat S^+_{0}\,\WIGdot\,\Ic_{(s_1,\ldots,s_n)}\Lc\tilde g^{N}_{0}\nonumber\\
&&+\frac{t_N}{N^2}\mathds1_{m_0\ge2}\Big({\tiny\begin{tikzpicture}[baseline={([yshift=-1ex]current bounding box.center)},scale=0.8]
\begin{scope}[every node/.style={circle,draw,fill=white,inner sep=0pt,minimum size=3pt}]
\node (1) at (0,0) {};
\node (2) at (1.5,0) {};
\node (3) at (0.5,0.33) {};
\node (4) at (1,0.33) {};
\end{scope}
\draw (0.5,0) -- (1) -- (0,0.33) -- (3) -- (0.5,0.66);
\draw (1) -- (0,0.33);
\draw (2) -- (1.5,0.33);
\draw (4) -- (1,0.66);
\draw [decorate,decoration={snake,amplitude=1pt,segment length=3pt}] (3) -- (4) -- (1.5,0.33);
\draw [decorate,decoration={snake,amplitude=1pt,segment length=3pt}] (0.5,0) -- (2);
\draw [decorate,decoration={snake,amplitude=1pt,segment length=3pt}] (0.5,0.66) -- (1,0.66);
\end{tikzpicture}}
+
{\tiny\begin{tikzpicture}[baseline={([yshift=-1ex]current bounding box.center)},scale=0.8]
\begin{scope}[every node/.style={circle,draw,fill=white,inner sep=0pt,minimum size=3pt}]
\node (1) at (0,0) {};
\node (2) at (0.5,0) {};
\node (3) at (1,0) {};
\node (4) at (1.5,0) {};
\end{scope}
\draw (0.5,0.33) -- (2) -- (1) -- (0,0.66) -- (0.5,0.66);
\draw (3) -- (1,0.33);
\draw (1.5,0.66) -- (4);
\draw [decorate,decoration={snake,amplitude=1pt,segment length=3pt}] (2) -- (3) -- (4);
\draw [decorate,decoration={snake,amplitude=1pt,segment length=3pt}] (0.5,0.33) -- (1,0.33);
\draw [decorate,decoration={snake,amplitude=1pt,segment length=3pt}] (0.5,0.66) -- (1.5,0.66);
\end{tikzpicture}}
\Big)\Lc\tilde g_0^N,\label{eq:RN0}\\[1mm]
\Lc R_m^{N,m_0}&:=&\frac{it_N}{N^{3/2}}\mathds1_{m<m_0}
\sum_{n\ge3}\Big(\frac{i}{\sqrt N}\Big)^{n-3}\sum_{(s_1,\ldots,s_n)\in\partial\Omega_{m}}\Ic_{(s_1,\ldots,s_n)}\Lc\tilde g_0^N.\label{eq:RNm}
\end{eqnarray}
\end{lem}

\begin{proof}
For $m=0$, inserting the definition of $\tilde g_0^N$ and $\tilde g_1^{N,m_0}$, cf.~\eqref{eq:def-tilg0-re} and~\eqref{eq:def-tilgm}, we find
\begin{eqnarray*}
\lefteqn{\Big(1+i\alpha +\tfrac{\kappa t_N}{N}\hat D_0\Big)\Lc\tilde g^{N}_0 - \frac{it_N}{\sqrt{N}} \hat S^+_{0}\Lc\tilde g^{N,m_0}_{1}}\\
&=&\mathfrak g+t_N\Big(\frac{i}{\sqrt N}\Big)^2\BOX\,\Lc\tilde g_0^N- t_N\sum_{n\ge1}\Big(\frac{i}{\sqrt N}\Big)^{n+1}\!\!\sum_{(s_1,\ldots,s_n)\in\Omega_1}\hat S^+_{0}\,\WIGdot\,\Ic_{(s_1,\ldots,s_n)}\Lc\tilde g^{N}_{0}\\
&=&\mathfrak g+t_N\Big(\frac{i}{\sqrt N}\Big)^2\Big(\BOX-\WIGbox\Big)\,\Lc\tilde g_0^N- t_N\sum_{n\ge3}\Big(\frac{i}{\sqrt N}\Big)^{n+1}\!\!\sum_{(s_1,\ldots,s_n)\in\Omega_1}\hat S^+_{0}\,\WIGdot\,\Ic_{(s_1,\ldots,s_n)}\Lc\tilde g^{N}_{0},
\end{eqnarray*}
where we have noted that for $n=1$ the only length-$1$ abstract $(s_1)\in\Omega_1$ is $(s_1)=(-1)$, and that the length of abstracts in $\Omega_1$ is always odd. By definition of the renormalized propagator, cf.~Definition~\ref{def:renorm}, the resolvent identity yields
\[\BOX-\WIGbox\,=\,\frac1N\mathds1_{m_0\ge2}\Big({\tiny\begin{tikzpicture}[baseline={([yshift=-1ex]current bounding box.center)},scale=0.8]
\begin{scope}[every node/.style={circle,draw,fill=white,inner sep=0pt,minimum size=3pt}]
\node (1) at (0,0) {};
\node (2) at (1.5,0) {};
\node (3) at (0.5,0.33) {};
\node (4) at (1,0.33) {};
\end{scope}
\draw (0.5,0) -- (1) -- (0,0.33) -- (3) -- (0.5,0.66);
\draw (1) -- (0,0.33);
\draw (2) -- (1.5,0.33);
\draw (4) -- (1,0.66);
\draw [decorate,decoration={snake,amplitude=1pt,segment length=3pt}] (3) -- (4) -- (1.5,0.33);
\draw [decorate,decoration={snake,amplitude=1pt,segment length=3pt}] (0.5,0) -- (2);
\draw [decorate,decoration={snake,amplitude=1pt,segment length=3pt}] (0.5,0.66) -- (1,0.66);
\end{tikzpicture}}
+
{\tiny\begin{tikzpicture}[baseline={([yshift=-1ex]current bounding box.center)},scale=0.8]
\begin{scope}[every node/.style={circle,draw,fill=white,inner sep=0pt,minimum size=3pt}]
\node (1) at (0,0) {};
\node (2) at (0.5,0) {};
\node (3) at (1,0) {};
\node (4) at (1.5,0) {};
\end{scope}
\draw (0.5,0.33) -- (2) -- (1) -- (0,0.66) -- (0.5,0.66);
\draw (3) -- (1,0.33);
\draw (1.5,0.66) -- (4);
\draw [decorate,decoration={snake,amplitude=1pt,segment length=3pt}] (2) -- (3) -- (4);
\draw [decorate,decoration={snake,amplitude=1pt,segment length=3pt}] (0.5,0.33) -- (1,0.33);
\draw [decorate,decoration={snake,amplitude=1pt,segment length=3pt}] (0.5,0.66) -- (1.5,0.66);
\end{tikzpicture}}
\Big),\]
and we conclude that equation~\eqref{eq:rem} indeed holds for $m=0$ with remainder $R_0^{N,m_0}$ given by the claimed formula~\eqref{eq:RN0}.

We turn to the equation for higher correlations, $m\ge1$. Given an abstract $(s_1,\ldots,s_n)$ with degree~$m$, by definition of the renormalized propagator, the resolvent identity yields
\begin{equation*}
\Big(1+i\alpha+it_N\hat L_m+\tfrac{\kappa t_N}N\hat D_m\Big)\WIGdot\,\Ic_{(s_1,\ldots,s_n)}
\,=\,t_N\Ic_{(s_1,\ldots,s_n)}+t_N\Big(\frac{i}{\sqrt N}\Big)^2\Big(\sum\WIGBOX\Big)\WIGdot\,\Ic_{(s_1,\ldots,s_n)}.
\end{equation*}
By definition~\eqref{eq:contr-abstr} of $\Ic_{(s_1,\ldots,s_n)}$ and by definition of histories, cf.~Definition~\ref{def:hist}, we can decompose
\begin{equation*}
\Ic_{(s_1,\ldots,s_n)}
\,=\,\Big(\mathds1_{s_1=1}\hat S^{+}_m
+\mathds1_{s_1=-1}\hat S^{-}_m\Big)\WIGdot\,\Ic_{(s_2,\ldots,s_n)}\\
-\mathds1_{(s_1,s_2)=(1,-1)}\Big(\sum\WIGBOX\Big)\WIGdot\,\Ic_{(s_3,\ldots,s_n)},
\end{equation*}
so the above becomes
\begin{multline*}
\Big(1+i\alpha+it_N\hat L_m+\tfrac{\kappa t_N}N\hat D_m\Big)\WIGdot\,\Ic_{(s_1,\ldots,s_n)}
-t_N\Big(\mathds1_{s_1=1}\hat S^{+}_m
+\mathds1_{s_1=-1}\hat S^{-}_m\Big)\WIGdot\,\Ic_{(s_2,\ldots,s_n)}\\
\,=\,t_N\Big(\sum\WIGBOX\Big)\bigg(\Big(\frac{i}{\sqrt N}\Big)^2\WIGdot\,\Ic_{(s_1,\ldots,s_n)}-\mathds1_{(s_1,s_2)=(1,-1)}\WIGdot\,\Ic_{(s_3,\ldots,s_n)}\bigg).
\end{multline*}
Summing over admissible abstracts, recalling the definition of the ansatz, cf.~\eqref{eq:def-tilgm}, and simplifying the telescoping sum, we find
\begin{multline}\label{eq:decomp-eqn-gm}
\Big(1+i\alpha+it_N\hat L_m+\tfrac{\kappa t_N}N\hat D_m\Big)\Lc\tilde g_m^{N,m_0}
-\frac{it_N}{\sqrt N}\hat S^{+}_m\sum_{n\ge1}\Big(\frac{i}{\sqrt N}\Big)^{n}\sum_{(s_1,\ldots,s_n)\in\Omega_{m+1}\atop(1,s_1,\ldots,s_n)\in\Omega_m}\WIGdot\,\Ic_{(s_1,\ldots,s_n)}\Lc\tilde g_0^N\\
-\mathds1_{m=1}\frac{it_N}{\sqrt N}\hat S^{-}_m\Lc\tilde g_0^N
-\mathds1_{m\ge2}\frac{it_N}{\sqrt N} \hat S^{-}_m \sum_{n\ge1}\Big(\frac{i}{\sqrt N}\Big)^{n}\sum_{(s_1,\ldots,s_{n})\in\Omega_{m-1}\atop(-1,s_1,\ldots,s_n)\in\Omega_m}\WIGdot\,\Ic_{(s_1,\ldots,s_n)}\Lc\tilde g_0^N\\
\,=\,-\frac{t_N}{N}\Big(\sum\WIGBOX\Big)\sum_{n\ge1}\Big(\frac{i}{\sqrt N}\Big)^{n}\sum_{(s_1,\ldots,s_n)\in\Omega_m\atop(1,-1,s_1,\ldots,s_n)\notin\Omega_m}\WIGdot\,\Ic_{(s_1,\ldots,s_n)}\Lc\tilde g_0^N.
\end{multline}
Now using that $\Omega$ is an admissible set of abstracts, and recognizing the definition~\eqref{eq:def-tilgm} of the ansatz, we obtain
\begin{multline*}
\Big(1+i\alpha+it_N\hat L_m+\tfrac{\kappa t_N}N\hat D_m\Big)\Lc\tilde g_m^{N,m_0}
-\frac{it_N}{\sqrt N}\Big(\hat S^{+}_m\Lc\tilde g_{m+1}^{N,m_0}+\hat S^{-}_m\Lc\tilde g_{m-1}^{N,m_0}\Big)\\
\,=\,-\frac{it_N}{\sqrt N}\hat S^{+}_m\sum_{n\ge1}\Big(\frac{i}{\sqrt N}\Big)^{n}\sum_{(s_1,\ldots,s_n)\in\Omega_{m+1}\atop(1,s_1,\ldots,s_n)\notin\Omega_m}\WIGdot\,\Ic_{(s_1,\ldots,s_n)}\Lc\tilde g_0^N\\
-\frac{t_N}{N}\Big(\sum\WIGBOX\Big)\sum_{n\ge1}\Big(\frac{i}{\sqrt N}\Big)^{n}\sum_{(s_1,\ldots,s_n)\in\Omega_m\atop(1,-1,s_1,\ldots,s_n)\notin\Omega_m}\WIGdot\,\Ic_{(s_1,\ldots,s_n)}\Lc\tilde g_0^N.
\end{multline*}
Further reorganizing the right-hand side and recalling the definition of the boundary $\partial\Omega_m$, we conclude that equation~\eqref{eq:rem} indeed holds for $1\le m\le m_0$ with remainder $R_m^{N,m_0}$ given by the claimed formula~\eqref{eq:RNm}.
\end{proof}

\subsection{Error estimates}
As the ansatz is an approximate solution of the hierarchy~\eqref{eq:hierarchy}, cf.~Lemma~\ref{lem:approx-hier}, we can now appeal to the unitary structure of the latter to deduce error estimates.

\begin{lem}\label{lem:err-est0}
The ansatz~\eqref{eq:def-tilg0-re}--\eqref{eq:def-tilgm} satisfies for all $0\le m\le m_0$,
\begin{equation*}
\sup_{\tau\ge0}\Big(e^{-\tau}\|(g_m^{N,m_0}-\tilde g_m^{N,m_0})(\tau)\|\Big)+\3\Lc g_m^{N,m_0}-\Lc\tilde g_m^{N,m_0}\3\,\lesssim\,\Big(\sum_{m=0}^{m_0}\3\Lc R_m^{N,m_0}\3^2\Big)^\frac12,
\end{equation*}
where $(R_m^{N,m_0})_{0\le m\le m_0}$ is defined in~\eqref{eq:RN0}--\eqref{eq:RNm}.
\end{lem}

\begin{proof}
As the ansatz satisfies the approximate hierarchy~\eqref{eq:rem}, we have for all $0\le m\le m_0$,
\begin{multline*}
\Big(t_N^{-1}\partial_\tau+i\hat L_m+\tfrac\kappa N\hat D_m\Big)(g_m^{N}-\tilde g_m^{N,m_0})\\[-2mm]
=\frac1{\sqrt N}\Big(i\hat S^+_m(g_{m+1}^N-\tilde g_{m+1}^{N,m_0})+i\hat S_m^-(g_{m-1}^{N}-\tilde g_{m-1}^{N,m_0})\Big)-t_N^{-1}R_m^{N,m_0},
\end{multline*}
with vanishing initial data $(g_m^{N,m_0}-\tilde g_m^{N,m_0})|_{\tau=0}=0$.
By the unitarity structure of the hierarchy, cf.~\eqref{eq:adjS-re}, we obtain the energy identity
\begin{equation*}
\partial_\tau\sum_{m=0}^{m_0}\|g_m^{N,m_0}-\tilde g_m^{N,m_0}\|^2
+2\frac{\kappa t_N}N\sum_{m=0}^{m_0}\|\nabla_{v_{[m]}}(g_m^{N,m_0}-\tilde g_m^{N,m_0})\|^2
=-2\sum_{m=0}^{m_0}\<g_m^{N,m_0}-\tilde g_m^{N,m_0},R_m^{N,m_0}\>.
\end{equation*}
Integrating in time and using nonnegativity of the dissipation term, we find
\begin{equation*}
\Big(\sum_{m=0}^{m_0}\|(g_m^{N,m_0}-\tilde g_m^{N,m_0})(\tau)\|^2\Big)^\frac12
\le\int_0^\tau\Big(\sum_{m=0}^{m_0}\|R_m^{N,m_0}\|^2\Big)^\frac12.
\end{equation*}
To estimate the right-hand side, we use the bound $\int_0^\tau \|\varphi\|\lesssim(\int_0^\tau e^{2\sigma}d\sigma)^{1/2}\3\Lc\varphi\3\le e^\tau\3\Lc\varphi\3$ in terms of the Laplace transform~\eqref{eq:Lap}. This yields the first part of the claim. For the second part, we apply Laplace transform to the evolution equation satisfied by $g_m^N-\tilde g_m^{N,m_0}$ and repeat the above estimate at the level of Laplace transforms.
\end{proof}

\section{Renormalized propagators}\label{sec:ren}
This section is devoted to the study of the renormalized propagators $\WIGdot$ introduced in Definition~\ref{def:renorm}, which amounts to formally resumming the contributions of hat (and nested hat) operators.

\subsection{Properties of the hat operator}
Before analyzing the renormalization procedure, we start by studying the hat operator $\BOXdot$ and we show that it is nonnegative and controlled by diffusion in velocity.
	
\begin{lem}\label{lem:box}
For all $0\le m<m_0$, $g_m,h_m\in C^\infty_c(\hat\Dd^{m+1})$, and $\ell\ge0$, we have
\begin{eqnarray*}
\Re\Big\<g_m\,,\,\Big(\sum\BOXdot\Big)\,g_m\Big\>_k&\ge&0,\\
\Big\|\nabla_{v_{[m]}}^\ell\Big(\sum\BOXdot\Big)g_m\Big\|_k&\lesssim_\ell&\sum_{s=0}^{\ell+1}C_s^\Vc\<k_{[m]}\>^s\|\nabla_{v_{[m]}}^{\ell+2-s}g_m\|_k,\\
\Big|\Big\<h_m\,,\,\nabla_{v_{[m]}}^\ell\Big(\sum\BOXdot\Big)g_m\Big\>_k\Big|
&\lesssim_\ell&\|\nabla_{v_{[m]}}h_m\|_k\sum_{s=0}^\ell C_s^\Vc\<k_{[m]}\>^s\|\nabla_{v_{[m]}}^{\ell+1-s}g_m\|_k,
\end{eqnarray*}
where we have defined
\begin{equation*}
C_s^\Vc\,:=\,\int_{\R^d}\<k\>^s|k|^{1-s}\V(k)^2\,dk.
\end{equation*}
In particular, we can deduce
\begin{equation}\label{eq:cor-box}
\Re\Big\<\nabla_{v_{[m]}}^\ell g_m\,,\, \nabla_{v_{[m]}}^\ell\Big(\sum\BOXdot\Big)g_m\Big\>_k
\,\gtrsim_\ell\,-\|\nabla_{v_{[m]}}^{\ell+1}g_m\|_k\sum_{s=1}^\ell C_s^\Vc\<k_{[m]}\>^s\|\nabla_{v_{[m]}}^{\ell+1-s}g_m\|_k.
\end{equation}
\end{lem}

\begin{proof}
We split the proof into three steps.

\medskip\noindent
{\bf Step~1:} Positivity.\\
By definition, for $0\le m< m_0$, the hat operator reads as follows,
\begin{multline}\label{eq:def-box-re}
\sum\BOXdot
\,=\,-\sum_{j=0}^{m}\Div_{v_j}\bigg[\int_{\Dd} (k_{m+1}\otimes k_{m+1})\V(k_{m+1})^2\sqrt M(v_{m+1})\\[-1mm]
\times\bigg(\tfrac{1+i\alpha}{t_N}+\sum_{l=0}^{m}ik_l\cdot v_l+ik_{m+1}\cdot(v_{m+1}-v_j)-\tfrac\kappa N\triangle_{v_{[m+1]}}\bigg)^{-1}\sqrt M(v_{m+1})\,d^*k_{m+1}\,dv_{m+1}\bigg]\nabla_{v_j}.
\end{multline}
As the resolvent identity yields
\begin{multline*}
\Re\bigg(\tfrac{1+i\alpha}{t_N}+\sum_{l=0}^{m}ik_l\cdot v_l+ik_{m+1}\cdot(v_{m+1}-v_j)-\tfrac\kappa N\triangle_{v_{[m+1]}}\bigg)^{-1}\\[-2mm]
=\bigg(\tfrac{1-i\alpha}{t_N}-\sum_{l=0}^{m}ik_l\cdot v_l-ik_{m+1}\cdot(v_{m+1}-v_j)-\tfrac\kappa N\triangle_{v_{[m+1]}}\bigg)^{-1}\\[-2mm]
\times\Big(\tfrac{1}{t_N}-\tfrac\kappa N\triangle_{v_{[m+1]}}\Big)
\bigg(\tfrac{1+i\alpha}{t_N}+\sum_{l=0}^{m}ik_l\cdot v_l+ik_{m+1}\cdot(v_{m+1}-v_j)-\tfrac\kappa N\triangle_{v_{[m+1]}}\bigg)^{-1},
\end{multline*}
we obtain for $g_m\in C^\infty_c(\hat\Dd^{m+1})$,
\begin{equation*}
\Re\Big\<g_m,\Big(\sum\BOXdot\Big)\,g_m\Big\>_k
\,=\,\sum_{j=0}^m\int_{(\R^d)^{m+3}}\overline{h_{m,j}(\hat z_{[m+1]})}\Big(\tfrac{1}{t_N}-\tfrac\kappa N\triangle_{v_{[m+1]}}\Big)h_{m,j}(\hat z_{[m+1]})\,dv_{[m]}d^*\hat z_{m+1},
\end{equation*}
where we have set for abbreviation
\begin{multline*}
h_{m,j}(\hat z_{[m+1]})\,:=\,
\bigg(\tfrac{1+i\alpha}{t_N}+\sum_{l=0}^{m}ik_l\cdot v_l+ik_{m+1}\cdot(v_{m+1}-v_j)-\tfrac\kappa N\triangle_{v_{[m+1]}}\bigg)^{-1}\\[-2mm]
\times\Big(k_{m+1}\V(k_{m+1})\cdot\nabla_{v_j}g(v_j)\sqrt M(v_{m+1})\Big).
\end{multline*}
The claimed positivity follows.

\medskip\noindent
{\bf Step~2:} Control by velocity diffusion.\\
Starting point is the following commutation relation, for $0\le j\le m$,
\[\bigg[\nabla_{v_{[m]}}~,~\tfrac{1+i\alpha}{t_N}+\sum_{l=0}^{m}ik_l\cdot v_l+ik_{m+1}\cdot(v_{m+1}-v_j)-\tfrac\kappa N\triangle_{v_{[m+1]}}\bigg]\,=\,i(k_{[m]\setminus\{j\}},k_j-k_{m+1}),\]
from which we can deduce, for $n\ge0$,
\begin{multline*}
\nabla_{v_{[m]}}^n\bigg(\tfrac{1+i\alpha}{t_N}+\sum_{l=0}^{m}ik_l\cdot v_l+ik_{m+1}\cdot(v_{m+1}-v_j)-\tfrac\kappa N\triangle_{v_{[m+1]}}\bigg)^{-1}\\
\,=\,\sum_{s=0}^n(-i)^s\binom{n}{s}(k_{[m]\setminus\{j\}},k_j-k_{m+1})^{\otimes s}\bigg(\tfrac{1+i\alpha}{t_N}+\sum_{l=0}^{m}ik_l\cdot v_l+ik_{m+1}\cdot(v_{m+1}-v_j)-\tfrac\kappa N\triangle_{v_{[m+1]}}\bigg)^{-s-1}\nabla_{v_{[m]}}^{n-s}.
\end{multline*}
By definition of the hat operator, cf.~\eqref{eq:def-box-re}, we then get
\begin{equation}\label{eq:decomp-nab-box}
\nabla_{v_{[m]}}^n\Big(\sum\BOXdot\Big)g_m
\,=\,-\sum_{j=0}^{m}\sum_{s=0}^n\frac{n!}{(n-s)!}\Div_{v_j}\Big(\Gc_s\nabla_{v_j}\nabla_{v_{[m]}}^{n-s}g_m\Big),
\end{equation}
or alternatively, further expanding the remaining divergence,
\begin{equation}\label{eq:decomp-nab-box0}
\nabla_{v_{[m]}}^n\Big(\sum\BOXdot\Big)g_m
\,=\,-\sum_{j=0}^{m}\sum_{s=0}^n\frac{n!}{(n-s)!}\Gc_s:\nabla_{v_j}^2\nabla_{v_{[m]}}^{n-s}g_m
-\sum_{j=0}^{m}\sum_{s=0}^n\frac{n!(s+1)}{(n-s)!}\Gc'_{s,j}\cdot\nabla_{v_j}\nabla_{v_{[m]}}^{n-s}g_m,
\end{equation}
in terms of the operators
\begin{eqnarray*}
\Gc_s&:=&(-i)^s\int_{\Dd} (k_{m+1}\otimes k_{m+1})\V(k_{m+1})^2\sqrt M(v_{m+1})(k_{[m]\setminus\{j\}},k_j-k_{m+1})^{\otimes s}\\
&&\times \bigg(\tfrac{1+i\alpha}{t_N}+\sum_{l=0}^{m}ik_l\cdot v_l+ik_{m+1}\cdot(v_{m+1}-v_j)-\tfrac\kappa N\triangle_{v_{[m+1]}}\bigg)^{-s-1}\sqrt M(v_{m+1})\,d^*\hat z_{m+1},\\
\Gc'_{s,j}&:=&(-i)^{s+1}\int_{\Dd} (k_{m+1}\otimes k_{m+1})\V(k_{m+1})^2\sqrt M(v_{m+1})(k_{[m]\setminus\{j\}},k_j-k_{m+1})^{\otimes s+1}\\
&&\times \bigg(\tfrac{1+i\alpha}{t_N}+\sum_{l=0}^{m}ik_l\cdot v_l+ik_{m+1}\cdot(v_{m+1}-v_j)-\tfrac\kappa N\triangle_{v_{[m+1]}}\bigg)^{-s-2}\sqrt M(v_{m+1})\,d^*\hat z_{m+1}.
\end{eqnarray*}
In order to perform complex contour deformation in $v_{m+1}$, we rely on the following Green's representation formula for transport-diffusion resolvent operators, which is easily checked by a Fourier calculation: for $\sigma>0$, $\Re \omega>0$, and $h\in C^\infty_c(\R^d)$,
\[(\omega+ik\cdot v-\sigma\triangle_v)^{-1}h(v)=\frac1{(4\pi\sigma t)^{d/2}}\int_0^\infty\!\!\!\int_{\R^d} \exp\bigg(-t\omega-it\frac{k\cdot(u+v)}{2}-\frac{\sigma|k|^2t^3}{12}-\frac{|u-v|^2}{4\sigma t}\bigg)\,h(u)\,dudt.\]
Returning to the above operator $\Gc_s$, recalling that $M$ is the Maxwellian distribution, and using the holomorphy of the Green's function, we can appeal to the holomorphy of the integrands in $v_{m+1}$ and perform complex contour deformation,
\[v_{m+1}\mapsto v_{m+1}-i\hat k_{m+1},\]
where we use the short-hand notation $\hat k:=\tfrac{k}{|k|}$.
This allows us to rewrite $\Gc_s$ as
\begin{multline*}
\Gc_s\,=\,(-i)^s\int_{\Dd} (k_{m+1}\otimes k_{m+1})\V(k_{m+1})^2\sqrt M(v_{m+1}-i\hat k_{m+1})(k_{[m]\setminus\{j\}},k_j-k_{m+1})^{\otimes s}\\
\times \bigg(\tfrac{1+i\alpha}{t_N}+|k_{m+1}|+\sum_{l=0}^{m}ik_l\cdot v_l+ik_{m+1}\cdot(v_{m+1}-v_j)-\tfrac\kappa N\triangle_{v_{[m+1]}}\bigg)^{-s-1}\!\sqrt M(v_{m+1}-i\hat k_{m+1})\,d^*\hat z_{m+1},
\end{multline*}
from which we can directly estimate
\begin{equation}\label{eq:estim-Gsgm}
\|\Gc_sg_m\|_k
\,\lesssim\,\|g_m\|_k\int_{\R^d} \<k_{[m+1]}\>^s|k_{m+1}|^{1-s}\V(k_{m+1})^2\,dk_{m+1}
\,\lesssim\,C_s^\Vc\<k_{[m]}\>^s\|g_m\|_k,
\end{equation}
and similarly
\[\|\Gc_{s,j}'g_m\|_k\,\lesssim\,C_{s+1}^\Vc\<k_{[m]}\>^{s+1}\|g_m\|_k.\]
Combining this with~\eqref{eq:decomp-nab-box} and~\eqref{eq:decomp-nab-box0}, we obtain the desired estimates.

\medskip\noindent
{\bf Step~3:} proof of~\eqref{eq:cor-box}.\\
Singling out the term with $s=0$ in~\eqref{eq:decomp-nab-box}, we get
\begin{equation*}
\nabla_{v_{[m]}}^n\Big(\sum\BOXdot\Big)g_m
\,=\,-\sum_{j=0}^{m}\Div_{v_j}\Big(\Gc_0\nabla_{v_j}\nabla_{v_{[m]}}^{n}g_m\Big)
-\sum_{j=0}^{m}\sum_{s=1}^n\frac{n!}{(n-s)!}\Div_{v_j}\Big(\Gc_s\nabla_{v_j}\nabla_{v_{[m]}}^{n-s}g_m\Big),
\end{equation*}
and thus,
\begin{multline*}
\Big\<\nabla_{v_{[m]}}^ng_m\,,\, \nabla_{v_{[m]}}^n\Big(\sum\BOXdot\Big)g_m\Big\>_k
\,=\,\sum_{j=0}^{m}
\Big\<(\nabla_{v_j}\nabla_{v_{[m]}}^ng_m)\,,\,\Gc_0\,(\nabla_{v_j}\nabla_{v_{[m]}}^{n}g_m)\Big\>_k\\[-2mm]
+\sum_{j=0}^{m}\sum_{s=1}^n\frac{n!}{(n-s)!}\Big\<(\nabla_{v_j}\nabla_{v_{[m]}}^ng_m)\,,\,\Gc_s\,(\nabla_{v_j}\nabla_{v_{[m]}}^{n-s}g_m)\Big\>_k.
\end{multline*}
As the result of Step~1 yields $\Re\Gc_0\ge0$, we note that the first right-hand side term has nonnegative real part. The claim~\eqref{eq:cor-box} then follows from~\eqref{eq:estim-Gsgm}.
\end{proof}

\subsection{Renormalized propagator}
We start by showing that the iterative definition of the renormalized propagator in Definition~\ref{def:renorm} makes sense. This easily follows by showing at the same time that the renormalized hat operators are nonnegative, thus extending the corresponding positivity statement for the hat operator in Lemma~\ref{lem:box}.

\begin{lem}\label{lem:box-ren}
Definition~\ref{def:renorm} for the renormalized propagator makes sense and yields a bounded operator~$\WIGdot$ on $L^2(\hat\Dd^{m+1})$ for all $1\le m\le m_0$. In addition, the associated renormalized hat operator is nonnegative: for all $0\le m<m_0$ and $g_m\in C^\infty_c(\hat\Dd^{m+1})$,
\[\Re\Big\<g_m\,,\,\Big(\sum\WIGBOX\Big)g_m\Big\>_k\,\ge\,0.\]
\end{lem}

\begin{proof}
As the definition of the renormalized propagator is iterative, we argue by induction. Recalling that for $m=m_0$ we have $\WIGdot={\tiny\begin{tikzpicture}[baseline={([yshift=-.8ex]current bounding box.center)},scale=0.8]
\draw (0,0) -- (0.8,0);
\draw (0,0.45) -- (0.8,0.45);
\draw[dotted] (0.4,0.1) -- (0.4,0.4);
\end{tikzpicture}}$ on $L^2(\hat\Dd^{m_0+1})$, cf.~Definition~\ref{def:renorm}, we can use Lemma~\ref{lem:box} to initiate the induction. We then split the proof into two steps.

\medskip\noindent
{\bf Step~1.} Proof that, for all $0\le m<m_0$, if $\WIGdot$ is well-defined on $L^2(\hat\Dd^{m+2})$ and if the renormalized hat operator $\WIGBOX$ on $L^2(\hat\Dd^{m+2})$ is nonnegative in the sense of
\[\Re\Big\<g_{m+1}\,,\,\Big(\sum\WIGBOX\Big)g_{m+1}\Big\>_k\,\ge\,0,\]
then the operator $\WIGBOX$ on $L^2(\hat\Dd^{m+1})$ is also nonnegative in the sense of
\[\Re\Big\<g_m\,,\,\Big(\sum\WIGBOX\Big)g_m\Big\>_k\,\ge\,0.\]
By definition of collision operators in the diagrammatic notation, as $\WIGdot$ is assumed well-defined and bounded on $L^2(\hat\Dd^{m+2})$, we can write
\begin{eqnarray*}
\lefteqn{\Big\<g_m\,,\,\Big(\sum\WIGBOX\Big)g_m\Big\>_k
\,=\,\sum_{j=0}^m\Big\<g_m\,,\,\hat S_{j,m+1}^{m,+}\WIGdot \hat S^{m+1,-}_{j,m+1} g_m\Big\>_k}\\
&=&\sum_{j=0}^m\int_{\R^d}\int_{(\R^d)^{m+2}} \overline{h_{m+1}^j}\\
&&\times\bigg(\tfrac{i\alpha+1}{t_N}+\sum_{l=0}^mik_l\cdot v_l+ik_{m+1}\cdot (v_{m+1}-v_j)+\tfrac\kappa N\hat D_{m+1}+\tfrac1N\sum\WIGBOX\bigg)^{-1}h_{m+1}^j\,dv_{[m+1]}d^*k_{m+1},
\end{eqnarray*}
in terms of
\[h_{m+1}^j\,:=\,\sqrt M(v_{m+1})k_{m+1}\V(k_{m+1})\cdot\nabla_{v_j}g_m(\hat z_{[m]}).\]
Taking the real part and using the resolvent identity as in Step~1 of the proof of Lemma~\ref{lem:box}, we obtain
\begin{equation*}
\Re\Big\<g_m\,,\,\Big(\sum\WIGBOX\Big)g_m\Big\>_k
\,=\,\sum_{j=0}^m\Re\int_{\R^d}\int_{(\R^d)^{m+2}} \overline{H_{m+1}^j}
\bigg(\tfrac{1}{t_N}+\tfrac\kappa N\hat D_{m+1}+\tfrac1N\sum\WIGBOX\bigg)H_{m+1}^j\,dv_{[m+1]}d^*k_{m+1},
\end{equation*}
in terms of
\[H_{m+1}^j\,:=\,\bigg(\tfrac{i\alpha+1}{t_N}+\sum_{l=0}^mik_l\cdot v_l+ik_{m+1}\cdot (v_{m+1}-v_j)+\tfrac\kappa N\hat D_{m+1}+\tfrac1N\sum\WIGBOX\bigg)^{-1}h_{m+1}^j.\]
From this identity, as the operator $\sum\WIGBOX$ is assumed to be nonnegative on $L^2(\hat\Dd^{m+2})$, we conclude that the same must hold for the corresponding operator on $L^2(\hat\Dd^{m+1})$.

\medskip\noindent
{\bf Step~2.} Proof that, for all $1\le m\le m_0$, if $\WIGdot$ is well-defined and bounded on $L^2(\hat\Dd^{m+2})$ and if~$\WIGBOX$ is nonnegative on $L^2(\hat\Dd^{m+1})$ in the sense of Step~1, then the definition of $\WIGdot$ on $L^2(\hat\Dd^{m+1})$ in Definition~\ref{def:renorm} also makes sense and we have
\[\|\WIGdot g_m\|_k\,\le\,t_N\|g_m\|_k.\]
Given $g_m\in C^\infty_c(\hat\Dd^{m+1})$, consider the equation
\[\bigg(\tfrac{i\alpha+1}{t_N}+\sum_{l=0}^mik_l\cdot v_l+\tfrac\kappa N\hat D_{m}+\tfrac1N\sum\WIGBOX\bigg)h_m\,=\,g_m.\]
Testing with $h_m$, taking the real part, and using the assumed nonnegativity of $\WIGBOX$, we deduce
\[\tfrac1{t_N}\|h_m\|_k^2+\tfrac\kappa N\|\nabla_{v_{[m]}}h_m\|^2_k\,\le\,\|g_m\|_k\|h_m\|_k,\]
and thus
\[\|h_m\|_k\,\le\,t_N\|g_m\|_k.\]
This proves that $h_m=\WIGdot g_m$ is well-defined and satisfies the desired estimate.
\end{proof}

\subsection{Complex deformations of renormalized propagator}\label{sec:renorm et complex}
Similarly as the transport-diffusion resolvent $(\frac{i\alpha+1}{t_N}+ik\cdot v-\frac\kappa N\triangle_v)^{-1}$ can be deformed in the complex plane when applied to holomorphic functions in velocity, cf.~\eqref{eq:deform-0}, we show that similar deformations can be performed on renormalized propagators. This is useful as deformations can be chosen to obtain uniform-in-$N$ resolvent estimates.
Complex deformations are performed at the level of velocity translations and we start with some notation: for $\xi_{[m]}\in(\R^d)^{m+1}$, we denote by~$\tau_{\xi_{[m]}}$ the translation operator on~$\Ld^2(\hat\Dd^{m+1})$,
\begin{equation}\label{eq:not-tau}
(\tau_{\xi_{[m]}}g_m)(\hat z_{[m]}) =g_m(k_{[m]},v_{[m]}+\xi_{[m]}),
\end{equation}
and we consider the conjugation operator $T_{\xi_{[m]}}$, for any operator $X_m$ on $\Ld^2(\hat\Dd^{m+1})$,
\begin{equation}\label{eq:not-Tconj}
T_{\xi_{[m]}}[X_m]\,:=\,\tau_{\xi_{[m]}}\,X_m\,\tau_{-\xi_{[m]}}.
\end{equation}
On top of the deformability of $\WIGdot$, we also show at the same time that the renormalized hat operator~$\WIGBOX$ satisfies similar bounds as the hat operator in Lemma~\ref{lem:box}. This is split into two different statements, but both will be proven at once.

\begin{prop}\label{prop:box-bound}
Assume that $\kappa\gg C_0^\Vc:=\int_{\R^d}|k|\V(k)^2\,dk$ is sufficiently large. For all $0\le m< m_0$, $g_m,h_m\in C^\infty_c(\hat\Dd^{m+1})$, and $\ell\ge0$, the renormalized hat operator satisfies
\begin{eqnarray*}
\Big|\Big\<h_m\,,\,\Big(\sum\WIGBOX\Big)\, g_m\Big\>_k\Big|
&\lesssim&C_0^\Vc\|\nabla_{v_{[m]}} h_m\|_k\|\nabla_{v_{[m]}}g_m\|_k,\\
\Big\|\nabla_{v_{[m]}}^\ell \Big(\sum\WIGBOX\Big)\, g_m\Big\|_k
&\lesssim_\ell&\sum_{s=0}^{\ell+1}\<k_{[m]}\>^{s}\|\nabla_{v_{[m]}}^{\ell+2-s}g_m\|_k,\\
\Big|\Big\<h_m\,,\,\nabla_{v_{[m]}}^\ell\Big(\sum\WIGBOX\Big)\, g_m\Big\>_k\Big|
&\lesssim_\ell&\|\nabla_{v_{[m]}}h_m\|_k\sum_{s=0}^{\ell}\<k_{[m]}\>^{s}\|\nabla_{v_{[m]}}^{\ell+1-s}g_m\|_k.
\end{eqnarray*}
In addition, generalizing the positivity statement of Lemma~\ref{lem:box-ren}, we have for all $\ell\ge0$,
\begin{equation}\label{eq:posit-box-ren-re}
\Re\Big\<\nabla_{v_{[m]}}^\ell g_m\,,\,\nabla_{v_{[m]}}^\ell\Big(\sum\WIGBOX\Big)g_m\Big\>_k\,\gtrsim_\ell\,-\|\nabla_{v_{[m]}}^{\ell+1}g_m\|_k\sum_{s=1}^\ell\langle k_{[m]}\rangle^{s}\|\nabla_{v_{[m]}}^{\ell+1-s}g_m\|_k.
\end{equation}
\end{prop}
\begin{prop}\label{prop:box-anal}
Assume that $\kappa\gg C_0^\Vc:=\int_{\R^d}|k|\V(k)^2\,dk$ is sufficiently large.
Given $1\le m\le m_0$, let $k_{[m]}\in(\R^d)^{m+1}$ be fixed with $\sum_j k_j=0$. Choose $\nu_0,\ldots,\nu_m\in \Sp^{d-1}$ with $\sum_j k_j\cdot \nu_j>0$, and let
\[S_{\nu_{[m]}}^m\,:=\,\big(\R^d-i[0,\tfrac{m}{m_0}]\nu_{0}\big)\times\ldots\times\big(\R^d-i[0,\tfrac{m}{m_0}]\nu_{m}\big).\]
Then, for all $g_m\in C^\infty_c(\hat\Dd^{m+1})$ and $v_{[m]}\in(\R^d)^{m+1}$, the map
\[\xi_{[m]}\mapsto T_{\xi_{[m]}} \big[\WIGdot\big] g_m(\hat z_{[m]})\]
can be extended holomorphically to $\xi_{[m]}\in S_{\nu_{[m]}}^m$,
and this extension satisfies for all $\ell\ge0$,
\begin{eqnarray*}
\Big\|T_{\xi_{[m]}}\big[\WIGdot\big]g_m\Big\|_k
&\lesssim&\Big(\tfrac1{t_N}+\Re\big(ik_{[m]}\cdot\xi_{[m]}\big)\Big)^{-1}\|g_m\|_k,\\
\Big\|\nabla_{v_{[m]}}^{\ell}T_{\xi_{[m]}}\big[\WIGdot\big]g_m\Big\|_k
&\lesssim_\ell&\sum_{s=0}^\ell\langle k_{[m]}\rangle^{s}\bigg(1+\Big(\tfrac1{t_N}+\Re\big(ik_{[m]}\cdot\xi_{[m]}\big)\Big)^{-s-1}\bigg)\|\nabla_{v_{[m]}}^{\ell-s} g_m\|_k.
\end{eqnarray*}
\end{prop}

\begin{proof}[Proof of Propositions~\ref{prop:box-bound} and~\ref{prop:box-anal}]
We use the short-hand notation $\nabla_m:=\nabla_{v_{[m]}}$.
We split the statement of the two propositions into four items, where we further add some commutator estimates that will be useful to iterate on: provided that $\kappa\gg C_0^\Vc$ is sufficiently large, we show for all $0\le m\le m_0$,
\begin{enumerate}[(i)]
\item For all $\hat z_{[m]}$ and $g_m\in C^\infty_c(\hat\Dd^{m+1})$, the map
\[\xi_{[m]}\mapsto T_{\xi_{[m]}}\Big[\WIGdot\Big]g_m(\hat z_{[m]})\]
can be extended holomorphically to $\xi_{[m]}\in S_{\nu_{[m]}}^m$.
\item For all~$g_{m}\in C^\infty_c(\hat\Dd^{m+1})$ and $\xi_{[m]}\in S_{\nu_{[m]}}^m$, we have for all $\ell\ge0$,
\begin{eqnarray*}
\Big\|T_{\xi_{[m]}}\Big[\WIGdot\Big]g_m\Big\|_k
&\lesssim&\Big(\tfrac1{t_N}+\Re\big({\textstyle\sum_{l=0}^mik_l\cdot\xi_l}\big)\Big)^{-1}\|g_m\|_k,\\
\qquad\Big\|\Big[\nabla_{m}^{\ell},T_{\xi_{[m]}}\Big[\WIGdot\Big]\Big]g_m\Big\|_k
&\lesssim_\ell&\sum_{s=1}^\ell\langle k_{[m]}\rangle^{s}\Big(\tfrac1{t_N}+\Re\big({\textstyle\sum_{l=0}^mik_l\cdot\xi_l}\big)\Big)^{-1}\\[-2mm]
&&\hspace{1cm}\times\bigg(1+\Big(\tfrac1{t_N}+\Re\big({\textstyle\sum_{l=0}^mik_l\cdot\xi_l}\big)\Big)^{-s}\bigg)\|\nabla_{m}^{\ell-s} g_m\|_k.
\end{eqnarray*}
\item For all $\hat z_{[m]}$ and $g_m\in C^\infty_c(\hat\Dd^{m+1})$, the map
\[\xi_{[m]}\mapsto T_{\xi_{[m]}}\Big[\sum\WIGBOX\Big]g_m(\hat z_{[m]})\]
can be extended holomorphically to $\xi_{[m]}\in S_{\nu_{[m]}}^m$.
\item For all~$g_{m},h_{m}\in C^\infty_c(\hat\Dd^{m+1})$ and $\xi_{[m]}\in S_{\nu_{[m]}}$, we have for all $\ell\ge0$,
\begin{eqnarray*}
\Big|\Big\<h_m\,,\,T_{\xi_{[m]}}\Big[\sum\WIGBOX\Big]\, g_m\Big\>_k\Big|&\lesssim&C_0^\Vc\|\nabla_m h_m\|\|\nabla_m g_m\|_k,\\
\Big\|\nabla_{m}^{\ell} T_{\xi_{[m]}}\Big[\sum\WIGBOX\Big]\, g_m\Big\|_k&\lesssim_\ell&\sum_{s=0}^{\ell+1}\<k_{[m]}\>^{s}\|\nabla_{m}^{\ell+2-s}g_m\|_k,\\
\Big|\Big\<h_m\,,\,\Big[\nabla_{m}^{\ell},T_{\xi_{[m]}}\Big[\sum\WIGBOX\Big]\Big] g_m\Big\>_k&\lesssim_\ell&\|\nabla_{m}h_m\|_k\sum_{s=1}^{\ell}\<k_{[m]}\>^{s}\|\nabla_{m}^{\ell+1-s}g_m\|_k.
\end{eqnarray*}
\end{enumerate}
These properties immediately imply Propositions~\ref{prop:box-bound} and~\ref{prop:box-anal}. Note that the additional claim~\eqref{eq:posit-box-ren-re} in Proposition~\ref{prop:box-anal} follows from the commutator estimate in~(iv) together with the positivity statement of Lemma~\ref{lem:box-ren}.
Since Definition~\ref{def:renorm} for the renormalized propagator is iterative, we shall naturally prove these four properties~(i)--(iv) by induction, starting with $m=m_0$. We split the proof into five steps.

\medskip\noindent
{\bf Step~1:} Setting the induction: proof of~(i)--(iv) for $m=m_0$.\\
For $m=m_0$, noting that $\WIGBOX=0$ on $\Ld^2(\hat\Dd^{m_0+1})$ as it would involve $>m_0$ background particles, it only remains to check~(i)--(ii).
Recall that on $\Ld^2(\hat\Dd^{m_0+1})$ we have by definition
\[\WIGdot\,=\,{\tiny\begin{tikzpicture}[baseline={([yshift=-.8ex]current bounding box.center)},scale=0.8]
\draw (0,0) -- (0.8,0);
\draw (0,0.45) -- (0.8,0.45);
\draw[dotted] (0.4,0.1) -- (0.4,0.4);
\end{tikzpicture}}\,=\,\bigg(\tfrac{1+i\alpha}{t_N}+\sum_{l=0}^{m_0}ik_l\cdot v_l+\tfrac\kappa N\hat D_{m_0}\bigg)^{-1}.\]
Conjugating with velocity translation, this becomes
\[T_{\xi_{[m_0]}}\Big[\WIGdot\Big]\,=\,\bigg(\tfrac{1+i\alpha}{t_N}+\sum_{l=0}^{m_0}ik_l\cdot (v_l+\xi_l)+\tfrac\kappa N\hat D_{m_0}\bigg)^{-1}.\]
As this transport-diffusion resolvent remains well-defined and bounded for all $\xi_{[m_0]}\in S_{\nu_{[m_0]}}^{m_0}$, since by definition this ensures $\Re(i\sum_lk_l\cdot\xi_l)\ge0$, property~(i) easily follows. The bounds in~(ii) are then obtained similarly as in Step~2 of the proof of Lemma~\ref{lem:box}; we skip the details for shortness.

\medskip
From here, we argue by induction: we assume that there is some $n<m_0$ such that properties~(i)--(iv) hold for all $m$ with $n<m\le m_0$, and we aim to show that the same then also holds for $m=n$. This is the purpose of the next four steps.

\medskip\noindent
{\bf Step~2:} Proof of~(iii) with $m=n$.\\
By definition of collision operators, applying conjugation, for $\xi_{[n]}\in(\R^d)^{n+1}$, we find
\begin{multline}\label{eq:def-rec-Txiwigbox}
T_{\xi_{[n]}}\Big[\sum\WIGBOX\Big]g_{n}(\hat z_{[n]})
\,=\,-\sum_{j=0}^{n}\Div_{v_j}\bigg(\int_{\Dd} (k_{n+1}\otimes k_{n+1})\V(k_{n+1})^2\sqrt M(v_{n+1})\\[-5mm]
\times H_{\hat z_{[n]},\xi_{[n]}}^{n,j}(\hat z_{n+1})\,d^*\hat z_{n+1}\bigg),
\end{multline}
where we have set for abbreviation
\begin{equation}\label{eq:def-Hm1}
H_{\hat z_{[n]},\xi_{[n]}}^{n,j}(\hat z_{n+1})
\,:=\,T_{(\xi_{[n]},-v_j-\xi_j)}\Big[\WIGdot\Big]\nabla_{v_j}g_{n}(\hat z_{[n]})\sqrt M(v_{n+1}).
\end{equation}
By the induction assumption, appealing to~(i) with $m=n+1$, we find that the map
\[v_{n+1}\,\mapsto\, H_{\hat z_{[n]},\xi_{[n]}}^{n,j}(\hat z_{n+1})\,=\,\Big(T_{(\xi_{[n]},v_{n+1}-v_j-\xi_j)}\Big[\WIGdot\Big]\big(\nabla_{v_j}g_{n}\otimes\tau_{v_{n+1}}\sqrt M\big)\Big)(\hat z_{[n]},(k_{n+1},0))\]
can be extended holomorphically to $\R^d-i[0,\frac{n+1}{m_0}]\hat k_{n+1}$. By complex deformation, the integral in~\eqref{eq:def-rec-Txiwigbox} is thus equal to
\begin{multline}\label{eq:form-def-Tbox}
T_{\xi_{[n]}}\Big[\sum\WIGBOX\Big]g_{n}(\hat z_{[n]})
\,=\,-\sum_{j=0}^{n}\Div_{v_j}\bigg(\int_{\Dd} (k_{n+1}\otimes k_{n+1})\V(k_{n+1})^2\sqrt M(v_{n+1}-i\tfrac{n+1}{m_0}\hat k_{n+1})\\[-2mm]
\times T_{(\xi_{[n]},-i\frac{n+1}{m_0}\hat k_{n+1}-v_j-\xi_j)}\Big[\WIGdot\Big]\nabla_{v_j}g_{n}(\hat z_{[n]})\sqrt M(v_{n+1}-i\tfrac{n+1}{m_0}\hat k_{n+1})\,d^*\hat z_{n+1}\bigg).
\end{multline}
For $\xi_{[n]}\in S_{\nu_{[n]}}^n$, with $\nu_{[n]}$ as in the statement, we find
\begin{multline}\label{eq:lower-kn1}
\Re\Big(\sum_{l=0}^nik_l\cdot\xi_l+ ik_{n+1}\cdot(-i\tfrac{n+1}{m_0}\hat k_{n+1}-v_j-\xi_j)\Big)\\[-3mm]
\,\ge\,
\Re\Big(ik_{n+1}\cdot(-i\tfrac{n+1}{m_0}\hat k_{n+1}-v_j-\xi_j)\Big)\,=\,
\tfrac{n+1}{m_0}|k_{n+1}|-\Re(ik_{n+1}\cdot\xi_j)\,\ge\,\tfrac{1}{m_0}|k_{n+1}|.
\end{multline}
Appealing again to~(i) with $m=n+1$, we can conclude that~(iii) holds with $m=n$.

\medskip\noindent
{\bf Step~3:} Proof of~(iv) with $m=n$.\\
Starting from identity~\eqref{eq:form-def-Tbox}, we can write
\begin{multline}
\Big\<h_n\,,\,T_{\xi_{[n]}}\Big[\sum\WIGBOX\Big]g_{n}\Big\>_k
\,=\,\sum_{j=0}^{n}\int_{\R^d}\V(k_{n+1})^2\bigg\<(k_{n+1}\cdot\nabla_{v_j}h_n)\otimes\sqrt M(\cdot+i\tfrac{n+1}{m_0}\hat k_{n+1})\,,\,\\[-2mm]
\times T_{(\xi_{[n]},-i\frac{n+1}{m_0}\hat k_{n+1}-v_j-\xi_j)}\Big[\WIGdot\Big](k_{n+1}\cdot\nabla_{v_j}g_{n})\otimes\sqrt M(\cdot-i\tfrac{n+1}{m_0}\hat k_{n+1})\bigg\>_k\,d^*k_{n+1},
\end{multline}
and thus, by the induction assumption, using~(ii) with $m=n+1$, together with~\eqref{eq:lower-kn1},
\begin{equation}
\Big\<h_n\,,\,T_{\xi_{[n]}}\Big[\sum\WIGBOX\Big]g_{n}\Big\>_k
\,\lesssim\,\|\nabla_{v_j}h_n\|_k \|\nabla_{v_j}g_{n}\|_k\int_{\R^d}|k_{n+1}|\V(k_{n+1})^2\,dk_{n+1},
\end{equation}
that is, the first estimate in~(iv) with $m=n$.
Next, for all $\ell\ge0$, starting again from~\eqref{eq:form-def-Tbox}, we can bound
\begin{multline*}
\Big\|\nabla_n^\ell T_{\xi_{[n]}}\Big[\sum\WIGBOX\Big]g_{n}\Big\|_k
\,\lesssim\,\int_{\Dd} |k_{n+1}|^2\V(k_{n+1})^2\\
\times \Big\|\nabla_n^{\ell+1}\Big( T_{(\xi_{[n]},-i\frac{n+1}{m_0}\hat k_{n+1}-\xi_j)}\Big[\WIGdot\Big]\nabla_{v_j}g_{n}\otimes\sqrt M(\cdot-i\tfrac{n+1}{m_0}\hat k_{n+1})\Big)\Big\|_k dk_{n+1}.
\end{multline*}
By the induction assumption, appealing again to~(ii) with $m=n+1$, together with~\eqref{eq:lower-kn1}, we then obtain
\begin{equation*}
\Big\|\nabla_n^\ell T_{\xi_{[n]}}\Big[\sum\WIGBOX\Big]g_{n}\Big\|_k
\,\lesssim_\ell\,\sum_{s=0}^{\ell+1}\|\nabla_n^{\ell+2-s}g_{n}\|_k\int_{\Dd}\langle k_{[n+1]}\rangle^{s} |k_{n+1}|(1+|k_{n+1}|^{-s})\V(k_{n+1})^2
\,dk_{n+1} ,
\end{equation*}
and the second estimate in~(iv) follows with $m=n$. The last estimate in~(iv) is obtained similarly and we skip the proof for shortness.

\medskip\noindent
{\bf Step~4:} Proof of~(i) for $m=n$.\\
By definition, applying conjugation, for $\xi_{[n]}\in(\R^d)^{n+1}$, we have
\begin{equation*}
T_{\xi_{[n]}}\Big[\WIGdot\Big]g_{n}(\hat z_{[n]})\,=\,\bigg(\tfrac{1+i\alpha}{t_N}+\sum_{l=0}^{n}ik_l\cdot (v_l+\xi_l)+\tfrac\kappa N\hat D_n+\tfrac1NT_{\xi_{[n]}}\Big[\sum\WIGBOX\Big]\bigg)^{-1}g_{n}(\hat z_{[n]}).
\end{equation*}
Now, given $g_n\in C^\infty_c(\hat\Dd^{n+1})$ and $\xi_{[n]}\in S_{\nu_{[n]}}^n$, consider the extended resolvent equation
\begin{equation*}
\bigg(\tfrac{1+i\alpha}{t_N}+\sum_{l=0}^{n}ik_l\cdot (v_l+\xi_l)+\tfrac\kappa N\hat D_n+\tfrac1NT_{\xi_{[n]}}\Big[\sum\WIGBOX\Big]\bigg)h_n\,=\,g_{n}.
\end{equation*}
Testing with $h_n$ and taking the real part, we find
\[\Big(\tfrac1{t_N}+\Re\sum_{l=0}^nik_l\cdot\xi_l\Big)\|h_n\|_k^2+\tfrac\kappa N\|\nabla_nh_n\|^2_k+\tfrac1N\Re\Big\< h_n\,,\,T_{\xi_{[n]}}\Big[\sum\WIGBOX\Big]h_n\Big\>_k\,\le\,\|g_n\|_k\|h_n\|_k,\]
and thus, by the induction assumption, using~(iv) with $m=n$ and $\ell=0$,
\[\Big(\tfrac1{t_N}+\Re\sum_{l=0}^nik_l\cdot\xi_l\Big)\|h_n\|_k^2+\tfrac1N(\kappa-CC_0^\Vc)\|\nabla_nh_n\|^2_k\,\le\,\|g_n\|_k\|h_n\|_k.\]
(Note that we could not simply appeal to the positivity statement of Lemma~\ref{lem:box-ren} unless $\xi_{[n]}\in(\R^d)^n$, which we do not assume.)
Provided that $\kappa\gg C_0^\Vc$ is sufficiently large, the second left-hand side term is nonnegative and we deduce
\[\|h_n\|_k\,\le\,\Big(\tfrac1{t_N}+\Re\sum_{l=0}^nik_l\cdot\xi_l\Big)^{-1}\|g_n\|_k,\]
that is, for all $\xi_{[n]}\in S_{\nu_{[n]}}^n$,
\begin{equation}\label{eq:resolv-wig-0der}
\bigg\|\bigg(\tfrac{1+i\alpha}{t_N}+\sum_{l=0}^{n}ik_l\cdot (v_l+\xi_l)+\tfrac\kappa N\hat D_n+\tfrac1NT_{\xi_{[n]}}\Big[\sum\WIGBOX\Big]\bigg)^{-1}g_n\bigg\|_{k}\,\le\,\Big(\tfrac1{t_N}+\Re\sum_{l=0}^nik_l\cdot\xi_l\Big)^{-1}\|g_n\|_k.
\end{equation}
In addition, by the induction assumption, using~(iii) with $m=n$, we can deduce that this bounded resolvent is holomorphic with respect to $\xi_{[n]}\in S_{\nu_{[n]}}^n$. This concludes~(i) with $m=n$.

\medskip\noindent
{\bf Step~5:} Proof of~(ii) with $m=n$.\\
The first estimate in~(ii) with $m=n$ already follows from~\eqref{eq:resolv-wig-0der}.
Given $\xi_{[n]}\in S_{\nu_{[n]}}^n$ and $g_n\in C^\infty_c(\hat\Dd^{n+1})$, let $h_n:=T_{\xi_{[n]}}\big[\WIGdot\big]g_{n}$, which satisfies by definition the resolvent equation
\begin{equation}\label{eq:hn-eqn}
\bigg(\tfrac{1+i\alpha}{t_N}+\sum_{l=0}^{n}ik_l\cdot (v_l+\xi_l)+\tfrac\kappa N\hat D_n+\tfrac1NT_{\xi_{[n]}}\Big[\sum\WIGBOX\Big]\bigg)h_n\,=\,g_n.
\end{equation}
Applying $\nabla_n^\ell$ to both sides of this equation, testing with $\nabla_n^\ell h_n$, and taking the real part, we obtain
\begin{multline*}
\Big(\tfrac1{t_N}+\Re\sum_{l=0}^nik_l\cdot\xi_l\Big)\|\nabla_n^\ell h_n\|_k^2+\tfrac\kappa N\|\nabla_n^{\ell+1}h_n\|_k^2\\[-2mm]
\,\le\,\|\nabla_n^\ell h_n\|_k\Big(\|\nabla_n^\ell g_n\|_k+\ell|k_{[n]}|\|\nabla_n^{\ell-1}h_n\|_k\Big)-\tfrac1N\Re\Big\<\nabla_n^\ell h_n\,,\,\nabla_n^\ell T_{\xi_{[n]}}\Big[\sum\WIGBOX\Big]h_n\Big\>_k.
\end{multline*}
Decomposing the last term using a commutator, and appealing to~(iv) with $m=n$ (already proven in Step~3), we are led to
\begin{multline*}
\Big(\tfrac1{t_N}+\Re\sum_{l=0}^nik_l\cdot\xi_l\Big)\|\nabla_n^\ell h_n\|_k^2+\tfrac1N(\kappa-CC_0^\Vc)\|\nabla_n^{\ell+1}h_n\|_k^2\\[-3mm]
\,\lesssim_\ell\,\|\nabla_n^\ell h_n\|_k\Big(\|\nabla_n^\ell g_n\|_k+|k_{[n]}|\|\nabla_n^{\ell-1}h_n\|_k\Big)+\tfrac1N\|\nabla_n^{\ell+1} h_n\|_k\sum_{s=1}^\ell\<k_{[n]}\>^s\|\nabla_n^{\ell+1-s}h_n\|_k.
\end{multline*}
Provided that $\kappa\gg C_0^\Vc$ is sufficiently large, and using Young's inequality to absorb the norms of $\nabla_n^\ell h_n$ and $\nabla_n^{\ell+1} h_n$, we obtain
\begin{multline}\label{eq:estim-ell-iter}
\Big(\tfrac1{t_N}+\Re\sum_{l=0}^nik_l\cdot\xi_l\Big)\|\nabla_n^\ell h_n\|_k^2+\tfrac1N\|\nabla_n^{\ell+1}h_n\|_k^2\\[-3mm]
\,\lesssim_\ell\,\Big(\tfrac1{t_N}+\Re\sum_{l=0}^nik_l\cdot\xi_l\Big)^{-1}\Big(\|\nabla_n^\ell g_n\|_k^2+|k_{[n]}|^2\|\nabla_n^{\ell-1}h_n\|_k^2\Big)+\tfrac1N\sum_{s=1}^\ell\<k_{[n]}\>^{2s}\|\nabla_n^{\ell+1-s}h_n\|_k^2.
\end{multline}
Iterating this estimate, we deduce for all $\ell\ge0$,
\begin{multline}\label{eq:estim-nabhn}
\Big(\tfrac1{t_N}+\Re\sum_{l=0}^nik_l\cdot\xi_l\Big)\|\nabla_n^\ell h_n\|_k^2+\tfrac1N\|\nabla_n^{\ell+1} h_n\|_k^2\\[-3mm]
\,\lesssim_\ell\,\Big(\tfrac1{t_N}+\Re\sum_{l=0}^nik_l\cdot\xi_l\Big)^{-1}\sum_{s=0}^\ell\< k_{[n]}\>^{2s}\bigg(1+\Big(\tfrac1{t_N}+\Re\sum_{l=0}^nik_l\cdot\xi_l\Big)^{-2s}\bigg)\|\nabla_n^{\ell-s}g_n\|_k^2.
\end{multline}
With this bound at hand, we can now turn to the proof of the second estimate in~(ii) with $m=n$.
Applying $\nabla_n^\ell$ to both sides of the resolvent equation~\eqref{eq:hn-eqn}, we find that the commutator
\[H_n^\ell\,:=\,\Big[\nabla_n^\ell,T_{\xi_{[n]}}\Big[\WIGdot\Big]\Big]g_n\,=\,\nabla_n^\ell h_n-T_{\xi_{[n]}}\Big[\WIGdot\Big]\nabla_n^\ell g_n\]
satisfies
\begin{multline*}
\bigg(\tfrac{1+i\alpha}{t_N}+\sum_{l=0}^{n}ik_l\cdot (v_l+\xi_l)+\tfrac\kappa N\hat D_n+\tfrac1NT_{\xi_{[n]}}\Big[\sum\WIGBOX\Big]\bigg)H_n^\ell\\[-2mm]
\,=\,-i\ell k_{[n]}\otimes\nabla_n^{\ell-1}h_n-\tfrac1N\Big[\nabla_n^\ell,T_{\xi_{[n]}}\Big[\sum\WIGBOX\Big] \Big]h_n.
\end{multline*}
Testing with $H_n^\ell$, taking the real part, using~(iv) with $m=n$, and arguing as in Step~4, with $\kappa\gg C_0^\Vc$ sufficiently large, we obtain
\begin{multline*}
\Big(\tfrac{1}{t_N}+\Re\sum_{l=0}^{n}ik_l\cdot\xi_l\Big)\|H_n^\ell\|_k^2+\tfrac1N\|\nabla_n H_n^\ell\|_k^2\\[-3mm]
\,\lesssim\,|k_{[n]}|\|H_n^\ell\|_k\|\nabla_n^{\ell-1}h_n\|_k+\tfrac1N\Big|\Big\< H_n^\ell,\Big[\nabla_n^\ell,T_{\xi_{[n]}}\Big[\sum\WIGBOX\Big] \Big]h_n\Big\>_k\Big|,
\end{multline*}
and thus, by the commutator estimate in~(iv) with $m=n$,
\begin{multline*}
\Big(\tfrac{1}{t_N}+\Re\sum_{l=0}^{n}ik_l\cdot\xi_l\Big)\|H_n^\ell\|_k^2+\tfrac1N\|\nabla_n H_n^\ell\|_k^2\\[-6mm]
\,\lesssim_\ell\,|k_{[n]}|\|H_n^\ell\|_k\|\nabla_n^{\ell-1}h_n\|_k+\tfrac1N\|\nabla_nH_n^\ell\|_k\sum_{s=1}^\ell\<k_{[n]}\>^s\|\nabla_n^{\ell+1-s}h_n\|_k.
\end{multline*}
By Young's inequality, we can deduce
\begin{equation*}
\Big(\tfrac{1}{t_N}+\Re\sum_{l=0}^{n}ik_l\cdot\xi_l\Big)\|H_n^\ell\|_k^2
\,\lesssim_\ell\,|k_{[n]}|^2\Big(\tfrac{1}{t_N}+\Re\sum_{l=0}^{n}ik_l\cdot\xi_l\Big)^{-1}\|\nabla_n^{\ell-1}h_n\|_k^2+\tfrac1N\sum_{s=1}^\ell\<k_{[n]}\>^{2s}\|\nabla_n^{\ell+1-s}h_n\|_k^2,
\end{equation*}
and the second estimate in~(ii) with $m=n$ then follows by inserting~\eqref{eq:estim-nabhn}.
This concludes the proof.
\end{proof}

\subsection{Airy-type resolvent estimates for renormalized propagator}\label{sec:PDE-est}
For the transport-diffusion operator $ik\cdot v-\frac\kappa N\triangle_v$ with $O(\frac1N)$ velocity diffusion, we can exploit hypoelliptic estimates and the resulting enhanced dissipation to improve on the naive resolvent estimate
\[\|(\tfrac1{t_N}+ik\cdot v-\tfrac1N\triangle_v)^{-1}g\|_k\,\le\, t_N\|g\|_k.\]
More precisely, we obtain the following Airy-type estimates.

\begin{lem}\label{lem:Airy-res}
We have
\begin{equation}\label{eq:airy}
\sup_{\e>0}\|(\e+ik\cdot v-\tfrac1N\triangle_v)^{-1}g\|_{k}\,\lesssim\,N^\frac13|k|^{-\frac23}\|g\|_{k},
\end{equation}
and for all $\ell\ge1$,
\[\sup_{\e>0}\|\nabla_v^{\ell}(\e+ik\cdot v-\tfrac1N\triangle_v)^{-1}g\|_{k}\,\lesssim_\ell\,|k|^{-1}(N|k|)^\frac{\ell+1}3\|\langle\nabla_v\rangle^{\ell-1}g\|_{k}.\]
\end{lem}
\begin{proof}
The estimate for velocity derivatives of the resolvent follows easily by scaling and by successive energy estimates starting from~\eqref{eq:airy}. We focus on the proof of the latter.
By rotational invariance, we can assume $k=|k|e_1$ and we are then reduced to proving the following resolvent estimate for a one-dimensional operator: for all $\epsilon>0$,
\[\|(\epsilon+i|k|w-\tfrac1N\partial_{w}^2)^{-1}g\|_{L^2(\R)}\,\lesssim\,N^\frac13|k|^{-\frac23}\|g\|_{L^2(\R)}.\]
By scaling, setting $z=(|k|N)^{\frac13}w$ and $\eta=|k|^{-\frac23}N^{\frac13}\epsilon$, we note that
\[(\epsilon+i|k|w-\tfrac1N\partial_{w}^2)^{-1}\,=\,N^{\frac13}|k|^{-\frac23}(\eta+iz-\partial_{z}^2)^{-1}.\]
Hence, it suffices to prove for all $\eta>0$,
\[\|(\eta+iz-\partial_z^2)^{-1}g\|_{L^2(\R)}\,\lesssim\,\|g\|_{L^2(\R)}.\]
Now this is a standard result for the complex Airy operator $iz-\partial_z^2$, which can be checked to hold for instance by an explicit computation in Fourier space.
\end{proof}

We show that essentially the same resolvent estimates remain valid for the renormalized propagators.
Although the $k$-dependence and the loss of derivatives in the bounds below do not match the optimal estimates obtained above, the scaling in~$N$ is identical --- and this is the key feature for our purposes. Our proof relies on a robust PDE argument that applies directly to renormalized propagators but is not sharp enough to recover the optimal $k$-dependence; we do not try to optimize this aspect here.

\begin{prop}\label{prop:pde}
Assume that $\kappa\gg C_0^\Vc$ is sufficiently large.
For all $0\le m\le m_0$, $g_m\in C^\infty_c(\hat\Dd^{m+1})$, and $\ell\ge0$, we have
\begin{equation*}
\|\nabla_{v_{[m]}}^{\ell}\WIGdot\,g_m\|_k
\,\lesssim_\ell\,|k_{[m]}|^{-2} \langle k_{[m]}\rangle^{\ell+4}\sum_{s=0}^{(\ell-1)\vee2}N^{\frac{\ell+1-s}3}\|\<\nabla_{v_{[m]}}\>^sg_m\|_k,
\end{equation*}
\end{prop}

To prepare for the proof of this main proposition, we first establish the following bilinear estimate for the resolvent of transport operators. It is based on the following observation: the Hilbert transform $f\mapsto Hf$ is not a bounded operator $W^{1,1}\to L^\infty$ (it is only bounded $B^1_{1,1}\to L^\infty$), but the related bilinear estimate~\eqref{eq:bnd-Heps} below still holds.

\begin{lem}\label{lem:Hil}
For all $g,h\in C^\infty_c(\R^d)$, $k\in\R^d$, and $\e>0$,
\[\langle h,(\e+ik\cdot v)^{-1}g\rangle\,\lesssim\,|k|^{-1}\Big(\|\<\nabla\> h\|\|g\|+\|h\|\|\nabla g\|\Big).\]
\end{lem}

\begin{proof}
For fixed $k$, recalling the notation $\hat k=\frac{k}{|k|}$, we decompose the integral
\[\langle h,(\e+ik\cdot v)^{-1}g\rangle
\,=\,\int_{\R^d} (\e+ik\cdot v)^{-1}(\overline{h}g)(v)\,dv
\,=\,\int_{\hat k^\bot}\int_{\R} (\e+i|k|s)^{-1}(\overline{h}g)(\hat ks+v')\,dsdv'.\]
In terms of $H_{\e}(s):=(\e+is)^{-1}$, we may then bound
\begin{equation*}
\langle h,(\e+ik\cdot v)^{-1}g\rangle
\,\le\,|k|^{-1}\int_{\hat k^\bot}\|H_{\e |k|^{-1}}\ast((\overline{h}g)(\hat k\cdot+v'))\|_{\Ld^\infty(\R)}\,dv'.
\end{equation*}
It remains to prove the following bound on the Hilbert transform: for all $g,h\in C^\infty_c(\R)$,
\begin{equation}\label{eq:bnd-Heps}
\sup_{\e>0}\|H_\e\ast(hg)\|_{L^\infty(\R)}\lesssim\|h\|_{H^1(\R)}\|g\|_{L^2(\R)}+\|h\|_{L^2(\R)}\|g\|_{H^1(\R)}.
\end{equation}
Let $g,h\in C^\infty_c(\R)$. Decomposing
\[H_\e\ast(hg)(t)\,=\,\int_{\R}\frac{\e}{\e^2+s^2}h(t-s)g(t-s)ds-i\int_{\R}\frac{s}{\e^2+s^2}h(t-s)g(t-s)ds,\]
and using the symmetry around $s=0$ in the second term,
we can bound
\begin{equation}\label{eq:decomp-bnd-Heps}
|H_\e\ast(hg)(t)|\,\le\,\pi\|hg\|_{L^\infty}+\int_{0}^\infty\frac{s}{\e^2+s^2}\Big|h(t+s)g(t+s)-h(t-s)g(t-s)\Big|ds.
\end{equation}
For the first term in this estimate, the Sobolev embedding $W^{1,1}(\R)\subset L^\infty(\R)$ yields
\[\|hg\|_{L^\infty}\,\lesssim\,\|hg\|_{W^{1,1}(\R)}\,\lesssim\,\|h\|_{H^1(\R)}\|g\|_{L^2(\R)}+\|h\|_{L^2(\R)}\|g\|_{H^1(\R)}.\]
For the second term in~\eqref{eq:decomp-bnd-Heps}, we can decompose
\begin{multline*}
\lefteqn{\int_{0}^\infty\frac{s}{\e^2+s^2}\Big|h(t+s)g(t+s)-h(t-s)g(t-s)\Big|ds}\\
\,\le\,\int_{0}^\infty \Big|\frac{h(t+s)-h(t-s)}s\Big||g(t+s)|ds
+\int_{0}^\infty \Big|\frac{g(t+s)-g(t-s)}s\Big||h(t-s)|ds,
\end{multline*}
and thus, noting that
\begin{multline*}
\int_0^\infty\Big|\frac{h(t+s)-h(t-s)}{s}\Big|^2ds
\,=\,\int_0^\infty\Big|\int_{-1}^1 h'(t+su)du\Big|^2ds
\,\le\,\bigg(\int_{-1}^1 \|h'(t+u\cdot)\|_{L^2(\R)}du\bigg)^2\\
\le\,16\|h'\|^2_{L^2(\R)},
\end{multline*}
we obtain
\[\int_{0}^\infty\frac{s}{\e^2+s^2}\Big|h(t+s)g(t+s)-h(t-s)g(t-s)\Big|ds\,\le\,4\|h'\|_{L^2(\R)}\|g\|_{L^2(\R)}+4\|h\|_{L^2(\R)}\|g'\|_{L^2(\R)}.\]
This concludes the proof of~\eqref{eq:bnd-Heps}.
\end{proof}

With the above lemma at hand, we can now turn to the proof of Proposition~\ref{prop:pde}.

\begin{proof}[Proof of Proposition~\ref{prop:pde}]
Given $g_m\in C^\infty_c(\hat\Dd^{m+1})$, let $h_m:=\WIGdot\,g_m$, which satisfies by definition the resolvent equation
\begin{equation}\label{eq:def-wig}
\Big(\tfrac{1+i\alpha}{t_N}+i\hat L_m+\tfrac\kappa N\hat D_{m}+\tfrac{1}{N}\sum\WIGBOX\Big)h_m\,=\,g_m.
\end{equation}
We set for abbreviation $\nabla_m=\nabla_{v_{[m]}}$, and we split the proof into three steps.

\medskip\noindent
{\bf Step~1:} Proof that for all $\ell\ge0$,
\begin{equation}\label{eq:bound-nabg-h}
\|\nabla_m^{\ell}h_m\|_k
\,\lesssim_\ell\,N^\frac\ell3 \<k_{[m]}\>^\ell\|h_m\|_k
+|k_{[m]}|^{-\frac{1}3} \<k_{[m]}\>^{\ell-1} \sum_{s=0}^{\ell-1}N^{\frac{\ell+1-s}3}\|\<\nabla_m\>^sg_m\|_k.
\end{equation}
Testing equation~\eqref{eq:def-wig} with $h_m$, taking the real part, and recalling the positivity of the renormalized hat operator, cf.~Lemma~\ref{lem:box-ren}, we find
\begin{equation*}
\tfrac1{t_N}\|h_m\|_k^2+\tfrac\kappa N\|\nabla_m h_m\|_k^2\le\|h_m\|_k\|g_m\|_k,
\end{equation*}
hence,
\begin{equation}\label{eq:caseell=0}
\|\nabla_m h_m\|_k\lesssim N^\frac12\|h_m\|_k^\frac12\|g_m\|_k^\frac12.
\end{equation}
Next, for $\ell\ge1$, applying $\nabla_m^\ell$ to both sides of equation~\eqref{eq:def-wig}, testing with $\nabla_m^{\ell}h_m$, taking the real part, and recalling $[\nabla_m^\ell,i\hat L_m]=i\ell k_{[m]}\nabla_m^{\ell-1}$, we find
\begin{multline*}
\tfrac1{t_N}\|\nabla_m^\ell h_m\|_k^2+\tfrac\kappa N\|\nabla_m^{\ell+1}h_m\|_k^2
\le\|\nabla_m^\ell h_m\|_k\Big(\|\nabla_m^\ell g_m\|_k+\ell|k_{[m]}|\|\nabla_m^{\ell-1}h_m\|_k\Big)\\
-\tfrac1N\Re\Big\<\nabla_m^\ell h_m\,,\,\sum\nabla_m^\ell \,\WIGBOX\,h_m\Big\>_k,
\end{multline*}
and thus, using~\eqref{eq:posit-box-ren-re} in Proposition~\ref{prop:box-bound},
\begin{equation*}
\tfrac1N\|\nabla_m^{\ell+1}h_m\|_k^2
\,\lesssim_\ell\,\|\nabla_m^\ell h_m\|_k\Big(\|\nabla_m^\ell g_m\|_k+|k_{[m]}|\|\nabla_m^{\ell-1}h_m\|_k\Big)
+\tfrac{1}{N}\|\nabla_m^{\ell+1}h_m\|_k\sum_{s=1}^\ell \<k_{[m]}\>^s\|\nabla_m^{\ell+1-s}h_m\|_k.
\end{equation*}
By Young's inequality, we are led to
\begin{equation*}
\|\nabla_m^{\ell+1}h_m\|_k
\,\lesssim_\ell\, N^\frac12\|\nabla_m^\ell h_m\|_k^\frac12\Big(\|\nabla_m^\ell g_m\|_k^\frac12
+|k_{[m]}|^\frac12\|\nabla_m^{\ell-1}h_m\|_k^\frac12\Big)
+\sum_{s=1}^\ell \<k_{[m]}\>^s\|\nabla_m^{\ell+1-s}h_m\|_k.
\end{equation*}
Now iterating this estimate to get rid of the last right-hand side term, starting from~\eqref{eq:caseell=0}, we obtain for all~$\ell\ge0$,
\begin{equation}\label{eq:rec-nabg}
\|\nabla_m^{\ell+1}h_m\|_k\,\lesssim_\ell\,N^{\frac12}\sum_{s=0}^\ell \<k_{[m]}\>^{\ell-s}\|\nabla_m^sh_m\|_k^\frac12\Big(\|\nabla_m^sg_m\|_k^\frac12+\mathds1_{s\ge1}|k_{[m]}|^\frac12\|\nabla_m^{s-1}h_m\|_k^\frac12\Big).
\end{equation}
Repeated application of these bounds allows us to control gradients of $g_m$ in terms of $g_m$ itself and the gradients of $h_m$. For instance, for the first three derivatives, we find
\begin{eqnarray}
\|\nabla_mh_m\|_k
&\lesssim& N^\frac12\|h_m\|_k^\frac12\|g_m\|_k^\frac12,\nonumber\\
\|\nabla_m^2h_m\|_k
&\lesssim&N^\frac12\<k_{[m]}\>\|h_m\|_k^\frac12\|g_m\|_k^\frac12
\nonumber\\
&&+N^\frac34\|h_m\|_k^\frac14\|g_m\|_k^\frac14\Big(\|\nabla_mg_m\|_k^\frac12+|k_{[m]}|^\frac12\|h_m\|_k^\frac12\Big),
\nonumber\\
\|\nabla_m^3h_m\|_k
&\lesssim&N^\frac12\<k_{[m]}\>^2\|h_m\|_k^\frac12\|g_m\|_k^\frac12
\nonumber\\
&&+N^\frac34\<k_{[m]}\>\|h_m\|_k^\frac14\|g_m\|_k^\frac14\Big(\|\nabla_mg_m\|_k^\frac12+|k_{[m]}|^\frac12\|h_m\|_k^\frac12\Big)
\nonumber\\
&&+N^\frac34\<k_{[m]}\>^\frac12\|h_m\|_k^\frac14\|g_m\|_k^\frac14\Big(\|\nabla_m^2g_m\|_k^\frac12+N^\frac14|k_{[m]}|^\frac12\|h_m\|_k^\frac14\|g_m\|_k^\frac14\Big)
\nonumber\\
&&+N^\frac78\|h_m\|_k^\frac18\|g_m\|_k^\frac18\Big(\|\nabla_mg_m\|_k^\frac14+|k_{[m]}|^\frac14\|h_m\|_k^\frac14\Big)
\nonumber\\
&&\hspace{3cm}\times\Big(\|\nabla_m^2g_m\|_k^\frac12+N^\frac14|k_{[m]}|^\frac12\|h_m\|_k^\frac14\|g_m\|_k^\frac14\Big).
\label{eq:estim-nabg-dir}
\end{eqnarray}
While direct iterations of~\eqref{eq:rec-nabg} quickly become cumbersome, as we see here, let us reorganize it into a more tractable form. First note that it implies for all $\ell\ge0$,
\begin{multline*}
\sum_{s=0}^{\ell+1}\<k_{[m]}\>^{-s}\|\nabla_m^{s}h_m\|_k\\[-5mm]
\,\lesssim_\ell\,\|h_m\|_k+N^{\frac12}\<k_{[m]}\>^{-1}\sum_{s=0}^\ell \<k_{[m]}\>^{-s}\|\nabla_m^sh_m\|_k^\frac12\Big(\|\nabla_m^sg_m\|_k^\frac12+\mathds1_{s\ge1}|k_{[m]}|^\frac12\|\nabla_m^{s-1}h_m\|_k^\frac12\Big).
\end{multline*}
Using the Cauchy-Schwarz inequality and reorganizing the factors, this yields
\begin{multline*}
\Big(\frac{N|k_{[m]}|}{\<k_{[m]}\>^{3}}\Big)^{-\frac{\ell+1}3}\sum_{s=0}^{\ell+1}\<k_{[m]}\>^{-s}\|\nabla_m^{s}h_m\|_k
\,\lesssim_\ell\,\Big(\frac{N|k_{[m]}|}{\<k_{[m]}\>^{3}}\Big)^{-\frac{\ell+1}3}\|h_m\|_k\\
+N^{\frac16}|k_{[m]}|^{-\frac{1}3}\bigg(\Big(\frac{N|k_{[m]}|}{\<k_{[m]}\>^{3}}\Big)^{-\frac{\ell}3}\sum_{s=0}^\ell \<k_{[m]}\>^{-s}\|\nabla_m^sh_m\|_k\bigg)^\frac12\bigg(\Big(\frac{N|k_{[m]}|}{\<k_{[m]}\>^{3}}\Big)^{-\frac{\ell}3}\sum_{s=0}^\ell \<k_{[m]}\>^{-s}\|\nabla_m^sg_m\|_k\bigg)^\frac12\\
+\bigg(\Big(\frac{N|k_{[m]}|}{\<k_{[m]}\>^{3}}\Big)^{-\frac{\ell}3}\sum_{s=0}^\ell \<k_{[m]}\>^{-s}\|\nabla_m^sh_m\|_k\bigg)^\frac12\bigg(\Big(\frac{N|k_{[m]}|}{\<k_{[m]}\>^{3}}\Big)^{-\frac{\ell-1}3}\sum_{s=0}^{\ell-1} \<k_{[m]}\>^{-s}\|\nabla_m^{s}h_m\|_k\bigg)^\frac12,
\end{multline*}
and thus, by Young's inequality,
\begin{multline*}
\Big(\frac{N|k_{[m]}|}{\<k_{[m]}\>^{3}}\Big)^{-\frac{\ell+1}3}\sum_{s=0}^{\ell+1}\<k_{[m]}\>^{-s}\|\nabla_m^{s}h_m\|_k\\
\,\lesssim_\ell\,\Big(\frac{N|k_{[m]}|}{\<k_{[m]}\>^{3}}\Big)^{-\frac{\ell+1}3}\|h_m\|_k
+N^{\frac13}|k_{[m]}|^{-\frac{2}3}\Big(\frac{N|k_{[m]}|}{\<k_{[m]}\>^{3}}\Big)^{-\frac{\ell}3}\sum_{s=0}^\ell \<k_{[m]}\>^{-s}\|\nabla_m^sg_m\|_k\\
+\Big(\frac{N|k_{[m]}|}{\<k_{[m]}\>^{3}}\Big)^{-\frac{\ell}3}\sum_{s=0}^\ell \<k_{[m]}\>^{-s}\|\nabla_m^sh_m\|_k 
+\Big(\frac{N|k_{[m]}|}{\<k_{[m]}\>^{3}}\Big)^{-\frac{\ell-1}3}\sum_{s=0}^{\ell-1} \<k_{[m]}\>^{-s}\|\nabla_m^{s}h_m\|_k.
\end{multline*}
From here, we can proceed to a direct iteration and we conclude for all $\ell\ge0$,
\begin{multline*}
\|\nabla_m^{\ell+1}h_m\|_k
\,\lesssim_\ell\,\<k_{[m]}\>^{\ell+1}\sum_{s=0}^{\ell+1}\Big(\frac{N|k_{[m]}|}{\<k_{[m]}\>^{3}}\Big)^{\frac{\ell+1-s}3}\|h_m\|_k\\[-4mm]
+N^{\frac23}|k_{[m]}|^{-\frac{1}3}\<k_{[m]}\>^{\ell}\sum_{v=0}^{\ell}\Big(\frac{N|k_{[m]}|}{\<k_{[m]}\>^{3}}\Big)^{\frac{\ell-v}3}\sum_{s=0}^{v} \<k_{[m]}\>^{-s}\|\nabla_m^sg_m\|_k.
\end{multline*}
The claim~\eqref{eq:bound-nabg-h} follows.

\medskip\noindent
{\bf Step~2:} Proof that
\begin{multline}\label{eq:estim-g0}
\|h_m\|_k^2
\,\lesssim\,|k_{[m]}|^{-1}\Big(\|\nabla_m h_m\|_k\|g_m\|_k+\|h_m\|_k\|\<\nabla_m\> g_m\|_k\Big)\\
+N^{-1}|k_{[m]}|^{-1}\|\nabla_mh_m\|_k\Big(\|\nabla_m^2h_m\|_k+\<k_{[m]}\>\|\nabla_mh_m\|_k\Big)\\
+N^{-1}|k_{[m]}|^{-1}\|h_m\|_k\Big(\|\nabla_m^3h_m\|_k+\<k_{[m]}\>\|\nabla_m^2h_m\|_k+\<k_{[m]}\>^2\|\nabla_mh_m\|_k\Big).
\end{multline}
The resolvent equation~\eqref{eq:def-wig} can be reorganized as
\begin{equation*}
h_m=\big(\tfrac{1+i\alpha}{t_N}+i \hat L_m \big)^{-1}g_m-\tfrac1N\big(\tfrac{1+i\alpha}{t_N}+i\hat L_m\big)^{-1}\Big(\kappa\hat D_m+\sum\WIGBOX\Big)h_m.
\end{equation*}
By Lemma~\ref{lem:Hil}, we may then bound
\begin{eqnarray*}
\|h_m\|_k^2
&=&\Big\<h_m,\big(\tfrac{1+i\alpha}{t_N}+i \hat L_m \big)^{-1}g_m\Big\>_k
-\tfrac1N\Big\<h_m,\big(\tfrac{1+i\alpha}{t_N}+i \hat L_m \big)^{-1}\Big(\kappa\hat  D_m+\sum \WIGBOX\Big)h_m\Big\>_k\\
&\lesssim&|k_{[m]}|^{-1}\Big(\|\nabla_m h_m\|_k\|g_m\|_k+\|h_m\|_k\|\<\nabla_m\> g_m\|_k\Big)\\
&&+N^{-1}|k_{[m]}|^{-1}\|\nabla_mh_m\|_k\Big\|\Big(\kappa\hat D_m+\sum \WIGBOX\Big)h_m\Big\|_k\\
&&+N^{-1}|k_{[m]}|^{-1}\|h_m\|_k\Big\|\<\nabla_m\>\Big(\kappa\hat D_m+\sum \WIGBOX\Big)h_m\Big\|_k,
\end{eqnarray*}
and the claim~\eqref{eq:estim-g0} follows using Proposition~\ref{prop:box-bound}.

\medskip\noindent
{\bf Step~3:} Conclusion.\\
Combining the result~\eqref{eq:estim-g0} of Step~2 with the bounds~\eqref{eq:estim-nabg-dir} of Step~1 on $\nabla _mh_m,\nabla^2_mh_m,\nabla^3_mh_m$, and expanding all the terms, we obtain
\begin{eqnarray*}
\|h_m\|_k^2
&\lesssim&
N^\frac12|k_{[m]}|^{-1}\|h_m\|_k^\frac12\|g_m\|_k^\frac32
+N^\frac14|k_{[m]}|^{-\frac12}\|h_m\|_k^\frac54\|g_m\|_k^\frac34
+|k_{[m]}|^{-1}\<k_{[m]}\>\|h_m\|_k\|g_m\|_k\\
&+&|k_{[m]}|^{-\frac12}\<k_{[m]}\>^\frac12\|h_m\|_k^\frac32\|g_m\|_k^\frac12
+N^{\frac18}|k_{[m]}|^{-\frac14}\|h_m\|_k^\frac{13}8\|g_m\|_k^\frac38
+N^{-\frac14}|k_{[m]}|^{-\frac12}\<k_{[m]}\>\|h_m\|_k^\frac74\|g_m\|_k^\frac14\\
&+&N^{-\frac12}|k_{[m]}|^{-1}\<k_{[m]}\>^2\|h_m\|_k^\frac32\|g_m\|_k^\frac12
+N^\frac14|k_{[m]}|^{-1}\|h_m\|_k^\frac34\|g_m\|_k^\frac34\|\nabla_mg_m\|_k^\frac12\\
&+&N^{\frac18}|k_{[m]}|^{-\frac12}\|h_m\|_k^\frac{11}8\|g_m\|_k^\frac38\|\nabla_mg_m\|_k^\frac14
+|k_{[m]}|^{-1}\|h_m\|_k\|\nabla_m g_m\|_k\\
&+&N^{-\frac14}|k_{[m]}|^{-1}\<k_{[m]}\>\|h_m\|_k^\frac54\|g_m\|_k^\frac14\|\nabla_mg_m\|_k^\frac12
+N^{-\frac18}|k_{[m]}|^{-1}\|h_m\|_k^\frac98\|g_m\|_k^\frac18\|\nabla_mg_m\|_k^\frac14\|\nabla_m^2g_m\|_k^\frac12\\
&+&N^{-\frac18}|k_{[m]}|^{-\frac34}\|h_m\|_k^\frac{11}8\|g_m\|_k^\frac18\|\nabla_m^2g_m\|_k^\frac12
+N^{-\frac14}|k_{[m]}|^{-1}\<k_{[m]}\>^\frac12\|h_m\|_k^\frac54\|g_m\|_k^\frac14\|\nabla_m^2g_m\|_k^\frac12.
\end{eqnarray*}
By Young's inequality, examining all the different terms, we deduce
\begin{equation*}
\|h_m\|_k
\,\lesssim\,
\Big(N^\frac13|k_{[m]}|^{-\frac23}
+N^{-1}|k_{[m]}|^{-2}\<k_{[m]}\>^4\Big)\|g_m\|_k\\
+|k_{[m]}|^{-1}\|\nabla_m g_m\|_k
+N^{-\frac13}|k_{[m]}|^{-\frac43}\|\nabla_m^2g_m\|_k.
\end{equation*}
Finally, combining this with the result~\eqref{eq:bound-nabg-h} of Step~1, the conclusion follows.
\end{proof}

\section{Ansatz for the tagged particle density}\label{sec:g0}
In this section, we prove regularity estimates for the ansatz $\tilde g_0^N$ for the tagged particle density and we show that it converges to the solution $g_0$ of the desired diffusion equation~\eqref{eq:FP-deriv} as $N\uparrow\infty$. This provides the rigorous counterpart of the formal computations in Section~\ref{sec:formal} starting from~\eqref{eq:solve-trunc-hier}.

\subsection{A priori estimates}
Based on the properties of the hat operator established in Lemma~\ref{lem:box}, we easily deduce propagation of regularity for the ansatz $\tilde g_0^N$.

\begin{lem}\label{lem:estim-g0N}
The ansatz $\tilde g_0^N$ defined in~\eqref{eq:def-tilg0-re} satisfies for all $\ell\ge0$,
\begin{equation*}
\sup_{\tau\ge0}\Big(e^{-\tau}\|\nabla_{v_0}^\ell\tilde g_0^N(\tau)\|\Big)
+\3\!\nabla_{v_0}^\ell\Lc\tilde g_0^N\!\3
+(\tfrac{t_N}{N})^\frac12\3\!\nabla_{v_0}^{\ell+1}\Lc\tilde g_0^N\!\3
\,\lesssim_\ell\,\|\langle\nabla_{v_0}\rangle^\ell\mathfrak g\|.
\end{equation*}
\end{lem}

\begin{proof}
By definition of $\tilde g_0^N$, we have
\[\partial_\tau \tilde g_0^N+\kappa\frac{t_N}N\hat D_0\tilde g_0^N
\,=\,-\Big(\frac{t_N}{\sqrt N}\Big)^2\int_0^{\tau}
\hat S_0^+e^{-t_N(\tau-\tau_1)(i\hat L_1+\frac\kappa N\hat D_1)}\hat S_1^-g_{0}^N(\tau_1)\,d\tau_1,\]
or equivalently,
\[(\partial_\tau+1)(e^{-\tau} \tilde g_0^N)+\kappa \frac{t_N}N\hat D_0(e^{-\tau}\tilde g_0^N)
\,=\,-\Big(\frac{t_N}{\sqrt N}\Big)^2e^{-\tau}\int_0^{\tau}
\hat S_0^+e^{-t_N(\tau-\tau_1)(i\hat L_1+\frac\kappa N\hat D_1)}\hat S_1^-g_{0}^N(\tau_1)\,d\tau_1.\]
For $\ell\ge0$, applying $\nabla_{v_0}^\ell$ to both sides of this equation, the energy identity yields
\begin{multline*}
\partial_\tau\|\nabla_{v_0}^\ell(e^{-\tau} \tilde g_0^N)\|^2+2\|\nabla_{v_0}^\ell(e^{-\tau} \tilde g_0^N)\|^2
+2\kappa \frac{t_N}N\|\nabla_{v_0}^{\ell+1}(e^{-\tau}\tilde g_0^N)\|^2\\
\,=\,2\Big(\frac{t_N}{\sqrt N}\Big)^2\Re\left\< \nabla_{v_0}^\ell(e^{-\tau} \tilde g_0^N),e^{-\tau} \nabla_{v_0}^\ell\int_0^{\tau}
i\hat S_0^+e^{-t_N(\tau-\tau_1)(i\hat L_1+\frac\kappa N\hat D_1)}i\hat S_1^-g_{0}^N(\tau_1)\,d\tau_1\right\>,
\end{multline*}
and thus, for all~$T\ge0$, integrating over $\tau\in[0,T]$ and using Plancherel's theorem for the Laplace transform,
\begin{multline*}
\|\nabla_{v_0}^\ell(e^{-T} \tilde g_0^N(T))\|^2
+2\int_\R\|\nabla_{v_0}^\ell\Lc(\mathds1_{[0,T]}\tilde g_0^N)(\alpha)\|^2d^*\alpha
+2\kappa\frac{t_N}N\int_\R\|\nabla_{v_0}^{\ell+1}\Lc(\mathds1_{[0,T]}\tilde g_0^N)(\alpha)\|^2d^*\alpha\\
\,=\,\|\nabla_{v_0}^\ell\mathfrak g\|^2
-2\frac{t_N}{N}\int_\R\Re\Big\< \nabla_{v_0}^\ell\Lc(\mathds1_{[0,T]} \tilde g_0^N)(\alpha),\nabla_{v_0}^\ell\BOX\,\Lc(\mathds1_{[0,T]}g_{0}^N)(\alpha)\Big\>d^*\alpha.
\end{multline*}
Now appealing to the result~\eqref{eq:cor-box} of Lemma~\ref{lem:box}, we get
\begin{multline*}
\|\nabla_{v_0}^\ell(e^{-T} \tilde g_0^N(T))\|^2
+\int_\R\|\nabla_{v_0}^\ell\Lc(\mathds1_{[0,T]}\tilde g_0^N)(\alpha)\|^2d\alpha
+\kappa\frac{t_N}N\int_\R\|\nabla_{v_0}^{\ell+1}\Lc(\mathds1_{[0,T]}\tilde g_0^N)(\alpha)\|^2\,d\alpha\\
\,\lesssim_\ell\,\|\nabla_{v_0}^\ell\mathfrak g\|^2
+\frac{t_N}{N}\sum_{m=1}^n\int_\R\|\nabla_{v_0}^{\ell+1}\Lc(\mathds1_{[0,T]} \tilde g_0^N)(\alpha)\|\|\nabla_{v_0}^m\Lc(\mathds1_{[0,T]}g_{0}^N)(\alpha)\|\,d\alpha,
\end{multline*}
hence, by Young's inequality,
\begin{multline*}
\|\nabla_{v_0}^\ell(e^{-T} \tilde g_0^N(T))\|^2
+\int_\R\|\nabla_{v_0}^\ell\Lc(\mathds1_{[0,T]}\tilde g_0^N)(\alpha)\|^2d\alpha
+\kappa\frac{t_N}N\int_\R\|\nabla_{v_0}^{\ell+1}\Lc(\mathds1_{[0,T]}\tilde g_0^N)(\alpha)\|^2\,d\alpha\\
\,\lesssim_\ell\,\|\nabla_{v_0}^\ell\mathfrak g\|^2
+\frac{t_N}{N}\sum_{m=1}^\ell\int_\R\|\nabla_{v_0}^m\Lc(\mathds1_{[0,T]}g_{0}^N)(\alpha)\|^2\,d\alpha.
\end{multline*}
By a direct induction over $\ell$, we deduce for all $\ell,T\ge0$,
\begin{multline*}
\|\nabla_{v_0}^\ell(e^{-T} \tilde g_0^N(T))\|^2\\[-2mm]
+\int_\R\|\nabla_{v_0}^\ell\Lc(\mathds1_{[0,T]}\tilde g_0^N)(\alpha)\|^2d\alpha
+\kappa\frac{t_N}N\int_\R\|\nabla_{v_0}^{\ell+1}\Lc(\mathds1_{[0,T]}\tilde g_0^N)(\alpha)\|^2\,d\alpha
\,\lesssim_\ell\,\sum_{m=0}^\ell\|\nabla_{v_0}^m\mathfrak g\|^2,
\end{multline*}
and the conclusion follows, letting $T\uparrow\infty$ in the last two left-hand side terms.
\end{proof}

\subsection{Kinetic limit}
We show that the ansatz $\tilde g_0^N$ converges to the solution $g_0$ of the desired diffusion equation~\eqref{eq:FP-deriv} as $N\uparrow\infty$.

\begin{lem}\label{lem:kin-lim}
Assume $\mathfrak g\in H^2(\R^d)$, consider the ansatz $\tilde g_0^N$ defined in~\eqref{eq:def-tilg0-re}, with $t_N=N$, and consider the solution $g_0$ of the diffusion equation~\eqref{eq:FP-deriv}. Then we have
\begin{equation*}
\3\Lc\tilde g_0^N-\Lc g_0\3 \,\lesssim\,\frac1N\|\<\nabla_{v_0}\>^2\mathfrak g\|.
\end{equation*}
\end{lem}

\begin{proof}
For $t_N=N$, the equation~\eqref{eq:def-tilg0-re} for the ansatz reads
\begin{equation*}
\big(1+i\alpha-\kappa\triangle_{v_0}+\BOX\big)\Lc\tilde g_0^N\,=\,\mathfrak g.
\end{equation*}
Taking Laplace transform, the limit equation~\eqref{eq:FP-deriv} for $g_0$ reads
\begin{equation}\label{eq:FP-deriv-re}
(1+i\alpha-\kappa\triangle_{v_0})\Lc g_0-\Div_{v_0}(A_0\nabla_{v_0}\Lc g_0)=\mathfrak g,
\end{equation}
where we recall that $A_0$ is defined in~\eqref{eq:coeff-LB-re}.
The error $\tilde g_0^N-g_0$ thus satisfies
\begin{equation*}
\big(1+i\alpha-\kappa \triangle_{v_0}+\BOX\big)(\Lc\tilde g_0^N-\Lc g_0)\,=\,-\Div_{v_0}(A_0\nabla_{v_0}\Lc g_0)-\BOX\Lc g_0,
\end{equation*}
and therefore, by the positivity of the hat operator, cf.~Lemma~\ref{lem:box},
\begin{equation}\label{eq:eq a conv dom}
\3\Lc\tilde g_0^N-\Lc g_0\3 \,\le\, \3\Div_{v_0}(A_0\nabla_{v_0}\Lc g_0)+\BOX\Lc g_0\3
\end{equation}
Arguing as in Step~2 of the proof of Lemma~\ref{lem:box}, using complex deformation, we note that
\begin{equation}\label{eq:eq a conv dom-reexpr}
\BOX\Lc g_0\,=\,-\Div_{v_0}(H_N),
\end{equation}
in terms of
\begin{multline*}
H_N(\alpha,v_0)\,:=\,\int_{(\R^d)^2}(k\otimes k)\V(k)^2\sqrt M(v_1-i\hat k)\\[-2mm]
\times\Big(\tfrac{1+i\alpha}{N}+|k|+ik\cdot(v_1-v_0)-\tfrac\kappa N\triangle_{v_{[1]}}\Big)^{-1}\sqrt M(v_1-i\hat k)\,\nabla_{v_0}\Lc g_0(\alpha,v_0)\,d^*kdv_1.
\end{multline*}
In terms of the matrix field
\begin{equation}\label{eq:def-Dmat}
D(v_0)\,:=\,\int_{(\R^d)^2}(k\otimes k)\V(k)^2\frac{\big(\sqrt M(v_1-i\hat k)\big)^2}{|k|+ik\cdot(v_1-v_0)}\,d^*kdv_1,
\end{equation}
the resolvent identity ensures
\[\3\<\nabla_{v_0}\>(H_N-D\nabla_{v_0}\Lc g_0)\3\,\lesssim\,\frac1N\3\<(\alpha,\nabla_{v_0})\>\<\nabla_{v_0}\>\nabla_{v_0}\Lc g_0\3.\]
Using equation~\eqref{eq:FP-deriv-re} to bound $\alpha\Lc g_0$ in terms of gradients of $\Lc g_0$, recalling the bound of Lemma~\ref{lem:box} on the hat operator, and appealing to the a priori estimates of Lemma~\ref{lem:estim-g0N}, provided $\mathfrak g\in H^2(\R^d)$, we conclude
\[\3\<\nabla_{v_0}\>(H_N-D\nabla_{v_0}\Lc g_0)\3\,\lesssim\,\frac1N\|\<\nabla_{v_0}\>^2\mathfrak g\|.\]
Combining this with~\eqref{eq:eq a conv dom} and~\eqref{eq:eq a conv dom-reexpr}, we obtain
\begin{equation}\label{eq:conv-g0N-g0-pre}
\3\Lc\tilde g_0^N-\Lc g_0\3 \,\lesssim\, \3\Div_{v_0}((A_0-D)\nabla_{v_0}\Lc g_0)\3+\frac1N\|\<\nabla_{v_0}\>^2\mathfrak g\|.
\end{equation}
Let us now examine the matrix field $D$ defined in~\eqref{eq:def-Dmat}. Splitting the $v_1$-integral over $\R^d=\hat k\R\otimes\hat k^\bot$, and noting that for a holomorphic function $f$ we have by complex deformation and by the Plemelj formula
\[\int_\R \frac{f(s+i)}{s+i}\,ds\,=\,-i\pi f(0)+\text{p.v.}\int_\R\frac{f(s)}{s}\,ds,\]
we find
\begin{multline*}
D(v_0)\,=\,\pi\int_{(\R^d)^2}(k\otimes k)\V(k)^2 M(v_1)\,\delta( k\cdot(v_1-v_0))\,d^*kdv_1\\[-3mm]
+i\int_{(\R^d)^2}(k\otimes k)\V(k)^2\frac{M(v_1+\hat k(\hat k\cdot v_0))}{k\cdot v_1}\,d^*kdv_1.
\end{multline*}
As $\Vc$ is real-valued, we have $\V(k)=\V(-k)$, which ensures that the last integral vanishes by symmetry, hence
\begin{eqnarray*}
	D(v_0)\,=\,\pi\int_{(\R^d)^2}(k\otimes k)\V(k)^2 M(v_1)\,\delta( k\cdot(v_1-v_0))\,d^*kdv_1.
\end{eqnarray*}
By a direct computation of the $k$-integral, see e.g.~\cite[(2.12)]{DW-23}, we are then led to $D=A_0$ with $A_0$ defined in~\eqref{eq:coeff-LB-re}. This shows that the first right-hand side term in~\eqref{eq:conv-g0N-g0-pre} vanishes, and the conclusion follows.
\end{proof}

\section{Error analysis}
This final section is devoted to estimates on the remainder terms~\eqref{eq:RN0}--\eqref{eq:RNm} in the approximate hierarchy satisfied by the ansatz, which then allows to conclude the proof of Theorem~\ref{th:main}.
In the case $m_0=1$, note that $g_0^{N,1}$ coincides exactly with the ansatz $\tilde g_0^N$ defined in~\eqref{eq:def-tilg0-re}, so that Theorem~\ref{th:main} follows directly from Lemma~\ref{lem:kin-lim}. Below, we explicitly treat the case $m_0=2$ for illustration, before turning to the general case.

\subsection{Case  $m_0=2$}\label{sec:case-m02}
Recall that the ansatz $\tilde g_0^{N,2}=\tilde g_0^N$ defined in~\eqref{eq:def-tilg0-re} is independent of the truncation parameter.
For higher-order cumulants, we consider the ansatz~\eqref{eq:def-tilgm} where we choose the admissible history set $\Omega := \{(-1),(-1,-1)\}$. Using diagrammatic notation, this means
\begin{equation}\label{eq:tildeg0/3+}
\left\{\begin{array}{l}
\big(1+i\alpha+\kappa\tfrac{t_N}N\hat D_0+\tfrac{t_N}N\BOX\big)\Lc\tilde g_0^N\,=\,\mathfrak g,\\[2mm]
\Lc\tilde g_1^N\,=\,
\frac{i}{\sqrt N}\,{\tiny\begin{tikzpicture}[baseline={([yshift=-.8ex]current bounding box.center)},scale=0.8]
\begin{scope}[every node/.style={circle,draw,fill=white,inner sep=0pt,minimum size=3pt}]
\node (1) at (0.5,0) {};
\end{scope}
\draw[decorate,decoration={snake,amplitude=1pt,segment length=3pt}] (0,0) -- (1);
\draw[decorate,decoration={snake,amplitude=1pt,segment length=3pt}] (0,0.3) -- (0.5,0.3);
\draw (0.5,0.3) -- (1);
\end{tikzpicture}}
\,\Lc\tilde g_0^N,\\[2mm]
\Lc\tilde g_2^{N,2}
\,=\,
\big(\frac{i}{\sqrt N}\big)^2
\Big({\tiny\begin{tikzpicture}[baseline={([yshift=-.8ex]current bounding box.center)},scale=0.8]
\begin{scope}[every node/.style={circle,draw,fill=white,inner sep=0pt,minimum size=3pt}]
\node (1) at (0.5,0.25) {};
\node (2) at (1,0) {};
\end{scope}
\draw (0,0.5) -- (0.5,0.5) -- (1) -- (0,0.25);
\draw[decorate,decoration={snake,amplitude=1pt,segment length=3pt}] (1) -- (1,0.25);
\draw (0,0) -- (0.5,0);
\draw[decorate,decoration={snake,amplitude=1pt,segment length=3pt}] (0.5,0) -- (2);
\draw (1,0.25) -- (2);
\end{tikzpicture}}+
\,{\tiny\begin{tikzpicture}[baseline={([yshift=-.8ex]current bounding box.center)},scale=0.8]
\begin{scope}[every node/.style={circle,draw,fill=white,inner sep=0pt,minimum size=3pt}]
\node (1) at (0.5,0) {};
\node (2) at (1,0) {};
\end{scope}
\draw (0,0.5) -- (0.5,0.5) -- (1);
\draw (0,0.25) -- (0.45,0.25);
\draw[decorate,decoration={snake,amplitude=1pt,segment length=3pt}] (0.55,0.25) -- (1,0.25);
\draw (0,0) -- (1);
\draw[decorate,decoration={snake,amplitude=1pt,segment length=3pt}] (1) -- (2);
\draw (1,0.25) -- (2);
\end{tikzpicture}}+
\,{\tiny\begin{tikzpicture}[baseline={([yshift=-.8ex]current bounding box.center)},scale=0.8]
\begin{scope}[every node/.style={circle,draw,fill=white,inner sep=0pt,minimum size=3pt}]
\node (1) at (0.5,0.5) {};
\node (2) at (1,0) {};
\end{scope}
\draw (0,0.25) -- (0.5,0.25) -- (1) -- (0,0.5);
\draw[decorate,decoration={snake,amplitude=1pt,segment length=3pt}] (1) -- (1,0.5);
\draw (0,0) -- (0.5,0);
\draw[decorate,decoration={snake,amplitude=1pt,segment length=3pt}] (0.5,0) -- (2);
\draw (1,0.5) -- (2);
\end{tikzpicture}}+
\,{\tiny\begin{tikzpicture}[baseline={([yshift=-.8ex]current bounding box.center)},scale=0.8]
\begin{scope}[every node/.style={circle,draw,fill=white,inner sep=0pt,minimum size=3pt}]
\node (1) at (0.5,0) {};
\node (2) at (1,0) {};
\end{scope}
\draw (0,0.25) -- (0.5,0.25) -- (1) -- (0,0);
\draw (0,0.5) -- (0.5,0.5);
\draw[decorate,decoration={snake,amplitude=1pt,segment length=3pt}] (0.5,0.5) -- (1,0.5);
\draw[decorate,decoration={snake,amplitude=1pt,segment length=3pt}] (1) -- (2);
\draw (1,0.5) -- (2);
\end{tikzpicture}}\Big)
\,\Lc\tilde g_0^N.
\end{array}\right.
\end{equation}
A direct application of Lemma~\ref{lem:approx-hier} allows to compute the associated remainder terms.

\begin{lem}\label{lem:rem/3}
Given $m_0=2$, for the above choice of the ansatz, the remainder terms in the approximate hierarchy~\eqref{eq:rem} are explicitly given by
\begin{eqnarray*}
\Lc R_0^N&=&-\frac{t_N}{N^2}
\Big(\,
{\tiny\begin{tikzpicture}[baseline={([yshift=-1ex]current bounding box.center)},scale=0.8]
\begin{scope}[every node/.style={circle,draw,fill=white,inner sep=0pt,minimum size=3pt}]
\node (0) at (0,0) {};
\node (1) at (0.5,0.3) {};
\node (2) at (1,0.3) {};
\node (3) at (1.5,0) {};
\end{scope}
\draw (0)--(0,0.3)--(1)--(2)--(1,0.6)--(0.5,0.6)--(1);
\draw (0)--(1,0);
\draw[decorate,decoration={snake,amplitude=1pt,segment length=3pt}] (1,0) -- (3);
\draw[decorate,decoration={snake,amplitude=1pt,segment length=3pt}] (2) -- (1.5,0.3);
\draw (1.5,0.3)--(3);
\end{tikzpicture}}
+
{\tiny\begin{tikzpicture}[baseline={([yshift=-1ex]current bounding box.center)},scale=0.8]
\begin{scope}[every node/.style={circle,draw,fill=white,inner sep=0pt,minimum size=3pt}]
\node (0) at (0,0) {};
\node (1) at (0.5,0) {};
\node (2) at (1,0) {};
\node (3) at (1.5,0) {};
\end{scope}
\draw (0)--(1)--(2)--(1,0.3)--(0.5,0.3)--(1);
\draw (0)--(0,0.6)--(1,0.6);
\draw[decorate,decoration={snake,amplitude=1pt,segment length=3pt}] (1,0.6) -- (1.5,0.6);
\draw[decorate,decoration={snake,amplitude=1pt,segment length=3pt}] (2) -- (3);
\draw (1.5,0.6)--(3);
\end{tikzpicture}}
\,\Big)\Lc\tilde g_0^N,
\\
\Lc R_1^N&=&
\frac{it_N}{N^{3/2}}\Big(\,
{\tiny\begin{tikzpicture}[baseline={([yshift=-1ex]current bounding box.center)},scale=0.8]
\begin{scope}[every node/.style={circle,draw,fill=white,inner sep=0pt,minimum size=3pt}]
\node (0) at (0,0.6) {};
\node (1) at (0.5,0) {};
\node (2) at (1,0) {};
\end{scope}
\draw (0)--(0.5,0.6);
\draw (0)--(0,0.3)--(0.5,0.3)--(1);
\draw (0,0)--(1);
\draw[decorate,decoration={snake,amplitude=1pt,segment length=3pt}] (0.5,0.6) -- (1,0.6);
\draw[decorate,decoration={snake,amplitude=1pt,segment length=3pt}] (1) -- (2);
\draw (1,0.6)--(2);
\end{tikzpicture}}
+
{\tiny\begin{tikzpicture}[baseline={([yshift=-1ex]current bounding box.center)},scale=0.8]
\begin{scope}[every node/.style={circle,draw,fill=white,inner sep=0pt,minimum size=3pt}]
\node (0) at (0,0) {};
\node (1) at (0.5,0.6) {};
\node (2) at (1,0) {};
\end{scope}
\draw (0)--(0.5,0);
\draw (0)--(0,0.3)--(0.5,0.3)--(1);
\draw (0,0.6)--(1);
\draw[decorate,decoration={snake,amplitude=1pt,segment length=3pt}] (1) -- (1,0.6);
\draw[decorate,decoration={snake,amplitude=1pt,segment length=3pt}] (0.5,0) -- (2);
\draw (2)--(1,0.6);
\end{tikzpicture}}
+
{\tiny\begin{tikzpicture}[baseline={([yshift=-1ex]current bounding box.center)},scale=0.8]
\begin{scope}[every node/.style={circle,draw,fill=white,inner sep=0pt,minimum size=3pt}]
\node (0) at (0,0) {};
\node (1) at (0.5,0.3) {};
\node (2) at (1,0) {};
\end{scope}
\draw (0)--(0.5,0);
\draw (0)--(0,0.3)--(1);
\draw (0,0.6)--(0.5,0.6)--(1);
\draw[decorate,decoration={snake,amplitude=1pt,segment length=3pt}] (1) -- (1,0.3);
\draw[decorate,decoration={snake,amplitude=1pt,segment length=3pt}] (0.5,0) -- (2);
\draw (1,0.3)--(2);
\end{tikzpicture}}
+
{\tiny\begin{tikzpicture}[baseline={([yshift=-1ex]current bounding box.center)},scale=0.8]
\begin{scope}[every node/.style={circle,draw,fill=white,inner sep=0pt,minimum size=3pt}]
\node (0) at (0,0.3) {};
\node (1) at (0.5,0) {};
\node (2) at (1,0) {};
\end{scope}
\draw (0)--(0,0.6)--(0.5,0.6);
\draw (0)--(0.5,0.3)--(1);
\draw (0,0)--(1);
\draw[decorate,decoration={snake,amplitude=1pt,segment length=3pt}] (0.5,0.6) -- (1,0.6);
\draw[decorate,decoration={snake,amplitude=1pt,segment length=3pt}] (1) -- (2);
\draw (1,0.6)--(2);
\end{tikzpicture}}
+
{\tiny\begin{tikzpicture}[baseline={([yshift=-1ex]current bounding box.center)},scale=0.8]
\begin{scope}[every node/.style={circle,draw,fill=white,inner sep=0pt,minimum size=3pt}]
\node (0) at (0,0) {};
\node (1) at (0.5,0) {};
\node (2) at (1,0) {};
\end{scope}
\draw (0)--(1);
\draw (0)--(0,0.6)--(0.5,0.6);
\draw (-0.1,0.3)--(-0.05,0.3);
\draw (0.05,0.3)--(0.5,0.3)--(1);
\draw[decorate,decoration={snake,amplitude=1pt,segment length=3pt}] (0.5,0.6) -- (1,0.6);
\draw[decorate,decoration={snake,amplitude=1pt,segment length=3pt}] (1) -- (2);
\draw (1,0.6)--(2);
\end{tikzpicture}}
+
{\tiny\begin{tikzpicture}[baseline={([yshift=-1ex]current bounding box.center)},scale=0.8]
\begin{scope}[every node/.style={circle,draw,fill=white,inner sep=0pt,minimum size=3pt}]
\node (0) at (0,0.3) {};
\node (1) at (0.5,0.6) {};
\node (2) at (1,0) {};
\end{scope}
\draw (0)--(0,0.6)--(1);
\draw (0)--(0.5,0.3)--(1);
\draw (0,0)--(0.5,0);
\draw[decorate,decoration={snake,amplitude=1pt,segment length=3pt}] (1) -- (1,0.6);
\draw[decorate,decoration={snake,amplitude=1pt,segment length=3pt}] (0.5,0) -- (2);
\draw (1,0.6)--(2);
\end{tikzpicture}}
\,\Big)\Lc\tilde g_0^N,
\\
\Lc R_2^N&=&0
\end{eqnarray*}
\end{lem}

\begin{proof}
The choice of $\Omega$ yields $\partial\Omega:=\{(1,-1),(1,-1,-1)\}$ and the conclusion follows immediately from Lemma~\ref{lem:approx-hier}.
\end{proof}

We now turn to estimates on the above explicit formulas for remainder terms, using the properties of renormalized semigroups.

\begin{lem}\label{lem:R1N-m02}
Given $m_0=2$, the remainder terms defined in Lemma~\ref{lem:rem/3} are bounded as follows, provided $d\ge8$,
\begin{eqnarray*}
\3\!\Lc R_0^N\!\3&\lesssim&\frac{t_N}{N^2}\3\langle\nabla_{v_0}\rangle^4\Lc\tilde g_0^N\!\3,\\
\3\!\Lc R_1^N\!\3&\lesssim&\frac{t_N}{N^{3/2}}\3\langle\nabla_{v_0}\rangle^3\Lc\tilde g_0^N\!\3.
\end{eqnarray*}
\end{lem}

\begin{proof}
For shortness, we focus on the proof of the estimate on $R_1^N$, while the corresponding estimate on $R_0^N$ follows similarly and is actually simpler.
Starting point is the diagrammatic expression for $R_1^N$ in Lemma~\ref{lem:rem/3}, which we further decompose into three parts,
\[R_1^N\,=\,\frac{it_N}{N^{3/2}}\Big(R_{1,1}^N+R_{1,2}^N+R_{1,3}^N\Big),\]
in terms of
\begin{eqnarray*}
\Lc R_{1,1}^N&=&
{\tiny\begin{tikzpicture}[baseline={([yshift=-1ex]current bounding box.center)},scale=0.8]
\begin{scope}[every node/.style={circle,draw,fill=white,inner sep=0pt,minimum size=3pt}]
\node (0) at (0,0.3) {};
\node (1) at (0.5,0) {};
\node (2) at (1,0) {};
\end{scope}
\draw (0)--(0,0.6)--(0.5,0.6);
\draw (0)--(0.5,0.3)--(1);
\draw (0,0)--(1);
\draw[decorate,decoration={snake,amplitude=1pt,segment length=3pt}] (0.5,0.6) -- (1,0.6);
\draw[decorate,decoration={snake,amplitude=1pt,segment length=3pt}] (1) -- (2);
\draw (1,0.6)--(2);
\end{tikzpicture}}\,\Lc\tilde g_0^N
+
{\tiny\begin{tikzpicture}[baseline={([yshift=-1ex]current bounding box.center)},scale=0.8]
\begin{scope}[every node/.style={circle,draw,fill=white,inner sep=0pt,minimum size=3pt}]
\node (0) at (0,0) {};
\node (1) at (0.5,0) {};
\node (2) at (1,0) {};
\end{scope}
\draw (0)--(1);
\draw (0)--(0,0.6)--(0.5,0.6);
\draw (-0.1,0.3)--(-0.05,0.3);
\draw (0.05,0.3)--(0.5,0.3)--(1);
\draw[decorate,decoration={snake,amplitude=1pt,segment length=3pt}] (0.5,0.6) -- (1,0.6);
\draw[decorate,decoration={snake,amplitude=1pt,segment length=3pt}] (1) -- (2);
\draw (1,0.6)--(2);
\end{tikzpicture}}\,\Lc\tilde g_0^N,
\\
\Lc R_{1,2}^N&=&
{\tiny\begin{tikzpicture}[baseline={([yshift=-1ex]current bounding box.center)},scale=0.8]
\begin{scope}[every node/.style={circle,draw,fill=white,inner sep=0pt,minimum size=3pt}]
\node (0) at (0,0.3) {};
\node (1) at (0.5,0.6) {};
\node (2) at (1,0) {};
\end{scope}
\draw (0)--(0,0.6)--(1);
\draw (0)--(0.5,0.3)--(1);
\draw (0,0)--(0.5,0);
\draw[decorate,decoration={snake,amplitude=1pt,segment length=3pt}] (1) -- (1,0.6);
\draw[decorate,decoration={snake,amplitude=1pt,segment length=3pt}] (0.5,0) -- (2);
\draw (1,0.6)--(2);
\end{tikzpicture}}\,\Lc\tilde g_0^N
+
{\tiny\begin{tikzpicture}[baseline={([yshift=-1ex]current bounding box.center)},scale=0.8]
\begin{scope}[every node/.style={circle,draw,fill=white,inner sep=0pt,minimum size=3pt}]
\node (0) at (0,0) {};
\node (1) at (0.5,0.3) {};
\node (2) at (1,0) {};
\end{scope}
\draw (0)--(0.5,0);
\draw (0)--(0,0.3)--(1);
\draw (0,0.6)--(0.5,0.6)--(1);
\draw[decorate,decoration={snake,amplitude=1pt,segment length=3pt}] (1) -- (1,0.3);
\draw[decorate,decoration={snake,amplitude=1pt,segment length=3pt}] (0.5,0) -- (2);
\draw (1,0.3)--(2);
\end{tikzpicture}}\,\Lc\tilde g_0^N,
\\
\Lc R_{1,3}^N&=&
{\tiny\begin{tikzpicture}[baseline={([yshift=-1ex]current bounding box.center)},scale=0.8]
\begin{scope}[every node/.style={circle,draw,fill=white,inner sep=0pt,minimum size=3pt}]
\node (0) at (0,0.6) {};
\node (1) at (0.5,0) {};
\node (2) at (1,0) {};
\end{scope}
\draw (0)--(0.5,0.6);
\draw (0)--(0,0.3)--(0.5,0.3)--(1);
\draw (0,0)--(1);
\draw[decorate,decoration={snake,amplitude=1pt,segment length=3pt}] (0.5,0.6) -- (1,0.6);
\draw[decorate,decoration={snake,amplitude=1pt,segment length=3pt}] (1) -- (2);
\draw (1,0.6)--(2);
\end{tikzpicture}}\,\Lc\tilde g_0^N
+
{\tiny\begin{tikzpicture}[baseline={([yshift=-1ex]current bounding box.center)},scale=0.8]
\begin{scope}[every node/.style={circle,draw,fill=white,inner sep=0pt,minimum size=3pt}]
\node (0) at (0,0) {};
\node (1) at (0.5,0.6) {};
\node (2) at (1,0) {};
\end{scope}
\draw (0)--(0.5,0);
\draw (0)--(0,0.3)--(0.5,0.3)--(1);
\draw (0,0.6)--(1);
\draw[decorate,decoration={snake,amplitude=1pt,segment length=3pt}] (1) -- (1,0.6);
\draw[decorate,decoration={snake,amplitude=1pt,segment length=3pt}] (0.5,0) -- (2);
\draw (2)--(1,0.6);
\end{tikzpicture}}\,\Lc\tilde g_0^N.
\end{eqnarray*}
We split the proof into three steps. Recall the notation $\hat k=\frac{k}{|k|}$ for $k\in\R^d$.

\medskip\noindent
{\bf Step~1:} Proof that for $d\ge4$,
\begin{equation}\label{eq:est-R11N}
\3\Lc R_{1,1}^N\3\,\lesssim\,\3\langle\nabla_{v_0}\rangle^3\Lc\tilde g_0^N\!\3.
\end{equation}
The two terms in the definition of $\Lc R_{1,1}^N$ are similar, and we start with the first one. We use the coordinates $(k,v_0,v_1)\equiv((-k,v_0),(k,v_1))$ on $\hat\Dd^2$ and $(k,k',v_0,v_1,v_2)\equiv((-k-k',v_0),(k,v_1),(k',v_2))$ on $\hat\Dd^3$. By definition, we have
\begin{multline*}
{\tiny\begin{tikzpicture}[baseline={([yshift=-1ex]current bounding box.center)},scale=0.8]
\begin{scope}[every node/.style={circle,draw,fill=white,inner sep=0pt,minimum size=3pt}]
\node (0) at (0,0.3) {};
\node (1) at (0.5,0) {};
\node (2) at (1,0) {};
\end{scope}
\draw (0)--(0,0.6)--(0.5,0.6);
\draw (0)--(0.5,0.3)--(1);
\draw (0,0)--(1);
\draw[decorate,decoration={snake,amplitude=1pt,segment length=3pt}] (0.5,0.6) -- (1,0.6);
\draw[decorate,decoration={snake,amplitude=1pt,segment length=3pt}] (1) -- (2);
\draw (1,0.6)--(2);
\end{tikzpicture}}
\,\Lc\tilde g_0^N(\alpha,k,v_0,v_1)
\,=\,
-\int_{(\R^d)^2}\sqrt M(v_2)\,k'\V(k')\cdot\nabla_{v_1}\\
\times\Big(\tfrac{1+i\alpha}{t_N}+ik\cdot(v_1-v_0)+ik'\cdot(v_2-v_1)+\tfrac\kappa N\hat D_2\Big)^{-1}\sqrt M(v_1)\,(k-k')\V(k-k')\cdot\nabla_{v_0}\\
\times\Big(\WIG\hat S^{1,-}_{0,1}\Lc\tilde g_0^N\Big)(k',v_0,v_2)\,d^*k'dv_2,
\end{multline*}
where in the last factor the renormalized semigroup is applied to
\begin{equation}\label{eq:ref-def}
\hat S^{1,-}_{0,1}\Lc\tilde g_0^N(k',v_0,v_2)=-\sqrt M(v_2)\,k'\V(k')\cdot\nabla_{v_0}\Lc\tilde g_0^N(v_0).
\end{equation}
Since the latter is holomorphic in $v_2$, Proposition~\ref{prop:box-anal} ensures that for fixed $k'$ the map
\[v_2\,\mapsto\,\Big(\WIG\hat S^{1,-}_{0,1}\Lc\tilde g_0^N\Big)(\alpha,k',v_0,v_2)\,=\,\Big(\tau_{(0,v_2)}\WIG \hat S^{1,-}_{0,1}\Lc\tilde g_0^N\Big)(\alpha,k',v_0,0)\]
admits a holomorphic extension to $S_\nu^1$ for any $\nu\in\Sp^{d-1}$ with $k'\cdot\nu>0$. Choosing $\nu=\hat k'$, we may thus perform a contour deformation $v_2\mapsto v_2-i\hat k'$ in the above integral, to the effect of
\begin{multline*}
{\tiny\begin{tikzpicture}[baseline={([yshift=-1ex]current bounding box.center)},scale=0.8]
\begin{scope}[every node/.style={circle,draw,fill=white,inner sep=0pt,minimum size=3pt}]
\node (0) at (0,0.3) {};
\node (1) at (0.5,0) {};
\node (2) at (1,0) {};
\end{scope}
\draw (0)--(0,0.6)--(0.5,0.6);
\draw (0)--(0.5,0.3)--(1);
\draw (0,0)--(1);
\draw[decorate,decoration={snake,amplitude=1pt,segment length=3pt}] (0.5,0.6) -- (1,0.6);
\draw[decorate,decoration={snake,amplitude=1pt,segment length=3pt}] (1) -- (2);
\draw (1,0.6)--(2);
\end{tikzpicture}}\,\Lc\tilde g_0^N
\,(\alpha,k,v_0,v_1)
\,=\,
-\int_{(\R^d)^2}\sqrt M(v_2-i\hat k')\,k'\V(k')\cdot\nabla_{v_1}\\
\times\Big(\tfrac{1+i\alpha}{t_N}+|k'|+ik\cdot(v_1-v_0)+ik'\cdot(v_2-v_1)+\tfrac\kappa N\hat D_2\Big)^{-1}\sqrt M(v_1)\,(k-k')\V(k-k')\cdot\nabla_{v_0}\\
\times\Big(\tau_{(0,-i\hat k')}\WIG\hat S^{1,-}_{0,1}\Lc\tilde g_0^N\Big)(\alpha,k',v_0,v_2)\,d^*k'dv_2.
\end{multline*}
Diagrammatically, we use the following short-hand notation, where we add labels with the corresponding Fourier momentum variable above each edge and where we add in bracket the deformation of a velocity variable on the right of the corresponding line,
\begin{equation*}
{\tiny\begin{tikzpicture}[baseline={([yshift=-1ex]current bounding box.center)},scale=0.8]
\begin{scope}[every node/.style={circle,draw,fill=white,inner sep=0pt,minimum size=3pt}]
\node (0) at (0,0.3) {};
\node (1) at (0.5,0) {};
\node (2) at (1,0) {};
\end{scope}
\draw (0)--(0,0.6)--(0.5,0.6);
\draw (0)--(0.5,0.3)--(1);
\draw (0,0)--(1);
\draw[decorate,decoration={snake,amplitude=1pt,segment length=3pt}] (0.5,0.6) -- (1,0.6);
\draw[decorate,decoration={snake,amplitude=1pt,segment length=3pt}] (1) -- (2);
\draw (1,0.6)--(2);
\end{tikzpicture}}\,\Lc\tilde g_0^N(\alpha,k,v_0,v_1)
\,=\,
\int_{\R^d}
{\tiny\begin{tikzpicture}[baseline={([yshift=-2ex]current bounding box.center)},scale=0.8]
\begin{scope}[every node/.style={circle,draw,fill=white,inner sep=0pt,minimum size=3pt}]
\node (0) at (0,0.4) {};
\node (1) at (1.2,0) {};
\node (2) at (2.4,0) {};
\end{scope}
\draw (0)--(0,0.8)--(1.2,0.8);
\draw (0)--(1.2,0.4)--(1);
\draw (0,0)--(1);
\draw[decorate,decoration={snake,amplitude=1pt,segment length=3pt}] (1.2,0.8) -- (2.4,0.8);
\draw[decorate,decoration={snake,amplitude=1pt,segment length=3pt}] (1) -- (2);
\draw (2.4,0.8)--(2);
\node at (0.6,0.17) {$-k$};
\node at (1.8,0.97) {$k'$};
\node at (1.8,0.17) {$-k'$};
\node at (0.6,0.97) {$k'$};
\node at (0.6,0.57) {$k-k'$};
\node at (2.9,0.8) {$[-i\hat k']$};
\end{tikzpicture}}\!
\Lc\tilde g_0^N(\alpha,k,v_0,v_1)\,d^*k'.
\end{equation*}
A direct estimate yields
\begin{equation*}
\big\|{\tiny\begin{tikzpicture}[baseline={([yshift=-1ex]current bounding box.center)},scale=0.8]
\begin{scope}[every node/.style={circle,draw,fill=white,inner sep=0pt,minimum size=3pt}]
\node (0) at (0,0.3) {};
\node (1) at (0.5,0) {};
\node (2) at (1,0) {};
\end{scope}
\draw (0)--(0,0.6)--(0.5,0.6);
\draw (0)--(0.5,0.3)--(1);
\draw (0,0)--(1);
\draw[decorate,decoration={snake,amplitude=1pt,segment length=3pt}] (0.5,0.6) -- (1,0.6);
\draw[decorate,decoration={snake,amplitude=1pt,segment length=3pt}] (1) -- (2);
\draw (1,0.6)--(2);
\end{tikzpicture}}\,\Lc\tilde g_0^N\big\|_k
\,\lesssim\,
\int_{\R^d}
|k-k'|(1+\tfrac{|k-k'|}{|k'|})  \V(k') \V(k-k')
\big\|\nabla_{v_0}\tau_{(0,-i\hat k')}\WIG\hat S^{1,-}_{0,1}\Lc\tilde g_0^N(\alpha,k',\cdot)\big\|_{}\,dk',
\end{equation*}
and thus, by Proposition~\ref{prop:box-anal} and the definition of $\hat S^{1,-}_{0,1}$, cf.~\eqref{eq:ref-def},
\begin{eqnarray*}
\lefteqn{\big\|{\tiny\begin{tikzpicture}[baseline={([yshift=-1ex]current bounding box.center)},scale=0.8]
\begin{scope}[every node/.style={circle,draw,fill=white,inner sep=0pt,minimum size=3pt}]
\node (0) at (0,0.3) {};
\node (1) at (0.5,0) {};
\node (2) at (1,0) {};
\end{scope}
\draw (0)--(0,0.6)--(0.5,0.6);
\draw (0)--(0.5,0.3)--(1);
\draw (0,0)--(1);
\draw[decorate,decoration={snake,amplitude=1pt,segment length=3pt}] (0.5,0.6) -- (1,0.6);
\draw[decorate,decoration={snake,amplitude=1pt,segment length=3pt}] (1) -- (2);
\draw (1,0.6)--(2);
\end{tikzpicture}}\,\Lc\tilde g_0^N\big\|_k}\\
&\lesssim&
\int_{\R^d}
\langle k'\rangle|k-k'|(1+\tfrac{|k-k'|}{|k'|})(1+\tfrac1{|k'|})^2  \V(k') \V(k-k')
\|\langle\nabla_{v_{[1]}}\rangle\tau_{(0,-i\hat k')}\hat S^{1,-}_{0,1}\Lc\tilde g_0^N(\alpha,k',\cdot)\|\,dk'\\
&\lesssim&
\|\langle\nabla_{v_0}\rangle^2\Lc\tilde g_0^N\|
\int_{\R^d}
\langle k'\rangle|k-k'|(|k'|+|k-k'|)(1+\tfrac1{|k'|})^2\V(k')^2 \V(k-k')\,dk',
\end{eqnarray*}
which yields the desired estimate for this term for $d\ge3$.
We turn to the second term in the definition of $\Lc R_{1,1}^N$. Arguing similarly as above, we use complex deformation,
\[{\tiny\begin{tikzpicture}[baseline={([yshift=-1ex]current bounding box.center)},scale=0.8]
\begin{scope}[every node/.style={circle,draw,fill=white,inner sep=0pt,minimum size=3pt}]
\node (0) at (0,0) {};
\node (1) at (0.5,0) {};
\node (2) at (1,0) {};
\end{scope}
\draw (0)--(1);
\draw (0)--(0,0.6)--(0.5,0.6);
\draw (-0.1,0.3)--(-0.05,0.3);
\draw (0.05,0.3)--(0.5,0.3)--(1);
\draw[decorate,decoration={snake,amplitude=1pt,segment length=3pt}] (0.5,0.6) -- (1,0.6);
\draw[decorate,decoration={snake,amplitude=1pt,segment length=3pt}] (1) -- (2);
\draw (1,0.6)--(2);
\end{tikzpicture}}\,\Lc\tilde g_0^N(\alpha,k,v_0,v_1)
\,=\,\int_{\R^d}
{\tiny\begin{tikzpicture}[baseline={([yshift=-2ex]current bounding box.center)},scale=0.8]
\begin{scope}[every node/.style={circle,draw,fill=white,inner sep=0pt,minimum size=3pt}]
\node (0) at (0,0) {};
\node (1) at (1.2,0) {};
\node (2) at (2.4,0) {};
\end{scope}
\draw (0)--(1);
\draw (0)--(0,0.8)--(1.2,0.8);
\draw (-0.1,0.4)--(1.2,0.4)--(1);
\draw[decorate,decoration={snake,amplitude=1pt,segment length=3pt}] (1.2,0.8) -- (2.4,0.8);
\draw[decorate,decoration={snake,amplitude=1pt,segment length=3pt}] (1) -- (2);
\draw (2.4,0.8)--(2);
\node at (0.6,0.57) {$k$};
\node at (0.6,0.97) {$k'$};
\node at (0.6,0.17) {$-k-k'$};
\node at (1.8,0.97) {$k'$};
\node at (1.8,0.17) {$-k'$};
\node at (2.9,0.8) {$[-i\hat k']$};
\end{tikzpicture}}\!\Lc\tilde g_0^N(\alpha,k,v_0,v_1)\,d^*k',\]
and we find
\[\big\|{\tiny\begin{tikzpicture}[baseline={([yshift=-1ex]current bounding box.center)},scale=0.8]
\begin{scope}[every node/.style={circle,draw,fill=white,inner sep=0pt,minimum size=3pt}]
\node (0) at (0,0) {};
\node (1) at (0.5,0) {};
\node (2) at (1,0) {};
\end{scope}
\draw (0)--(1);
\draw (0)--(0,0.6)--(0.5,0.6);
\draw (-0.1,0.3)--(-0.05,0.3);
\draw (0.05,0.3)--(0.5,0.3)--(1);
\draw[decorate,decoration={snake,amplitude=1pt,segment length=3pt}] (0.5,0.6) -- (1,0.6);
\draw[decorate,decoration={snake,amplitude=1pt,segment length=3pt}] (1) -- (2);
\draw (1,0.6)--(2);
\end{tikzpicture}}\,\Lc\tilde g_0^N\big\|_k
\,\lesssim\,
\|\langle\nabla_{v_0}\rangle^3\Lc\tilde g_0^N\|\int_{\R^d}|k|\langle k+k'\rangle\langle k'\rangle^2(1+\tfrac1{|k'|})^3\V(k)\V(k')^2\,dk'.\]
The claim~\eqref{eq:est-R11N} follows for $d\ge4$.

\medskip\noindent
{\bf Step~2:} Proof that for $d\ge6$,
\begin{equation}\label{eq:est-R12N}
\3\!\Lc R_{1,2}^N\!\3\,\lesssim\,\3\langle\nabla_{v_0}\rangle^3\Lc\tilde g_0^N\!\3.
\end{equation}
The two terms in the definition of $\Lc R_{1,2}^N$ are similar, and we start with the first one.
Provided $\hat k'\ne-\hat k$, we let $\nu_{k,k'}:=(\hat k+\hat k')/|\hat k+\hat k'|$, which ensures both $k\cdot\nu_{k,k'}>0$ and $k'\cdot\nu_{k,k'}>0$. Using Proposition~\ref{prop:box-anal} to perform complex contour deformation, we then find
\[{\tiny\begin{tikzpicture}[baseline={([yshift=-1ex]current bounding box.center)},scale=0.8]
\begin{scope}[every node/.style={circle,draw,fill=white,inner sep=0pt,minimum size=3pt}]
\node (0) at (0,0.3) {};
\node (1) at (0.5,0.6) {};
\node (2) at (1,0) {};
\end{scope}
\draw (0)--(0,0.6)--(1);
\draw (0)--(0.5,0.3)--(1);
\draw (0,0)--(0.5,0);
\draw[decorate,decoration={snake,amplitude=1pt,segment length=3pt}] (1) -- (1,0.6);
\draw[decorate,decoration={snake,amplitude=1pt,segment length=3pt}] (0.5,0) -- (2);
\draw (1,0.6)--(2);
\end{tikzpicture}}\,\Lc\tilde g_0^N(\alpha,k,v_0,v_1)
\,=\,
\int_{\R^d}{\tiny\begin{tikzpicture}[baseline={([yshift=-2ex]current bounding box.center)},scale=0.8]
\begin{scope}[every node/.style={circle,draw,fill=white,inner sep=0pt,minimum size=3pt}]
\node (0) at (0,0.4) {};
\node (1) at (1.2,0.8) {};
\node (2) at (2.4,0) {};
\end{scope}
\draw (0)--(0,0.8)--(1);
\draw (0)--(1.2,0.4)--(1);
\draw (0,0)--(1.2,0);
\draw[decorate,decoration={snake,amplitude=1pt,segment length=3pt}] (1) -- (2.4,0.8);
\draw[decorate,decoration={snake,amplitude=1pt,segment length=3pt}] (1.2,0) -- (2);
\draw (2.4,0.8)--(2);
\node at (0.6,0.17) {$-k$};
\node at (0.6,0.57) {$k-k'$};
\node at (0.6,0.97) {$k'$};
\node at (1.8,0.97) {$k$};
\node at (1.8,0.17) {$-k$};
\node at (3.2,0.8) {$[-i\nu_{k,k'}]$};
\end{tikzpicture}}\,\Lc\tilde g_0^N(\alpha,k,v_0,v_1)\,d^*k',\]
which is bounded as follows,
\begin{multline*}
\big\|{\tiny\begin{tikzpicture}[baseline={([yshift=-1ex]current bounding box.center)},scale=0.8]
\begin{scope}[every node/.style={circle,draw,fill=white,inner sep=0pt,minimum size=3pt}]
\node (0) at (0,0.3) {};
\node (1) at (0.5,0.6) {};
\node (2) at (1,0) {};
\end{scope}
\draw (0)--(0,0.6)--(1);
\draw (0)--(0.5,0.3)--(1);
\draw (0,0)--(0.5,0);
\draw[decorate,decoration={snake,amplitude=1pt,segment length=3pt}] (1) -- (1,0.6);
\draw[decorate,decoration={snake,amplitude=1pt,segment length=3pt}] (0.5,0) -- (2);
\draw (1,0.6)--(2);
\end{tikzpicture}}\,\Lc\tilde g_0^N\big\|_k
\,\lesssim\, \|\langle\nabla_{v_0}\rangle^2 \Lc\tilde g_0^N\|
\int_{\R^d}\langle k\rangle\langle k-k'\rangle|k||k'||k-k'|\V(k)\V(k')\V(k-k')\\[-3mm]
\times(1+\tfrac1{k'\cdot\nu_{k,k'} })^{2}(1+\tfrac1{k\cdot\nu_{k,k'}})^{2}\,dk'.
\end{multline*}
Noting that $\hat k'\cdot\nu_{k,k'}=\hat k\cdot\nu_{k,k'}=\cos(\alpha/2)$ if $\hat k\cdot\hat k'=\cos\alpha$,
we can estimate the above $k'$-integral,
\begin{eqnarray*}
\lefteqn{\int_{\R^d}\langle k\rangle\langle k-k'\rangle|k||k'||k-k'|\V(k)\V(k')\V(k-k')(1+\tfrac1{k'\cdot\nu_{k,k'} })^{2}(1+\tfrac1{k\cdot\nu_{k,k'}})^{2}\,dk'}\\
&\lesssim&\langle k\rangle^4(1+\tfrac1{|k|})\V(k)\int_{\R^d}\langle k'\rangle^3\V(k')(1+\tfrac1{|k'|})(1+\tfrac1{\hat k\cdot\nu_{k,k'} })^4\,dk'\\
&\lesssim&\langle k\rangle^4(1+\tfrac1{|k|})\V(k)\int_0^\pi\cos(\alpha/2)^{-4}\sin(\alpha)^{d-2}\,d\alpha,
\end{eqnarray*}
where the remaining integral is bounded provided $d\ge6$.
We are thus led to
\[\big\|{\tiny\begin{tikzpicture}[baseline={([yshift=-1ex]current bounding box.center)},scale=0.8]
\begin{scope}[every node/.style={circle,draw,fill=white,inner sep=0pt,minimum size=3pt}]
\node (0) at (0,0.3) {};
\node (1) at (0.5,0.6) {};
\node (2) at (1,0) {};
\end{scope}
\draw (0)--(0,0.6)--(1);
\draw (0)--(0.5,0.3)--(1);
\draw (0,0)--(0.5,0);
\draw[decorate,decoration={snake,amplitude=1pt,segment length=3pt}] (1) -- (1,0.6);
\draw[decorate,decoration={snake,amplitude=1pt,segment length=3pt}] (0.5,0) -- (2);
\draw (1,0.6)--(2);
\end{tikzpicture}}\,\Lc\tilde g_0^N\big\|_k
\,\lesssim\,\langle k\rangle^4(1+\tfrac1{|k|})\V(k)\|\langle\nabla_{v_0}\rangle^2 \Lc\tilde g_0^N\|,\]
which yields the desired estimate for this term.
We turn to the second term in the definition of $\Lc R_{1,2}^N$.
Complex deformation now yields
\[{\tiny\begin{tikzpicture}[baseline={([yshift=-1ex]current bounding box.center)},scale=0.8]
\begin{scope}[every node/.style={circle,draw,fill=white,inner sep=0pt,minimum size=3pt}]
\node (0) at (0,0) {};
\node (1) at (0.5,0.3) {};
\node (2) at (1,0) {};
\end{scope}
\draw (0)--(0.5,0);
\draw (0)--(0,0.3)--(1);
\draw (0,0.6)--(0.5,0.6)--(1);
\draw[decorate,decoration={snake,amplitude=1pt,segment length=3pt}] (1) -- (1,0.3);
\draw[decorate,decoration={snake,amplitude=1pt,segment length=3pt}] (0.5,0) -- (2);
\draw (1,0.3)--(2);
\end{tikzpicture}}\,\Lc\tilde g_0^N(\alpha,k,v_0,v_1)
\,=\,
\int_{\R^d}
{\tiny\begin{tikzpicture}[baseline={([yshift=-2ex]current bounding box.center)},scale=0.8]
\begin{scope}[every node/.style={circle,draw,fill=white,inner sep=0pt,minimum size=3pt}]
\node (0) at (0,0) {};
\node (1) at (1.2,0.4) {};
\node (2) at (2.4,0) {};
\end{scope}
\draw (0)--(1.2,0);
\draw (0)--(0,0.4)--(1);
\draw (0,0.8)--(1.2,0.8)--(1);
\draw[decorate,decoration={snake,amplitude=1pt,segment length=3pt}] (1) -- (2.4,0.4);
\draw[decorate,decoration={snake,amplitude=1pt,segment length=3pt}] (1.2,0) -- (2);
\draw (2.4,0.4)--(2);
\node at (0.6,0.97) {$k$};
\node at (0.6,0.57) {$k'$};
\node at (1.2,0.17) {$-k-k'$};
\node at (1.8,0.57) {$k+k'$};
\node at (3.4,0.4) {$[-i\nu_{k',k+k'}]$};
\end{tikzpicture}}\Lc\tilde g_0^N(\alpha,k,v_0,v_1)\,d^*k',\]
which we can bound as follows,
\begin{multline*}
\big\|{\tiny\begin{tikzpicture}[baseline={([yshift=-1ex]current bounding box.center)},scale=0.8]
\begin{scope}[every node/.style={circle,draw,fill=white,inner sep=0pt,minimum size=3pt}]
\node (0) at (0,0) {};
\node (1) at (0.5,0.3) {};
\node (2) at (1,0) {};
\end{scope}
\draw (0)--(0.5,0);
\draw (0)--(0,0.3)--(1);
\draw (0,0.6)--(0.5,0.6)--(1);
\draw[decorate,decoration={snake,amplitude=1pt,segment length=3pt}] (1) -- (1,0.3);
\draw[decorate,decoration={snake,amplitude=1pt,segment length=3pt}] (0.5,0) -- (2);
\draw (1,0.3)--(2);
\end{tikzpicture}}\,\Lc\tilde g_0^N\big\|_k
\,\lesssim\,\|\langle\nabla_{v_0}\rangle^3 \Lc\tilde g_0^N\|\int_{\R^d}\langle k+k'\rangle^2|k||k'||k+k'|\V(k)\V(k')\V(k+k')\\
\times\Big(\tfrac1{k'\cdot\nu_{k',k+k'}}\big(1+\tfrac1{(k+k')\cdot\nu_{k',k+k'}}\big)^3
+\big(\tfrac1{k'\cdot\nu_{k',k+k'}}\big)^2\big(1+\tfrac1{(k+k')\cdot\nu_{k',k+k'}}\big)^2\Big)\,dk'.
\end{multline*}
Integrability requires again $d\ge6$ and the claimed estimate~\eqref{eq:est-R12N} follows.

\medskip\noindent
{\bf Step~3:} Proof that for $d\ge8$,
\begin{equation}\label{eq:est-R13N}
\3\!\Lc R_{1,3}^N\!\3\,\lesssim\,\3\langle\nabla_{v_0}\rangle^3\Lc\tilde g_0^N\!\3.
\end{equation}
The two terms in the definition of $\Lc R_{1,3}^N$ are similar, and we start with the first one. In order to be able to exploit contour deformations for both propagators in this term, let us write its norm as
\[\big\|
{\tiny\begin{tikzpicture}[baseline={([yshift=-1ex]current bounding box.center)},scale=0.8]
\begin{scope}[every node/.style={circle,draw,fill=white,inner sep=0pt,minimum size=3pt}]
\node (0) at (0,0.6) {};
\node (1) at (0.5,0) {};
\node (2) at (1,0) {};
\end{scope}
\draw (0)--(0.5,0.6);
\draw (0)--(0,0.3)--(0.5,0.3)--(1);
\draw (0,0)--(1);
\draw[decorate,decoration={snake,amplitude=1pt,segment length=3pt}] (0.5,0.6) -- (1,0.6);
\draw[decorate,decoration={snake,amplitude=1pt,segment length=3pt}] (1) -- (2);
\draw (1,0.6)--(2);
\end{tikzpicture}}\,\Lc\tilde g_0^N
\big\|_k^2
\,=\,
\int_{(\R^d)^2}
\Big\<\,
{\tiny\begin{tikzpicture}[baseline={([yshift=-2ex]current bounding box.center)},scale=0.8]
\begin{scope}[every node/.style={circle,draw,fill=white,inner sep=0pt,minimum size=3pt}]
\node (0) at (0,0.8) {};
\node (1) at (1.2,0) {};
\node (2) at (2.4,0) {};
\end{scope}
\draw (0)--(1.2,0.8);
\draw (0)--(0,0.4)--(1.2,0.4)--(1);
\draw (0,0)--(1);
\draw[decorate,decoration={snake,amplitude=1pt,segment length=3pt}] (1.2,0.8) -- (2.4,0.8);
\draw[decorate,decoration={snake,amplitude=1pt,segment length=3pt}] (1) -- (2);
\draw (2.4,0.8)--(2);
\node at (0.6,0.17) {$-k$};
\node at (0.6,0.57) {$k-k'$};
\node at (0.6,0.97) {$k'$};
\node at (1.8,0.17) {$-k'$};
\node at (1.8,0.97) {$k'$};
\end{tikzpicture}}\,\Lc\tilde g_0^N\,,\,
{\tiny\begin{tikzpicture}[baseline={([yshift=-2ex]current bounding box.center)},scale=0.8]
\begin{scope}[every node/.style={circle,draw,fill=white,inner sep=0pt,minimum size=3pt}]
\node (0) at (0,0.8) {};
\node (1) at (1.2,0) {};
\node (2) at (2.4,0) {};
\end{scope}
\draw (0)--(1.2,0.8);
\draw (0)--(0,0.4)--(1.2,0.4)--(1);
\draw (0,0)--(1);
\draw[decorate,decoration={snake,amplitude=1pt,segment length=3pt}] (1.2,0.8) -- (2.4,0.8);
\draw[decorate,decoration={snake,amplitude=1pt,segment length=3pt}] (1) -- (2);
\draw (2.4,0.8)--(2);
\node at (0.6,0.17) {$-k$};
\node at (0.6,0.57) {$k-k''$};
\node at (0.6,0.97) {$k''$};
\node at (1.8,0.17) {$-k''$};
\node at (1.8,0.97) {$k''$};
\end{tikzpicture}}\,\Lc\tilde g_0^N\Big\>_k
\,d^*k'd^*k''.
\]
Using Proposition~\ref{prop:box-anal}, we may now perform contour deformations: taking note of the complex conjugate in the scalar product, and setting $\nu=\nu_{-k',k''}$, $\sigma'=(k-k')/|k-k'|$, and $\sigma''=(k-k'')/|k-k''|$, we find
\begin{equation*}
\big\|
{\tiny\begin{tikzpicture}[baseline={([yshift=-1ex]current bounding box.center)},scale=0.8]
\begin{scope}[every node/.style={circle,draw,fill=white,inner sep=0pt,minimum size=3pt}]
\node (0) at (0,0.6) {};
\node (1) at (0.5,0) {};
\node (2) at (1,0) {};
\end{scope}
\draw (0)--(0.5,0.6);
\draw (0)--(0,0.3)--(0.5,0.3)--(1);
\draw (0,0)--(1);
\draw[decorate,decoration={snake,amplitude=1pt,segment length=3pt}] (0.5,0.6) -- (1,0.6);
\draw[decorate,decoration={snake,amplitude=1pt,segment length=3pt}] (1) -- (2);
\draw (1,0.6)--(2);
\end{tikzpicture}}\,\Lc\tilde g_0^N
\big\|^2_k
\,=\,
\int_{(\R^d)^2}
\Big\<\,
{\tiny\begin{tikzpicture}[baseline={([yshift=-2ex]current bounding box.center)},scale=0.8]
\begin{scope}[every node/.style={circle,draw,fill=white,inner sep=0pt,minimum size=3pt}]
\node (0) at (0,0.8) {};
\node (1) at (1.2,0) {};
\node (2) at (2.4,0) {};
\end{scope}
\draw (0)--(1.2,0.8);
\draw (0)--(0,0.4)--(1.2,0.4)--(1);
\draw (0,0)--(1);
\draw[decorate,decoration={snake,amplitude=1pt,segment length=3pt}] (1.2,0.8) -- (2.4,0.8);
\draw[decorate,decoration={snake,amplitude=1pt,segment length=3pt}] (1) -- (2);
\draw (2.4,0.8)--(2);
\node at (0.6,0.17) {$-k$};
\node at (0.6,0.57) {$k-k'$};
\node at (0.6,0.97) {$k'$};
\node at (1.8,0.17) {$-k'$};
\node at (1.8,0.97) {$k'$};
\node at (2.82,0.8) {$[i\nu]$};
\node at (2.97,0.4) {$[-i\sigma']$};
\end{tikzpicture}}\,\Lc\tilde g_0^N\,,\,
{\tiny\begin{tikzpicture}[baseline={([yshift=-2ex]current bounding box.center)},scale=0.8]
\begin{scope}[every node/.style={circle,draw,fill=white,inner sep=0pt,minimum size=3pt}]
\node (0) at (0,0.8) {};
\node (1) at (1.2,0) {};
\node (2) at (2.4,0) {};
\end{scope}
\draw (0)--(1.2,0.8);
\draw (0)--(0,0.4)--(1.2,0.4)--(1);
\draw (0,0)--(1);
\draw[decorate,decoration={snake,amplitude=1pt,segment length=3pt}] (1.2,0.8) -- (2.4,0.8);
\draw[decorate,decoration={snake,amplitude=1pt,segment length=3pt}] (1) -- (2);
\draw (2.4,0.8)--(2);
\node at (0.6,0.17) {$-k$};
\node at (0.6,0.57) {$k-k''$};
\node at (0.6,0.97) {$k''$};
\node at (1.8,0.17) {$-k''$};
\node at (1.8,0.97) {$k''$};
\node at (2.9,0.8) {$[-i\nu]$};
\node at (3,0.4) {$[-i\sigma'']$};
\end{tikzpicture}}\,\Lc\tilde g_0^N\Big\>_k
\,d^*k'd^*k'',
\end{equation*}
and a direct estimate then yields
\begin{multline*}
\big\|
{\tiny\begin{tikzpicture}[baseline={([yshift=-1ex]current bounding box.center)},scale=0.8]
\begin{scope}[every node/.style={circle,draw,fill=white,inner sep=0pt,minimum size=3pt}]
\node (0) at (0,0.6) {};
\node (1) at (0.5,0) {};
\node (2) at (1,0) {};
\end{scope}
\draw (0)--(0.5,0.6);
\draw (0)--(0,0.3)--(0.5,0.3)--(1);
\draw (0,0)--(1);
\draw[decorate,decoration={snake,amplitude=1pt,segment length=3pt}] (0.5,0.6) -- (1,0.6);
\draw[decorate,decoration={snake,amplitude=1pt,segment length=3pt}] (1) -- (2);
\draw (1,0.6)--(2);
\end{tikzpicture}}\,\Lc\tilde g_0^N
\big\|^2_k
\,\lesssim\,
\|\<\nabla_{v_0}\>^3\Lc\tilde g_0^N\|^2
\int_{(\R^d)^2}\< k-k'\>\< k-k''\>\<k'\>^2\<k''\>^2|k'||k''|\\
\times\V(k-k')^2\V(k-k'')^2\V(k')\V(k'')\big(1+\tfrac1{k'\cdot\nu_{-k',k''}}\big)^3\big(1+\tfrac1{k''\cdot\nu_{-k',k''}}\big)^3\,d^*k'd^*k'',
\end{multline*}
where integrability requires $d\ge 8$.
Arguing similarly for the other term in $\Lc R_{1,3}^N$, the claim~\eqref{eq:est-R13N} follows.
\end{proof}

\subsection{General case: term-by-term estimate}
Let the truncation parameter $m_0\ge3$ be fixed. Our goal is to generalize the strategy of the previous section to bound the remainder terms. For that purpose, we shall first prove the following general term-by-term estimate for the contribution of an arbitrary abstract in the Dyson expansion (recall Definition~\ref{def:hist} for abstracts, histories, and their contributions).

\begin{prop}\label{prop:estimation erreur n>1}
Let $(s_1,\ldots,s_n)$ be an abstract with $s_1=1$ and $m+\sum_{i=1}^n s_i=0$ for some $m\ge0$, and assume the space dimension satisfies $d\ge6n+m\vee n$. Then we have for all $g\in C^\infty_c(\R^d)$,
\begin{equation*}
\left\|\mathcal{I}_{(s_1,\ldots,s_n)} g\right\|
\,\lesssim\,N^{\frac5{12}n+\frac14m}\|\<\nabla_{v_0}\>^{n+1} {g}\|.
\end{equation*}
\end{prop}

Before turning to the proof, we introduce notation and several auxiliary results that streamline the argument. Fix an abstract $(s_1,\ldots,s_n)$ as in the above statement, with
\begin{equation}\label{eq:prop-abstr}
s_1=1,\qquad m+\sum_{i=1}^n s_i=0\quad\text{for some $m\ge0$.}
\end{equation}
Given an associated history $(s_i,a_i,b_i)_{1\le i\le n}$, we look for a bound on
\begin{equation}\label{eq:contr-history}
\hat{S}^{s_1}_{a_1,b_1}\WIGdot\ldots\WIGdot\hat{S}^{s_n}_{a_n,b_n}g.
\end{equation}
Note that the number of velocity/momentum variables involved in this contribution is
\[S \,:=\,\sharp\{i:s_i=-1\}\,=\,\tfrac{n+m}{2}.\]
For this history, denoting by $k_1,\cdots, k_S$ the momentum variables, the wave vector~$q^i_j$ of particle~$j$ after the $i$-th collision can be constructed iteratively as follows:
\begin{enumerate}[---]
\item for $i=0$, we set
\begin{equation}\label{eq:def-qij-1}
q_j^0\,:=\,\left\{\begin{array}{lll}
-\sum_{j=1}^m k_j&:&j=0,\\
k_j&:&1\le j\le m,\\
0&:&m<j\le S,
\end{array}\right.
\end{equation}
\item for $i\ge1$, if $s_i=1$ (creation), we set
\begin{equation}\label{eq:def-qij-2}
q_j^i\,:=\,\left\{\begin{array}{lll}
 q^{i-1}_{a_i}-k_i&:&j=a_i,\\
 k_i&:&j=b_i,\\
 q^{i-1}_j&:&j\notin\{a_i,b_i\},
\end{array}\right.
\end{equation}
\item for $i\ge1$, if $s_i=-1$ (annihilation), we set
\begin{equation}\label{eq:def-qij-3}
q_j^i\,:=\,\left\{\begin{array}{lll}
q^{i-1}_{a_i}+q^{i-1}_{b_i}&:&j=a_i,\\
0&:&j=b_i,\\
q^{i-1}_j&:&j\notin\{a_i,b_i\},
\end{array}\right.
\end{equation}
\end{enumerate}
This encodes the momentum transfers in collisions described by the history $(s_i,a_i,b_i)_{1\le i\le n}$.
The following lemma provides a convenient description for the structure of these wave vectors.

\begin{lem}\label{lem:str-rep-qij}
There exist two sequences $(\sigma_{j,\ell})_{j,\ell}\subset\{1,-1\}$ and $\{c^i_{j,\ell}\}_{i,j,\ell}\subset\{0,1\}$, such that we have for all~$0\le i\le n$ and $0\le j\le S$, 
\[q^i_j= \sum_{\ell=1}^S\sigma_{j,\ell} c_{j,\ell}^i k_{\ell}. \]
\end{lem}

\begin{proof}
From the above construction of wave vectors, we can easily check iteratively the following two properties:
\begin{enumerate}[---]
\item For all $i,\ell$, the set $\{j: \partial_{k_\ell}q_j^i\neq 0 \}$ has either $0$ or $ 2$ elements. If it has two elements, say~$j,j'$, then necessarily $\partial_{k_\ell}q_j^i=-\partial_{k_\ell}q_{j'}^i\in\{\pm1\}$.
\item For all $i,\ell$, if $\{j: \partial_{k_\ell}q_j^{i-1}\neq 0 \}\ne\varnothing$ and $\{j: \partial_{k_\ell}q_j^i\neq 0 \}\ne\varnothing$, then their intersection is non-empty. In addition, denoting by $\{j,j'\}$ the elements of the first one and $\{j,j''\}$ the elements of the second one, we have $\partial_{k_\ell}q_j^{i-1}=\partial_{k_\ell}q_j^{i}$ and $\partial_{k_\ell}q_{j'}^{i-1}=\partial_{k_\ell}q_{j''}^{i}$.
\end{enumerate} 
The desired representation follows. 
\end{proof}

With the above notation, we now introduce an indicator $\varpi_i\in\{0,1\}$ that will indicate whether contour deformations can be applied to the propagator after the $i$-th collision. Informally, we set~\mbox{$\varpi_i=1$}, and we say that $i$ is a `good' index, if the $i$-th propagator either
\begin{enumerate}[(i)]
\item involves a velocity variable that is averaged over, or
\item involves a velocity variable that is not averaged over but is associated with a modified wave vector.
\end{enumerate}
In both cases, contour deformation can be used to obtain an $O(1)$ bound on the propagator. This is clear in case~(i). In case~(ii), this follows as in Step~3 of the proof of Lemma~\ref{lem:R1N-m02}: although the velocity variable is not directly averaged, it is integrated over when taking the $L^2$ norm, and the modification of the wave vector still enables a suitable deformation argument. A more precise definition of this indicator $\varpi_i$ is as follows.

\begin{defin}
For $1\le i< n$, we define
\[\varpi_i\,:=\,\left\{\begin{array}{lll}
1&:&\text{if $\omega_i\setminus \{0,\ldots,m\}\ne\varnothing$},\\
1&:&\text{if there is a couple $(j,\ell)$ with $1\le j\le m<\ell\le S$ and $c^i_{j,\ell} = 1$},\\
0&:&\text{otherwise},
\end{array}\right.\]
where we recall that the index set $\omega_i$ is defined in~\eqref{eq:def-omj}.
\end{defin}

To select suitable directions for contour deformations, we rely on the following elementary observation from linear algebra.

\begin{lem}\label{lem:algebre lineaire de sup}
Let $d\ge n-1$ and $u_1,\ldots u_n\in\R^d$ be affinely independent.
Let $\tilde{u}_1$ be the orthogonal projection of $u_1$ onto the vector space $\Span(u_2-u_1,\cdots,u_n-u_1)$, and define
\begin{equation*}
\mathfrak{o}(u_1,\cdots,u_n) \,:=\, \frac{u_1-\tilde{u}_1}{|u_1-\tilde{u}_1|},
\end{equation*}
which is a unit normal vector to the affine subspace $\aff(u_1,\ldots,u_n)$.
Then for all $1\le i\le n$,
\begin{equation*}
u_i\cdot\mathfrak{o}(u_1,\cdots,u_n) = \frac{n|\!\conv(0,u_1,\cdots,u_n)|_n}{|\!\conv(u_1,\cdots,u_n)|_{n-1}},
\end{equation*}
where $|\cdot|_j$ denotes the $j$-dimensional Hausdorff measure and $\conv$ the convex hull.
\end{lem}

\begin{proof}
The standard formula for the volume of a pyramid yields
\[|\!\conv(0,u_1,\ldots,u_n)|_n\,=\,\tfrac1n|\!\conv(u_1,\ldots,u_n)|_{n-1}\dist(0,\aff(u_1,\ldots,u_n)).\]
As $\dist(0,\aff(u_1,\ldots,u_n))=u_i\cdot \mathfrak o(u_1,\ldots,u_n)$, the claim follows.
\end{proof}

Using this lemma, complex deformations will allow us to bound propagators by inverse powers of simplex volumes $|\!\conv(0,k_{j_1},\ldots,k_{j_n})|_n$. The next integrability result shows precisely when such inverse powers are integrable --- this is the source of the dimension restriction appearing in our results.

\begin{lem}\label{lem:integ-conv}
Let $d\ge n-1$ and $s\in\R$. The function $(k_1,\cdots,k_n)\mapsto |\!\conv(0,k_1,\cdots,k_n)|_n^s$ is locally integrable on $(\R^d)^n$ if $s>n-1-d$.
\end{lem}

\begin{proof}
Write each $k_i$ as $k_i = k_i^\|+k_i^\perp$ where $k_i^\|$ is the orthogonal projection of $k_i$ onto $\Span(k_1,\cdots,k_{i-1})$. By Gram-Schmidt,
\[|\!\det(k_1,\ldots,k_n)|=|k_1^\perp|\cdots |k_n^\perp|,\]
which entails for the volume of the simplex,
\[|\!\conv(0,k_1,\cdots,k_n)|_n =\tfrac{1}{n!}|k_1^\perp|\cdots |k_n^\perp|.\]
For almost every $k_1,\ldots,k_{i-1}$, the vector $k_i^\perp$ ranges over a space of dimension $d-i+1$. We then find
\begin{equation*}
\int_{(\R^d)^n} \mathds1_{|k_1|,\ldots,|k_n|\leq R}~ |\!\conv(0,k_1,\cdots,k_n)|_n^s\,dk_1\cdots dk_n
\,\lesssim\,\int_{[0,R]^n} r_1^{s+d-1}\cdots r_n^{s+d-n} dr_1\ldots dr_n.
\end{equation*}
Integrability thus requires $s+d-i>-1$ for all $1\le i\le n$, and the conclusion follows.
\end{proof}

With these preparations, we can now establish the following general bound on the contribution~\eqref{eq:contr-history} of a given history in the Dyson expansion.

\begin{lem}\label{lem:abstr-bnd-1}
Let $(s_i,a_i,b_i)_{1\le i\le n}$ be a history associated to an abstract as in~\eqref{eq:prop-abstr} above, and let the space dimension satisfy $d\ge6n+m\vee n$.
Then we have for all $g\in C^\infty_c(\R^d)$,
\begin{equation}
\Big\|\hat{S}^{s_1}_{a_1,b_1}\WIGdot\ldots\WIGdot\hat{S}^{s_n}_{a_n,b_n}g\Big\|
\,\lesssim\,  N^{\frac n3+\frac13\sharp\{i:\varpi_i =0\}}\|\<\nabla_{v_0}\>^{n+1}g\|.
\end{equation}
\end{lem}

\begin{proof}
For $\omega\subset[0,S]$, we use the short-hand notation $q_\omega=(q_j)_{j\in\omega}\in (\R^d)^\omega$, and we denote by $\WIGdoT{q_\omega}$ the renormalized propagator on $\Ld^2((\mathbb{R}^d)^\omega)$ with wave vectors given by $(q_j)_{j\in\omega}$.
Given $(\hat z_1,\ldots, \hat z_m)\in\hat\Dd^m$, with the above construction~\eqref{eq:def-qij-1}--\eqref{eq:def-qij-3} of associated wave vectors $(q^i_j)_{i,j}$, the definition of collision operators yields
\begin{multline*}
\hat{S}^{s_1}_{a_1,b_1}\WIGdot\cdots\WIGdot\hat{S}^{s_n}_{a_n,b_n}g(\hat z_1,\ldots, \hat z_m)\\
\,=\,(-1)^m\sqrt{m!}\int_{\Dd^{S-m}} \prod_{i=1}^{n-1}\Big( \sqrt M(v_{b_i})\,\V(\mathfrak{q}_i)\mathfrak{q}_i\cdot\nabla_{v_{a_i}}\WIGdoT{q^i_{\omega_i}} \Big) \sqrt M(v_{b_n})\,\V(\mathfrak{q}_n)\mathfrak{q}_n\cdot\nabla_{v_{0}}g(v_0)\,d^*\hat z_{[m+1,S]},
\end{multline*}
where we recall that the index sets $\omega_i$'s are defined in~\eqref{eq:def-omj} and where we have set
\begin{equation*}
\mathfrak{q}_i  := \left\{\begin{array}{lll}
k_i&:&\text{if $s_i=1$},\\
q_{b_{i}}^{i-1}&:&\text{if $s_i=-1$}.
\end{array}\right.
\end{equation*}
In order to compute the squared norm, we double variables. For $\bar k_{m+1},\ldots,\bar k_S\in\R^d$, we consider the wave vectors $(\bar q_j^i)_{i,j}$ defined just like $(q_j^i)_{i,j}$ with $(k_1,\ldots,k_S)$ replaced by $(k_1,\ldots,k_m,\bar k_{m+1},\ldots,\bar k_S)$,
and we then define $(\bar{\mathfrak{q}}_i)_i$ accordingly.
For notational convenience, we set $\bar{v}_j=v_j$ and $\bar k_j=k_j$ for $1\le j\le m$.
We then find
\begin{multline}\label{eq:estim-hist-ex}
\Big\|\hat{S}^{s_1}_{a_1,b_1}\WIGdot\cdots\WIGdot\hat{S}^{s_n}_{a_n,b_n}g\Big\|^2
\,=\,m!\int d^* k_{[S]}d^* \bar{k}_{[m+1,S]} \prod_{i=1}^{n}\V(\mathfrak{q}_i){\V(\bar{\mathfrak{q}}_i)}\\
\times\bigg\<\int d\bar{v}_{[m+1,S]} \prod_{i=1}^{n-1}\left( \sqrt M(\bar{v}_{b_i})\,{\bar{\mathfrak{q}}}_i \cdot\nabla_{\bar{v}_{a_i}} \WIGdoT{\bar{q}^i_{\omega_i}} \right) \sqrt M(\bar{v}_{b_m})\,{\bar{\mathfrak{q}}}_m\cdot\nabla_{{v}_{0}}g(v_0),\\
 \int d{v}_{[m+1,S]}\prod_{i=1}^{n-1}\Big( \sqrt M(v_{b_i})\,{\mathfrak{q}}_i\cdot\nabla_{v_{a_i}}\WIGdoT{q^i_{\omega_i}} \Big) \sqrt M(v_{b_m})\,{\mathfrak{q}}_m\cdot\nabla_{v_{0}}g(v_0)\bigg\>_k.
\end{multline}
We start by bounding the scalar product for fixed momentum variables $k_{[S]},\bar{k}_{[m+1,S]}$. If we can choose unit vectors $\nu_1,\ldots,\nu_S,\bar{\nu}_{m+1},\cdots,\bar{\nu}_S\in\Sp^{d-1}$ such that for all $1\le i< n$,
\[\sum_{j=1}^S q_j^i\cdot\nu_j \leq 0 \quad\text{and}\quad\sum_{j= 1}^m \bar{q}_j^i\cdot\nu_j +\sum_{j= m+1}^S \bar{q}_j^i\cdot\bar{\nu}_j\geq 0,\]
then contour deformation yields
\begin{align}
&\bigg\<\int d\bar{v}_{[m+1,S]} \prod_{i=1}^{n-1}\left( \sqrt M(\bar{v}_{b_i})\,{\bar{\mathfrak{q}}}_i \cdot\nabla_{\bar{v}_{a_i}} \WIGdoT{\bar{q}^i_{\omega_i}} \right) \sqrt M(\bar{v}_{b_m})\,{\bar{\mathfrak{q}}}_m\cdot\nabla_{{v}_{0}}g(v_0),\label{eq:estim-hist-exdec}\\
&\hspace{1cm} \int d{v}_{[m+1,S]}\prod_{i=1}^{n-1}\Big( \sqrt M(v_{b_i})\,{\mathfrak{q}}_i\cdot\nabla_{v_{a_i}}\WIGdoT{q^i_{\omega_i}} \Big) \sqrt M(v_{b_m})\,{\mathfrak{q}}_m\cdot\nabla_{v_{0}}g(v_0)\bigg\>_k\nonumber\\
=~&\bigg\<\int d\bar{v}_{[m+1,S]} \prod_{i=1}^{n-1}\left( \sqrt M(\bar{v}_{b_i}-i\bar\nu_{b_i})\,{\bar{\mathfrak{q}}}_i \cdot\nabla_{\bar{v}_{a_i}}T_{-i\bar\nu_{\omega_i}}\Big[ \WIGdoT{\bar{q}^i_{\omega_i}}\Big] \right) \sqrt M(\bar{v}_{b_m}-i\bar\nu_{b_m})\,{\bar{\mathfrak{q}}}_m\cdot\nabla_{{v}_{0}}g(v_0),\nonumber\\
&\hspace{1cm} \int d{v}_{[m+1,S]}\prod_{i=1}^{n-1}\Big( \sqrt M(v_{b_i}+i\nu_{b_i})\,{\mathfrak{q}}_i\cdot\nabla_{v_{a_i}}T_{i\nu_{\omega_i}}\Big[\WIGdoT{q^i_{\omega_i}}\Big] \Big) \sqrt M(v_{b_m}+i\nu_{b_m})\,{\mathfrak{q}}_m\cdot\nabla_{v_{0}}g(v_0)\bigg\>_k,\nonumber
 \end{align}
where we have also set for notational convenience $\bar\nu_j ={\nu}_j$ for $1\le j\le m$.
Let us now make a suitable choice for these unit vectors $\nu_j,\bar\nu_j$ to perform the deformation. Denoting by~$p^\perp$ the orthogonal projection on $\mathrm{span}(k_1,\cdots,k_m)^\bot$, we define $\nu_j$ as follows:
\begin{enumerate}[---]
\item for all $1\le j\le m$,
\[\qquad\nu_j \,=\, \bar\nu_j \,:=\, \mathfrak{o}\big(-\sigma_{j,m+1}p^\perp(k_{m+1})\,,\,\ldots\,,\,-\sigma_{j,S}p^\perp(k_{S})\,,\,\sigma_{j,m+1}p^\perp(\bar{k}_{m+1})\,,\,\ldots\,,\,\sigma_{j,S}p^\perp(\bar{k}_{S})\big),\]
\item for all $m<j\le S$,
\begin{eqnarray*}
\nu_j  &:=& -\mathfrak{o}\big(\sigma_{j,1} k_1\,,\,\ldots\,,\,\sigma_{j,m}k_m\big),\\
\bar{\nu}_j &:=& \mathfrak{o}\big(\sigma_{j,1} \bar{k}_1\,,\,\ldots\,,\,\sigma_{j,m}\bar{k}_m\big),
\end{eqnarray*}
\end{enumerate}
where we recall that the $\sigma_{j,\ell}$'s are defined in the representation of wave vectors in Lemma~\ref{lem:str-rep-qij} and that the notation~$\mathfrak o(\cdot)$ is defined in Lemma~\ref{lem:algebre lineaire de sup}.
For this choice, let us now examine the resulting quantities~$\sum_jq_j^i\cdot\nu_j$ and $\sum_j\bar q_j^i\cdot\bar \nu_j$. We distinguish two cases, whether $i$ is a good index or not:
\begin{enumerate}[---]
\item {\it Case~1:} $1\le i<n$ with $\varpi_i=1$.\\
Using the representation of Lemma~\ref{lem:str-rep-qij} and the properties of $\mathfrak o(\cdot)$ in Lemma~\ref{lem:algebre lineaire de sup}, we find
{\footnotesize\begin{multline*}
\quad-\sum_{j=1}^S q_j^i \cdot\nu_j
\,=\,\sum_{j= 1}^m \sum_{\ell = 1}^S c^i_{j,\ell}\Big(-\sigma_{j,\ell}p^\perp(k_\ell)\cdot\nu_j\Big)+\sum_{j= m+1}^S \sum_{\ell = 1}^S c^i_{j,\ell}\Big(-\sigma_{j,\ell}k_\ell\cdot\nu_j\Big)\\
\,\simeq\,\sum_{j=1}^m\sum_{\ell=m+1}^Sc^i_{j,\ell}\frac{|\!\conv(0,-\sigma_{j,m+1}p^\perp(k_{m+1})\,,\,\ldots\,,\,-\sigma_{j,S}p^\perp(k_{S})\,,\,\sigma_{j,m+1}p^\perp(\bar{k}_{m+1})\,,\,\ldots\,,\,\sigma_{j,S}p^\perp(\bar{k}_{S}))|_{2(S-m)}}{|\!\conv(-\sigma_{j,m+1}p^\perp(k_{m+1})\,,\,\ldots\,,\,-\sigma_{j,S}p^\perp(k_{S})\,,\,\sigma_{j,m+1}p^\perp(\bar{k}_{m+1})\,,\,\ldots\,,\,\sigma_{j,S}p^\perp(\bar{k}_{S}))|_{2(S-m)-1}}\\
+\sum_{j= m+1}^S \sum_{\ell = 1}^S c^i_{j,\ell}\frac{|\!\conv(0,\sigma_{j,1} k_1\,,\,\ldots\,,\,\sigma_{j,m}k_m)|_{m}}{|\!\conv(\sigma_{j,1} k_1\,,\,\ldots\,,\,\sigma_{j,m}k_m)|_{m-1}}.
\end{multline*}
}By definition of $\varpi_i$, we note that the condition $\varpi_i=1$ implies that there is some couple $(j,\ell)$ with $1\le j\le m<\ell\le S$ or $j>m$ such that~$c_{j,\ell}^i=1$. Using $|\!\conv(u_1,\ldots,u_n)|_{n-1}\lesssim\langle(u_1,\ldots,u_n)\rangle^{n-1}$, the above then yields the lower bound
{\footnotesize\begin{multline*}
\quad-\sum_{j=1}^S q_j^i \cdot\nu_j
\,\gtrsim\,\min_j\frac{|\!\conv(0,-\sigma_{j,m+1}p^\perp(k_{m+1})\,,\,\ldots\,,\,-\sigma_{j,S}p^\perp(k_{S})\,,\,\sigma_{j,m+1}p^\perp(\bar{k}_{m+1})\,,\,\ldots\,,\,\sigma_{j,S}p^\perp(\bar{k}_{S}))|_{2(S-m)}}{\langle (k_{[m+1,S]},\bar{k}_{[m+1,S]})\rangle^{2(S-m)-1}}\\[-1mm]
\wedge\min_j\frac{|\!\conv(0,\sigma_{j,1} k_1\,,\,\ldots\,,\,\sigma_{j,m}k_m)|_{m}}{\langle k_{[m]}\rangle^{m-1}}.
\end{multline*}
}
\smallskip\item {\it Case~2:} $1\le i<n$ with $\varpi_i=0$.\\
Choosing $j\in\omega_i$, using Lemmas~\ref{lem:str-rep-qij} and~\ref{lem:algebre lineaire de sup} as above, we can bound
\[|q_{\omega_i}^i|\,\ge\, |q_j^i|\,\ge\,q_j^i\cdot \mathfrak{o}(\sigma_{j,1}k_1,\cdots,\sigma_{j,S}k_S)\,\gtrsim\,\frac{|\!\conv(0,k_1,\cdots,k_S)|_S}{\<k_{[S]}\>^{S-1}}.\]
\end{enumerate}
Combining both cases, we obtain
\begin{equation}\label{eq:estim-case12-inv}
\ind_{\varpi_i = 0}\tfrac{1}{|q^i_{\omega_i}|}+\ind_{\varpi_i=1}\Big(1+\big(-\textstyle{\sum_{j= 0}^S} q_j^i \cdot\nu_j\big)^{-1}\Big)\,\lesssim\,\omega(k_{[S]},\bar{k}_{[S]}),
\end{equation}
where we have set
{\small\begin{multline*}
\omega(k_{[S]},\bar{k}_{[S]}) :=\<(k_{[S]},\bar k_{[S]})\>^{n\vee m-1}\\
\times\bigg(1
+\sum_j|\!\conv(0,-\sigma_{j,m+1}p^\perp(k_{m+1})\,,\,\ldots\,,\,-\sigma_{j,S}p^\perp(k_{S})\,,\,\sigma_{j,m+1}p^\perp(\bar{k}_{m+1})\,,\,\ldots\,,\,\sigma_{j,S}p^\perp(\bar{k}_{S}))|_{2(S-m)}^{-1}\\[-3mm]
+\sum_j|\!\conv(0,\sigma_{j,1} k_1\,,\,\ldots\,,\,\sigma_{j,m}k_m)|_{m}^{-1}
+|\!\conv(0,k_1,\cdots,k_S)|_S^{-1}+|\!\conv(0,\bar k_1,\cdots,\bar k_S)|_S^{-1}\bigg).
\end{multline*}
}In order to estimate the norm 
\[\bigg\|\int d{v}_{[m+1,S]} \prod_{i=1}^{n-1}\bigg( \sqrt M(v_{b_i}+i\nu_{b_i})\,\mathfrak{q}_i\cdot\nabla_{v_{a_i}}T_{i{\nu}_{\omega_i}}\WIGdoT{q^i_{\omega_i}} \bigg) \sqrt M(v_{b_m}+i\nu_{b_m})\,\mathfrak{q}_m\cdot\nabla_{v_{a_m}}g\bigg\|_k,\]
we apply iteratively Propositions~\ref{prop:box-anal} and~\ref{prop:pde} to bound the (deformed) renormalized propagators: for the $i$-th resolvent in the product, we apply Proposition~\ref{prop:box-anal} if $\varpi_i=1$ and Proposition~\ref{prop:pde} if $\varpi_i=0$. In doing so, we need to take into account all possible ways to distribute the velocity gradients in the estimates. When considering $\ell$ derivatives of a propagator, we note that under Proposition~\ref{prop:box-anal} the bound preserves the maximal number $\ell$ of derivatives, while under Proposition~\ref{prop:pde} it requires $(\ell-1)\vee2$ derivatives instead. In order to capture all possible admissible patterns of derivative counts, we thus consider the set~$\mathfrak L_n$ of all finite sequences $(\ell_1,\ldots, \ell_n)$ such that
\begin{enumerate}[---]
\item $\ell_1=1$,
\item if $\varpi_i =1$, then $\ell_{i+1} \in [1,\ell_{i}+1]$,
\item if $\varpi_i =0$, then $ \ell_{i+1} \in [1,\ell_{i}\vee3]$.
\end{enumerate}
In these terms, we can bound
\begin{multline*}
\bigg\|\int d{v}_{[m+1,S]}\prod_{i=1}^{n-1}\bigg( \sqrt M(v_{b_i}+i\nu_{b_i})\,\mathfrak{q}_i\cdot\nabla_{v_{a_i}}T_{i{\nu}_{\omega_i}}\WIGdoT{q^i_{\omega_i}} \bigg) \sqrt M(v_{b_m}+i\nu_{b_m})\,\mathfrak{q}_m\cdot\nabla_{v_{a_m}}g\bigg\|_k\\
\hspace{-3cm}\,\lesssim\,\<k_{[S]}\>^n\sum_{(\ell_i)_i\in \mathfrak L_n} \|\<\nabla_{v_0}\>^{\ell_n}g\|\prod_{i=1}^{n-1}\bigg(\ind_{\varpi_i=0}N^{\frac{2+\ell_{i}-\ell_{i+1}}3}\<k_{[S]}\>^{\ell_i+4}|q_{\omega_i}^i|^{-2}\\[-3mm]
+\ind_{\varpi_i=1}\<k_{[S]}\>^{1+\ell_i-\ell_{i+1}}\Big(1+\big(-\textstyle{\sum_{j= 0}^S} q_j^i \cdot\nu_j\big)^{-1}\Big)^{2+\ell_i-\ell_{i+1}}\bigg).
\end{multline*}
By a careful summation argument, note that for $(\ell_i)_i\in\mathfrak L_n$ we always have
\[\sum_{i=1}^{n-1}2\vee(2+\ell_i-\ell_{i+1})\,\le\,3n,\qquad
\sum_{i=1}^{n-1}(2+\ell_i-\ell_{i+1})\mathds1_{\varpi_i=0}\le n+\sharp\{i:\omega_i=0\}.\]
Inserting this into the above, and using~\eqref{eq:estim-case12-inv}, we obtain
\begin{multline}\label{eq:estimation d un truc qua pas de nom}
\bigg\|\int d{v}_{[m+1,S]}\prod_{i=1}^{n-1}\bigg( \sqrt M(v_{b_i}+i\nu_{b_i})\,\mathfrak{q}_i\cdot\nabla_{v_{a_i}}T_{i{\nu}_{\omega_i}}\WIGdoT{q^i_{\omega_i}} \bigg) \sqrt M(v_{b_m}+i\nu_{b_m})\,\mathfrak{q}_m\cdot\nabla_{v_{a_m}}g\bigg\|_k\\
\,\lesssim\,N^{\frac n3+\frac13\sharp\{i:\omega_i=0\}}\|\<\nabla_{v_0}\>^{n+1}g\|\<k_{[S]}\>^{Cn^2}\omega(k_{[S]},\bar k_{[S]})^{3n}.
\end{multline}
Combining this with~\eqref{eq:estim-hist-ex} and~\eqref{eq:estim-hist-exdec}, we obtain
\begin{multline*}
\Big\|\hat{S}^{s_1}_{a_1,b_1}\WIGdot\cdots\WIGdot\hat{S}^{s_n}_{a_n,b_n}g\Big\|^2
\,\lesssim\,N^{\frac{2n}3+\frac23\sharp\{i:\omega_i=0\}}\|\<\nabla_{v_0}\>^{n+1}g\|^2\\[-1mm]
\times\int \<k_{[S]}\>^{Cn^2}\<\bar k_{[S]}\>^{Cn^2}\Big(\prod_{i=1}^{n}|\V(\mathfrak{q}_i)||\V(\bar{\mathfrak{q}}_i)|\Big)\omega(k_{[S]},\bar k_{[S]})^{6n}\,d^* k_{[S]}d^* \bar{k}_{[m+1,S]}.
\end{multline*}
It remains to check the integrability in $k$: as the subspace $\mathrm{span}(k_1,\ldots,k_m)^\perp$ has dimension $\ge d-m$, it follows from Lemma~\ref{lem:integ-conv} that the map $\omega(k_{[n]},\bar{k}_{[k]})^{6n}$ is locally summable if
\[\left\{\begin{array}{l}
-6n>2(S-m)-1-(d-m),\\
-6n>\max(S,m)-1-d.
\end{array}\right.\]
This holds whenever $d\ge6n+m\vee n$, and the conclusion follows.
\end{proof}

To apply the above estimate, we must control the number of bad indices~$i$ with $\varpi_i=0$. This counting problem is subtle, especially because it depends on the chosen history. To streamline the analysis, we introduce a special class of subsequences of the abstract, which we call \emph{tents}.

\begin{defin}
Given an abstract $(s_1,\ldots,s_n)$, a contiguous sub-sequence $(s_i)_{\alpha\le i\le\beta}$ is called a {\it tent} if it satisfies the following recursive rules:
\begin{enumerate}[---]
\item \emph{Base cases:} If $(s_i)_{\alpha\le i\le\beta}=(1,-1,-1)$, then it is a tent.
If $\sum_{i=\alpha}^\beta s_i=0$ and $\sum_{i=\alpha}^j s_i>0$ for all~$\alpha\le j<\beta$, then it is also a tent.
\smallskip\item \emph{Recursive step:} If $(s_i)_{\alpha\le i\le\beta}$ starts with $(s_\alpha,s_{\alpha+1})= (1,-1)$ and if the remainder $(s_{\alpha+2},\cdots,s_\beta)$ is already a tent, then $(s_i)_{\alpha\le i\le\beta}$ is also a tent.
\smallskip\item \emph{Exclusion:} If $(s_i)_{\alpha\le i\le\beta}=(1,-1)$, then it is not a tent.
\end{enumerate}
\end{defin}

With this definition, the key observation is that all indices contained inside a tent are guaranteed to be good, regardless of the associated history. While not every good index belongs to a tent, this will allow us to reduce the counting problem to a much more tractable combinatorial question.

\begin{lem}\label{lem:carac-tent-varpi}
Let $(s_1,\ldots,s_n)$ be an abstract and let $(s_i)_{\alpha\le i\le\beta}$ be a tent. Then, for any associated history, we have $\varpi_i = 1$ for all $\alpha\le i<\beta$
\end{lem}

\begin{proof}
We treat the three defining cases.
\begin{enumerate}[---]
\item\emph{Base case~1:} assume $(s_i)_{\alpha\le i\le\beta}=(1,-1,-1)$.\\
We already have $\varpi_{\alpha}=1$ and it remains to show $\varpi_{\alpha+1}=1$. The subsequence $(1,-1)$ corresponds to one of the following possible sub-histories, represented diagrammatically,
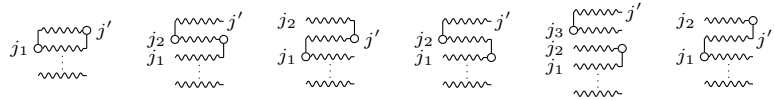
\[{\tiny\begin{tikzpicture}[baseline={([yshift=-1ex]current bounding box.center)},scale=0.8]
\begin{scope}[every node/.style={circle,draw,fill=white,inner sep=0pt,minimum size=3pt}]
\node (1) at (0,0.6) {};
\node (2) at (0.8,0.9) {};
\end{scope}
\node (a) at (-0.3,0.6) {$j_1$};
\node (b) at (1.1,0.9) {$j'$};
\draw (1) -- (0,0.9);
\draw (2) -- (0.8,0.6);
\draw [decorate,decoration={snake,amplitude=1pt,segment length=3pt}] (1) -- (0.8,0.6);
\draw [decorate,decoration={snake,amplitude=1pt,segment length=3pt}] (2) -- (0,0.9);
\draw [decorate,decoration={snake,amplitude=1pt,segment length=3pt}] (0,0.15) -- (0.8,0.15);
\draw[dotted] (0.4,0.25) -- (0.4,0.55);
\end{tikzpicture}}
~~{\tiny\begin{tikzpicture}[baseline={([yshift=-1ex]current bounding box.center)},scale=0.8]
\begin{scope}[every node/.style={circle,draw,fill=white,inner sep=0pt,minimum size=3pt}]
\node (1) at (0,0.9) {};
\node (2) at (0.8,0.9) {};
\end{scope}
\node (a) at (-0.3,0.6) {$j_1$};
\node (b) at (-0.3,0.9) {$j_2$};
\node (c) at (1,1.2) {$j'$};
\draw (1) -- (0,1.2);
\draw (2) -- (0.8,0.6);
\draw [decorate,decoration={snake,amplitude=1pt,segment length=3pt}] (1) -- (2);
\draw [decorate,decoration={snake,amplitude=1pt,segment length=3pt}] (0,0.6) -- (0.8,0.6);
\draw [decorate,decoration={snake,amplitude=1pt,segment length=3pt}] (0,1.2) -- (0.8,1.2);
\draw [decorate,decoration={snake,amplitude=1pt,segment length=3pt}] (0,0.15) -- (0.8,0.15);
\draw[dotted] (0.4,0.25) -- (0.4,0.55);
\end{tikzpicture}}
~~{\tiny\begin{tikzpicture}[baseline={([yshift=-1ex]current bounding box.center)},scale=0.8]
\begin{scope}[every node/.style={circle,draw,fill=white,inner sep=0pt,minimum size=3pt}]
\node (1) at (0,0.6) {};
\node (2) at (0.8,0.9) {};
\end{scope}
\node (a) at (-0.3,0.6) {$j_1$};
\node (b) at (-0.3,1.2) {$j_2$};
\node (c) at (1.1,0.9) {$j'$};
\draw (1) -- (0,0.9);
\draw (2) -- (0.8,1.2);
\draw [decorate,decoration={snake,amplitude=1pt,segment length=3pt}] (1) -- (0.8,0.6);
\draw [decorate,decoration={snake,amplitude=1pt,segment length=3pt}] (0,0.9) -- (2);
\draw [decorate,decoration={snake,amplitude=1pt,segment length=3pt}] (0,1.2) -- (0.8,1.2);
\draw [decorate,decoration={snake,amplitude=1pt,segment length=3pt}] (0,0.15) -- (0.8,0.15);
\draw[dotted] (0.4,0.25) -- (0.4,0.55);
\end{tikzpicture}}
~~{\tiny\begin{tikzpicture}[baseline={([yshift=-1ex]current bounding box.center)},scale=0.8]
\begin{scope}[every node/.style={circle,draw,fill=white,inner sep=0pt,minimum size=3pt}]
\node (1) at (0,0.9) {};
\node (2) at (0.8,0.6) {};
\end{scope}
\node (a) at (-0.3,0.6) {$j_1$};
\node (b) at (-0.3,0.9) {$j_2$};
\node (c) at (1,1.2) {$j'$};
\draw (1) -- (0,1.2);
\draw (2) -- (0.8,0.9);
\draw [decorate,decoration={snake,amplitude=1pt,segment length=3pt}] (1) -- (0.8,0.9);
\draw [decorate,decoration={snake,amplitude=1pt,segment length=3pt}] (0,0.6) -- (2);
\draw [decorate,decoration={snake,amplitude=1pt,segment length=3pt}] (0,1.2) -- (0.8,1.2);
\draw [decorate,decoration={snake,amplitude=1pt,segment length=3pt}] (0,0.15) -- (0.8,0.15);
\draw[dotted] (0.4,0.25) -- (0.4,0.55);
\end{tikzpicture}}
~~{\tiny\begin{tikzpicture}[baseline={([yshift=-1ex]current bounding box.center)},scale=0.8]
\begin{scope}[every node/.style={circle,draw,fill=white,inner sep=0pt,minimum size=3pt}]
\node (1) at (0,1.2) {};
\node (2) at (0.8,0.9) {};
\end{scope}
\node (a) at (-0.3,0.6) {$j_1$};
\node (b) at (-0.3,0.9) {$j_2$};
\node (c) at (-0.3,1.2) {$j_3$};
\node (d) at (1,1.5) {$j'$};
\draw (1) -- (0,1.5);
\draw (2) -- (0.8,0.6);
\draw [decorate,decoration={snake,amplitude=1pt,segment length=3pt}] (1) -- (0.8,1.2);
\draw [decorate,decoration={snake,amplitude=1pt,segment length=3pt}] (0,1.5) -- (0.8,1.5);
\draw [decorate,decoration={snake,amplitude=1pt,segment length=3pt}] (0,0.6) -- (0.8,0.6);
\draw [decorate,decoration={snake,amplitude=1pt,segment length=3pt}] (0,0.9) -- (2);
\draw [decorate,decoration={snake,amplitude=1pt,segment length=3pt}] (0,0.15) -- (0.8,0.15);
\draw[dotted] (0.4,0.25) -- (0.4,0.55);
\end{tikzpicture}}
~~{\tiny\begin{tikzpicture}[baseline={([yshift=-1ex]current bounding box.center)},scale=0.8]
\begin{scope}[every node/.style={circle,draw,fill=white,inner sep=0pt,minimum size=3pt}]
\node (1) at (0,0.6) {};
\node (2) at (0.8,1.2) {};
\end{scope}
\node (a) at (-0.3,0.6) {$j_1$};
\node (b) at (-0.3,1.2) {$j_2$};
\node (c) at (1,0.9) {$j'$};
\draw (1) -- (0,0.9);
\draw (2) -- (0.8,0.9);
\draw [decorate,decoration={snake,amplitude=1pt,segment length=3pt}]  (1) -- (0.8,0.6);
\draw [decorate,decoration={snake,amplitude=1pt,segment length=3pt}]  (0,1.2) -- (2);
\draw [decorate,decoration={snake,amplitude=1pt,segment length=3pt}]  (0,0.9) -- (0.8,0.9);
\draw [decorate,decoration={snake,amplitude=1pt,segment length=3pt}] (0,0.15) -- (0.8,0.15);
\draw[dotted] (0.4,0.25) -- (0.4,0.55);
\end{tikzpicture}}\]
where the labels indicate the velocity variables.
In the five first diagrams, we find $j'\in\omega_\alpha\cap\omega_{\alpha+1}$ and $j'\notin\omega_{\alpha-1}$, hence $\varpi_{\alpha+1}=1$. In the sixth diagram, the wave vector of particle~$j_1$ after the $(\alpha+1)$-th collision is $q^{\alpha+1}_{j_1}=q^{\alpha}_{j_1}-k_{j'}$, hence again $\varpi_{\alpha+1}=1$

\smallskip\item\emph{Base case~2:} assume that $\sum_{i=\alpha}^\beta s_i=0$ and $\sum_{i=\alpha}^j s_i>0$ for all~$\alpha\le j<\beta$.\\
Recalling that $\sharp\omega_i=1+n+\sum_{j=1}^{i} s_j$, we deduce $\sharp\omega_i>\sharp\omega_{\alpha-1}$ for all $\alpha\le i<\beta$. Thus, for each such $i$, there exists $j_i\in \omega_i\setminus\omega_{\alpha-1}$, necessarily with $j_i>m$, and hence $\varpi_i=1$.

\smallskip\item\emph{Recursive step:}
This can be treated as Base case 1.
\end{enumerate}
This concludes the proof.
\end{proof}

With the above construction of tents, we are now in position to prove an easy upper bound on the number of `bad' indices $i$ with $\varpi_i=0$.

\begin{lem}\label{lem:count-om0}
Let $(s_1,\ldots,s_n)$ be an abstract with $s_1 = 1$ and $m+\sum_{i=1}^n s_i = 0$. Then, for any associated history, we have
\begin{equation}
\sharp\{i:\varpi_i = 0\}\,\le\,\tfrac14n+\tfrac34m.
\end{equation}
\end{lem}

\begin{proof}
Let $(s_1,\ldots,s_n)$ be an abstract as in the statement.
By construction, we can decompose it uniquely into tents separated by down steps. That is, there exist $r\ge1$ and a double sequence
\[1=\alpha_1<\beta_1<\ldots<\alpha_r<\beta_r<\alpha_{r+1}=n+1,\]
such that:
\begin{enumerate}[---]
\item each block $(s_i)_{\alpha_u\le i\le\beta_u}$ is a tent;
\item between tents, we only have down steps: $s_i=-1$ for $\beta_u\le i<\alpha_{u+1}$.
\end{enumerate}
By definition, every tent has either
\[\sum_{i=\alpha_u}^{\beta_u}s_i=0~~\text{(type $0$)},\quad\text{or}\quad\sum_{i=\alpha_u}^{\beta_u}s_i=-1~~\text{(type $1$)}.\]
Let $T_0$ and $T_1$ denote the number of tents of type $0$ and type $1$, respectively, and let 
\[D=\sum_{u=1}^r(\alpha_{u+1}-\beta_u-1)\]
be the total number of down steps outside tents. Then
\[m\,=\,T_1+D.\]
By Lemma~\ref{lem:carac-tent-varpi}, the only bad indices lie either at the last index of a tent or in the down steps, whence
\[\sharp\{i:\varpi_i=0\}=T_0+T_1+D.\]
Moreover, we note that tents of type $0$ have length $\ge4$ and that tents of type $1$ have length $\ge3$, so
\[3T_1+4T_0+D\le n.\]
Combining these results, we can bound
\[\sharp\{i:\varpi_i=0\}\,=\,T_0+T_1+D\,\le\,\tfrac14n+\tfrac14T_1+\tfrac34D\,\le\,\tfrac14n+\tfrac34(T_1+D)\,\le\,\tfrac14n+\tfrac34m,\]
which proves the claim.
\end{proof}

We can now conclude the proof of Proposition~\ref{prop:estimation erreur n>1}: combining the bound of Lemma~\ref{lem:abstr-bnd-1} on the contribution of each history, together with the bound of Lemma~\ref{lem:count-om0} on the number of bad indices, the conclusion of Proposition~\ref{prop:estimation erreur n>1} follows.\qed

\subsection{Proof of Theorem~\ref{th:main}}
Recall that we use the ansatz~\eqref{eq:def-tilgm} for a suitably chosen admissible set~$\Omega$ of abstracts. We define $\Omega$ as the smallest admissible set containing all abstracts of length non-bigger than
\[K:=3m_0+9.\]
By definition of admissible sets of abstracts, this means that $\Omega$ consists of all abstracts of the form
\[(\underbrace{-1,\ldots,-1}_{\ell~\text{times}},s_1,\ldots,s_n)\]
where $s_i\in\{\pm1\}$, $0\le n\le K$, $0\le m\le m_0$, and $0\le\ell\le m_0-m$, such that
\[m+\sum_{i=1}^ns_i=0,\qquad 0<m+\sum_{i=1}^js_i\le m_0\quad\text{for all $1\le j<n$.}\]
In particular, any abstract in $\Omega$ has length $\le K+m_0$, and any abstract in $\partial\Omega$ has length in the interval $[K+1,K+m_0+1]$.
We now estimate the remainder terms appearing in the approximate hierarchy satisfied by this ansatz. By Lemma~\ref{lem:approx-hier}, these terms can be written as sums over $\partial\Omega$. Applying Proposition~\ref{prop:estimation erreur n>1}, we obtain for all $1\le m\le m_0$,
\begin{eqnarray*}
\|\Lc R_m^{N,m_0}\|
&\lesssim&\frac{t_N}{N^{3/2}}\mathds1_{m<m_0}\sum_{n\ge3}N^{-\frac12(n-3)}\sum_{(s_1,\ldots,s_n)\in\partial\Omega_m}\|\Ic_{(s_1,\ldots,s_n)}\Lc\tilde g_0^N\|\\
&\lesssim&\frac{t_N}{N^{3/2}}\mathds1_{m<m_0}\sum_{n\ge3}\mathds1_{\exists (s_1,\ldots,s_n)\in\partial\Omega}N^{-\frac12(n-3)}N^{\frac5{12}n+\frac14m}\|\<\nabla_{v_0}\>^{n+1}\Lc\tilde g_0^N\|\\
&\lesssim&t_NN^{\frac14(m_0-1)-\frac1{12}(K+1)}\|\<\nabla_{v_0}\>^{K+m_0+2}\Lc\tilde g_0^N\|,
\end{eqnarray*}
provided that the space dimension satisfies $d\ge 7(K+m_0+1)=28m_0+70$. By the choice of~$K$, this means for all $1\le m\le m_0$,
\[\|\Lc R_m^{N,m_0}\|\,\lesssim\,t_NN^{-1-\frac{1}{12}}\|\<\nabla_{v_0}\>^{4m_0+11}\Lc\tilde g_0^N\|,\]
and thus, by Lemma~\ref{lem:estim-g0N}, for $t_N=N$,
\begin{equation*}
\3\Lc R_m^{N,m_0}\3
\,\lesssim\,N^{-\frac1{12}}\|\<\nabla_{v_0}\>^{4m_0+10}\mathfrak g\|.
\end{equation*}
For $m=0$, recall from Lemma~\ref{lem:approx-hier} that $R_0^{N,m_0}$ splits into two contributions. The first is a sum over~$\partial\Omega$, estimated exactly as above. The second term is
\begin{equation*}
-\frac{t_N}{N^2}\mathds1_{m_0\ge2}\Big({\tiny\begin{tikzpicture}[baseline={([yshift=-1ex]current bounding box.center)},scale=0.8]
\begin{scope}[every node/.style={circle,draw,fill=white,inner sep=0pt,minimum size=3pt}]
\node (1) at (0,0) {};
\node (2) at (1.5,0) {};
\node (3) at (0.5,0.33) {};
\node (4) at (1,0.33) {};
\end{scope}
\draw (0.5,0) -- (1) -- (0,0.33) -- (3) -- (0.5,0.66);
\draw (1) -- (0,0.33);
\draw (2) -- (1.5,0.33);
\draw (4) -- (1,0.66);
\draw [decorate,decoration={snake,amplitude=1pt,segment length=3pt}] (3) -- (4) -- (1.5,0.33);
\draw [decorate,decoration={snake,amplitude=1pt,segment length=3pt}] (0.5,0) -- (2);
\draw [decorate,decoration={snake,amplitude=1pt,segment length=3pt}] (0.5,0.66) -- (1,0.66);
\end{tikzpicture}}
+
{\tiny\begin{tikzpicture}[baseline={([yshift=-1ex]current bounding box.center)},scale=0.8]
\begin{scope}[every node/.style={circle,draw,fill=white,inner sep=0pt,minimum size=3pt}]
\node (1) at (0,0) {};
\node (2) at (0.5,0) {};
\node (3) at (1,0) {};
\node (4) at (1.5,0) {};
\end{scope}
\draw (0.5,0.33) -- (2) -- (1) -- (0,0.66) -- (0.5,0.66);
\draw (3) -- (1,0.33);
\draw (1.5,0.66) -- (4);
\draw [decorate,decoration={snake,amplitude=1pt,segment length=3pt}] (2) -- (3) -- (4);
\draw [decorate,decoration={snake,amplitude=1pt,segment length=3pt}] (0.5,0.33) -- (1,0.33);
\draw [decorate,decoration={snake,amplitude=1pt,segment length=3pt}] (0.5,0.66) -- (1.5,0.66);
\end{tikzpicture}}
\Big)\Lc\tilde g_0^N,
\end{equation*}
which can be bounded using contour deformation together with Proposition~\ref{prop:box-anal} similarly as in Section~\ref{sec:case-m02}, and we omit the details. We conclude for all $0\le m\le m_0$, for $t_N=N$,
\begin{equation}
\3\Lc R_m^{N,m_0}\3\,\lesssim\, N^{-\frac1{12}}\|\<\nabla_{v_0}\>^{4m_0+10} \mathfrak g\|.
\end{equation}
Combining these bounds on the remainder terms and applying Lemma~\ref{lem:err-est0} yields
\[\3\Lc g_0^{N,m_0}-\Lc\tilde g_0^{N}\3\,\lesssim\, N^{-\frac1{12}}\|\<\nabla_{v_0}\>^{4m_0+10} \mathfrak g\|.\]
Together with Lemma~\ref{lem:kin-lim}, this concludes the proof.
\qed

\section*{Acknowledgements}
The authors acknowledge financial support from the European Union (ERC, PASTIS, Grant Agreement n$^\circ$101075879).\footnote{{Views and opinions expressed are however those of the authors only and do not necessarily reflect those of the European Union or the European Research Council Executive Agency. Neither the European Union nor the granting authority can be held responsible for them.}}

\bibliographystyle{plain}
\bibliography{biblio}

\end{document}